\documentclass[a4paper,USenglish,cleveref, autoref, thm-restate]{oasics-v2021}

\title{Fast and Delay-Robust Multimodal Journey Planning}

\newcommand{\affiliationKIT}{Karlsruhe Institute of Technology (KIT), Karlsruhe, Germany}

\author{Dominik Bez}{\affiliationKIT}{dominik.bez@student.kit.edu}{}{}

\author{Jonas Sauer}{\affiliationKIT}{jonas.sauer2@kit.edu}{https://orcid.org/0000-0002-7196-7468}{}

\authorrunning{D. Bez and J. Sauer}

\Copyright{Dominik Bez and Jonas Sauer}



\supplement{Source code and data are available at~\url{https://github.com/kit-algo/ULTRA}.}

\funding{This research was funded by the DFG under grant number WA 654/23-2.}

\nolinenumbers
\hideOASIcs

\usepackage{niceAlgorithm}
\usepackage{booktabs}
\usepackage{colortbl}
\usepackage{mathtools}

\newcommand{\atime}{\ensuremath{\tau}\xspace}
\newcommand{\departure}[1]{\ensuremath{#1_\text{dep}}\xspace}
\newcommand{\arrival}[1]{\ensuremath{#1_\text{arr}}\xspace}
\newcommand{\departureTime}{\ensuremath{\departure{\atime}}\xspace}
\newcommand{\arrivalTime}{\ensuremath{\arrival{\atime}}\xspace}
\newcommand{\transferTime}{\ensuremath{\atime_{\textsf{tra}}}\xspace}

\newcommand{\minTime}{\ensuremath{\atime_{\textsf{min}}}\xspace}

\newcommand{\graph}{\ensuremath{G}\xspace}
\newcommand{\vertices}{\ensuremath{V}\xspace}
\newcommand{\aVertex}{\ensuremath{v}\xspace}

\newcommand{\bVertex}{\ensuremath{w}\xspace}

\newcommand{\edges}{\ensuremath{E}\xspace}
\newcommand{\edge}{\ensuremath{e}\xspace}

\newcommand{\aPath}{\ensuremath{P}\xspace}

\newcommand{\stops}{\ensuremath{\mathcal{S}}\xspace}
\newcommand{\astop}{\aVertex}

\newcommand{\stopEvents}{\ensuremath{\mathcal{E}}\xspace}
\newcommand{\stopEvent}{\ensuremath{\varepsilon}\xspace}
\newcommand{\precedingEvents}{\ensuremath{\overleftarrow{\stopEvents}}\xspace}
\newcommand{\succeedingEvents}{\ensuremath{\overrightarrow{\stopEvents}}\xspace}
\newcommand{\trips}{\ensuremath{\mathcal{T}}\xspace}
\newcommand{\aTrip}{\ensuremath{T}\xspace}
\newcommand{\aTripA}{\ensuremath{\aTrip_a}\xspace}
\newcommand{\aTripB}{\ensuremath{\aTrip_b}\xspace}

\newcommand{\aTripSegment}[2]{\ensuremath{\aTrip[#1,#2]}\xspace}
\newcommand{\aNamedTripSegment}[3]{\ensuremath{\aTrip_{#1}[#2,#3]}\xspace}

\newcommand{\routes}{\ensuremath{\mathcal{R}}\xspace}
\newcommand{\aRoute}{\ensuremath{R}\xspace}

\newcommand{\sourceIndex}{\ensuremath{\textsf{s}}\xspace}
\newcommand{\originIndex}{\ensuremath{\textsf{o}}\xspace}
\newcommand{\destinationIndex}{\ensuremath{\textsf{d}}\xspace}
\newcommand{\targetIndex}{\ensuremath{\textsf{t}}\xspace}
\newcommand{\sourceVertex}{\ensuremath{\aVertex_\sourceIndex}\xspace}

\newcommand{\originVertex}{\ensuremath{\aVertex_\originIndex}\xspace}
\newcommand{\destinationVertex}{\ensuremath{\aVertex_\destinationIndex}\xspace}
\newcommand{\targetVertex}{\ensuremath{\aVertex_\targetIndex}\xspace}

\newcommand{\sourceEvent}{\ensuremath{\stopEvent_\sourceIndex}\xspace}
\newcommand{\originEvent}{\ensuremath{\stopEvent_\originIndex}\xspace}
\newcommand{\destinationEvent}{\ensuremath{\stopEvent_\destinationIndex}\xspace}
\newcommand{\targetEvent}{\ensuremath{\stopEvent_\targetIndex}\xspace}

\newcommand{\aJourney}{\ensuremath{J}\xspace}
\newcommand{\journeys}{\ensuremath{\mathcal{J}}\xspace}

\newcommand{\prefix}[1]{\ensuremath{#1_\textsf{p}}\xspace}
\newcommand{\aCandidateJourney}{\ensuremath{\aJourney^{\textsf{c}}}\xspace}
\newcommand{\aWitnessJourney}{\ensuremath{\aJourney^{\textsf{w}}}\xspace}

\newcommand{\candidateJourneys}{\ensuremath{\journeys^\textsf{c}}\xspace}

\newcommand{\activeTrip}{\ensuremath{\aTrip_{\textsf{min}}}\xspace}

\SetKwFunction{CollectRoutes}{CollectRoutes}
\SetKwFunction{ScanRoutes}{ScanRoutes}
\SetKwFunction{RelaxTransfers}{RelaxTransfers}
\SetKwFunction{FindEarliestTripFrom}{FindEarliestTripFrom}

\SetKwFunction{dominates}{dominates}
\SetKwFunction{merge}{merge}

\newcommand{\reachedIndex}{\ensuremath{r}\xspace}

\SetKwFunction{Initialize}{Initialize}
\SetKwFunction{RelaxInitialTransfers}{RelaxInitialTransfers}
\SetKwFunction{CollectInitialTrips}{CollectInitialTrips}
\SetKwFunction{ScanTrips}{ScanTrips}
\SetKwFunction{Enqueue}{Enqueue}
\SetKwFunction{FindEarliestTrip}{FindEarliestTrip}

\SetKwFunction{ScanInitialRoutes}{ScanInitialRoutes}
\SetKwFunction{LinkTransfer}{LinkTransfer}
\SetKwFunction{MergeTransferLabel}{MergeTransferLabel}
\SetKwFunction{MergeTripLabel}{MergeTripLabel}

\newcommand{\shortcutEdges}{\ensuremath{\edges^\textsf{s}}\xspace}
\newcommand{\newShortcutEdges}{\ensuremath{\shortcutEdges_\text{new}}\xspace}

\newcommand{\departures}{\ensuremath{\mathcal{D}}\xspace}

\newcommand{\forwardArrivalTime}{\ensuremath{\overrightarrow{\arrivalTime}}\xspace}
\newcommand{\backwardDepartureTime}{\ensuremath{\overleftarrow{\departureTime}}\xspace}

\newcommand{\delay}{\ensuremath{\delta}\xspace}
\newcommand{\maxDelay}{\ensuremath{\delay^\text{max}}\xspace}
\newcommand{\delayScenario}{\ensuremath{\Delta}\xspace}
\newcommand{\departureDelay}{\ensuremath{\departure{\delayScenario}}\xspace}
\newcommand{\arrivalDelay}{\ensuremath{\arrival{\delayScenario}}\xspace}
\newcommand{\bestDelayScenario}{\ensuremath{\delayScenario^\text{best}}\xspace}
\newcommand{\parameterizedBestDelayScenario}{\ensuremath{\bestDelayScenario}\xspace}
\newcommand{\worstDelayScenario}{\ensuremath{\delayScenario^\text{worst}}\xspace}
\newcommand{\virtualDelayScenario}{\ensuremath{\delayScenario^\text{virt}}\xspace}
\newcommand{\candidateDelayScenario}{\ensuremath{\delayScenario^\text{can}}\xspace}
\newcommand{\witnessDelayScenario}{\ensuremath{\delayScenario^\text{wit}}\xspace}
\newcommand{\timeTravelDelayScenario}{\ensuremath{\delayScenario^\textsf{tt}}\xspace}
\newcommand{\leqEval}{\ensuremath{\preceq_\text{eval}}\xspace}
\newcommand{\geqEval}{\ensuremath{\succeq_\text{eval}}\xspace}
\newcommand{\leqDom}{\ensuremath{\preceq_\text{dom}}\xspace}
\newcommand{\geqDom}{\ensuremath{\succeq_\text{dom}}\xspace}

\newcommand{\sourceEvents}{\ensuremath{\stopEvents_\sourceIndex}\xspace}

\newcommand{\destinationEvents}{\ensuremath{\stopEvents_\destinationIndex}\xspace}
\newcommand{\targetEvents}{\ensuremath{\stopEvents_\targetIndex}\xspace}
\newcommand{\sourcePrefix}{\ensuremath{\aCandidateJourney_\sourceIndex}\xspace}
\newcommand{\originPrefix}{\ensuremath{\aCandidateJourney_\originIndex}\xspace}
\newcommand{\destinationPrefix}{\ensuremath{\aCandidateJourney_\destinationIndex}\xspace}
\newcommand{\sourcePrefixSequence}{\ensuremath{\langle\sourceEvent\rangle}\xspace}
\newcommand{\originStopPrefixSequence}[1]{\ensuremath{\langle[\sourceEvent,\originEvent],#1\rangle}\xspace}
\newcommand{\originPrefixSequence}{\ensuremath{\langle[\sourceEvent,\originEvent],\destinationVertex\rangle}\xspace}
\newcommand{\destinationPrefixSequence}{\ensuremath{\langle[\sourceEvent,\originEvent],\destinationEvent\rangle}\xspace}
\newcommand{\candidateSequence}{\ensuremath{\langle[\sourceEvent,\originEvent],[\destinationEvent,\targetEvent]\rangle}\xspace}

\newcommand{\slack}{\ensuremath{\operatorname{sl}}\xspace}
\newcommand{\splitLimit}{\ensuremath{\lambda_\textsf{s}}\xspace}
\newcommand{\destinationSplitLimit}{\ensuremath{\lambda_\textsf{\destinationIndex.s}}\xspace}
\newcommand{\directTargetSplitLimit}{\ensuremath{\lambda_\textsf{1\targetIndex.s}}\xspace}
\newcommand{\indirectTargetSplitLimit}{\ensuremath{\lambda_\textsf{2\targetIndex.s}}\xspace}
\newcommand{\indirectTargetSplitLimitBound}{\ensuremath{\underline{\indirectTargetSplitLimit}}\xspace}
\newcommand{\aggTargetSplitLimit}{\ensuremath{\lambda^\text{agg}_\textsf{\targetIndex.s}}\xspace}
\newcommand{\aggTargetSplitLimitBound}{\ensuremath{\underline{\aggTargetSplitLimit}}\xspace}
\newcommand{\joinLimit}{\ensuremath{\lambda_\textsf{j}}\xspace}
\newcommand{\feasibilityLimit}{\ensuremath{\lambda_\textsf{f}}\xspace}
\newcommand{\timeTravelLimit}{\ensuremath{\lambda_\textsf{tt}}\xspace}
\newcommand{\aggTimeTravelLimit}{\ensuremath{\lambda^\text{agg}_\textsf{tt}}\xspace}
\newcommand{\minOriginDelay}{\ensuremath{\delay^\originIndex_\text{min}}\xspace}
\newcommand{\minOriginDelayBound}{\ensuremath{\underline{\minOriginDelay}}\xspace}
\newcommand{\maxOriginDelay}{\ensuremath{\delay^\originIndex_\text{max}}\xspace}
\newcommand{\interval}{\ensuremath{I}\xspace}
\newcommand{\originDelayInterval}{\ensuremath{\interval_\delay^\originIndex}\xspace}
\newcommand{\originDelayIntervalBound}{\ensuremath{\underline{\originDelayInterval}}\xspace}
\newcommand{\maxWitnessDelay}{\ensuremath{\delay^\textsf{w}_\text{max}}\xspace}

\newcommand{\feasibleDestinationEvents}{\ensuremath{\stopEvents^\textsf{j/f}}\xspace}
\newcommand{\optimalCandidateJourneys}{\ensuremath{\journeys^{\textsf{c}}_\text{opt}}\xspace}
\newcommand{\optimalTargetEvents}{\ensuremath{\targetEvents^\text{opt}}\xspace}
\newcommand{\optimalDestinationEvents}{\ensuremath{\destinationEvents^\text{opt}}\xspace}
\newcommand{\virtualCandidateJourneys}{\ensuremath{\journeys^{\textsf{c}}_\text{virt}}\xspace}
\newcommand{\virtualTargetEvents}{\ensuremath{\targetEvents^\text{virt}}\xspace}
\newcommand{\witnessEntryEvents}{\ensuremath{\stopEvents_\textsf{w}^\text{in}}\xspace}
\newcommand{\candidateExitEvents}{\ensuremath{\stopEvents_\textsf{c}^\text{out}}\xspace}

\newcommand{\delayUpdate}{\ensuremath{\mu}\xspace}
\newcommand{\query}{\ensuremath{q}\xspace}
\newcommand{\queryExecutionTime}{\ensuremath{\atime_\text{ex}}\xspace}

\newcommand{\candidateArrivalTime}{\ensuremath{\arrivalTime^\textsf{c}}\xspace}
\newcommand{\witnessArrivalTime}{\ensuremath{\arrivalTime^\textsf{w}}\xspace}
\newcommand{\arrivalTimeDifference}{\ensuremath{\arrival{\atime}^\text{diff}}\xspace}
\newcommand{\maxWitnessTime}{\ensuremath{\atime^\textsf{w}_\text{max}}\xspace}
\newcommand{\timeTravelArrivalTime}{\ensuremath{\arrivalTime^\textsf{tt}}\xspace}
\newcommand{\entryIndex}{\ensuremath{\reachedIndex_\text{in}}\xspace}
\newcommand{\exitIndex}{\ensuremath{\reachedIndex_\text{out}}\xspace}
\newcommand{\witnessIndex}{\ensuremath{\reachedIndex_\textsf{w}}\xspace}
\newcommand{\tripBegin}{\ensuremath{\aTrip_b}\xspace}
\newcommand{\tripEnd}{\ensuremath{\aTrip_e}\xspace}

\newcommand{\delayedStopEvents}{\ensuremath{\stopEvents_\text{del}\xspace}}
\newcommand{\targetStops}{\ensuremath{\vertices_\targetIndex}\xspace}
\newcommand{\maxArrivalTime}{\ensuremath{\overline{\atime}_\text{arr}}\xspace}
\newcommand{\infeasibleShortcuts}{\ensuremath{\shortcutEdges_\text{inf}}\xspace}
\newcommand{\replacementShortcutEdges}{\ensuremath{\shortcutEdges_\textsf{r}}\xspace}
\newcommand{\originArrivalTime}{\ensuremath{\arrivalTime^\originIndex}\xspace}
\SetKwFunction{FindReplacements}{FindReplacements}
\usepackage{xcolor}
\definecolor{KITgreen}          {rgb}{0,    0.588,0.509}
\definecolor{KITgreen70}        {rgb}{0.3,  0.711,0.656}
\definecolor{KITgreen50}        {rgb}{0.5,  0.794,0.754}
\definecolor{KITgreen30}        {rgb}{0.7,  0.876,0.852}
\definecolor{KITgreen15}        {rgb}{0.85, 0.938,0.926}
\definecolor{KITblue}           {rgb}{0.274,0.392,0.666}
\definecolor{KITblue70}         {rgb}{0.492,0.574,0.766}
\definecolor{KITblue50}         {rgb}{0.637,0.696,0.833}
\definecolor{KITblue30}         {rgb}{0.782,0.817,0.9}
\definecolor{KITblue15}         {rgb}{0.891,0.908,0.95}
\definecolor{KITpalegreen}      {rgb}{0.509,0.745,0.235}
\definecolor{KITpalegreen70}    {rgb}{0.656,0.821,0.464}
\definecolor{KITpalegreen50}    {rgb}{0.754,0.872,0.617}
\definecolor{KITpalegreen30}    {rgb}{0.852,0.923,0.77}
\definecolor{KITpalegreen15}    {rgb}{0.926,0.961,0.885}
\definecolor{KITyellow}         {rgb}{0.98, 0.901,0.078}
\definecolor{KITyellow70}       {rgb}{0.986,0.931,0.354}
\definecolor{KITyellow50}       {rgb}{0.99, 0.95, 0.539}
\definecolor{KITyellow30}       {rgb}{0.994,0.97, 0.723}
\definecolor{KITyellow15}       {rgb}{0.997,0.985,0.861}
\definecolor{KITorange}         {rgb}{0.862,0.627,0.117}
\definecolor{KITorange70}       {rgb}{0.903,0.739,0.382}
\definecolor{KITorange50}       {rgb}{0.931,0.813,0.558}
\definecolor{KITorange30}       {rgb}{0.958,0.888,0.735}
\definecolor{KITorange15}       {rgb}{0.979,0.944,0.867}
\definecolor{KITbrown}          {rgb}{0.627,0.509,0.196}
\definecolor{KITbrown70}        {rgb}{0.739,0.656,0.437}
\definecolor{KITbrown50}        {rgb}{0.813,0.754,0.598}
\definecolor{KITbrown30}        {rgb}{0.888,0.852,0.758}
\definecolor{KITbrown15}        {rgb}{0.944,0.926,0.879}
\definecolor{KITred}            {rgb}{0.627,0.117,0.156}
\definecolor{KITred70}          {rgb}{0.739,0.382,0.409}
\definecolor{KITred50}          {rgb}{0.813,0.558,0.578}
\definecolor{KITred30}          {rgb}{0.888,0.735,0.747}
\definecolor{KITred15}          {rgb}{0.944,0.867,0.873}
\definecolor{KITlilac}          {rgb}{0.627,0,    0.47}
\definecolor{KITlilac70}        {rgb}{0.739,0.3,  0.629}
\definecolor{KITlilac50}        {rgb}{0.813,0.5,  0.735}
\definecolor{KITlilac30}        {rgb}{0.888,0.7,  0.841}
\definecolor{KITlilac15}        {rgb}{0.944,0.85, 0.92}
\definecolor{KITcyanblue}       {rgb}{0.313,0.666,0.901}
\definecolor{KITcyanblue70}     {rgb}{0.519,0.766,0.931}
\definecolor{KITcyanblue50}     {rgb}{0.656,0.833,0.95}
\definecolor{KITcyanblue30}     {rgb}{0.794,0.9,  0.97}
\definecolor{KITcyanblue15}     {rgb}{0.897,0.95, 0.985}
\definecolor{KITseablue}        {rgb}{0.196,0.313,0.549}
\definecolor{KITseablue70}      {rgb}{0.437,0.519,0.684}
\definecolor{KITseablue50}      {rgb}{0.598,0.656,0.774}
\definecolor{KITseablue30}      {rgb}{0.758,0.794,0.864}
\definecolor{KITseablue15}      {rgb}{0.879,0.897,0.932}
\definecolor{KITblack}          {rgb}{0,    0,    0}
\definecolor{KITblack90}        {rgb}{0.1,  0.1,  0.1}
\definecolor{KITblack80}        {rgb}{0.2,  0.2,  0.3}
\definecolor{KITblack75}        {rgb}{0.25, 0.25, 0.25}
\definecolor{KITblack70}        {rgb}{0.3,  0.3,  0.3}
\definecolor{KITblack60}        {rgb}{0.4,  0.4,  0.4}
\definecolor{KITblack50}        {rgb}{0.5,  0.5,  0.5}
\definecolor{KITblack40}        {rgb}{0.6,  0.6,  0.6}
\definecolor{KITblack30}        {rgb}{0.7,  0.7,  0.7}
\definecolor{KITblack25}        {rgb}{0.75, 0.75, 0.75}
\definecolor{KITblack20}        {rgb}{0.8,  0.8,  0.8}
\definecolor{KITblack10}        {rgb}{0.9,  0.9,  0.9}
\definecolor{KITwhite}          {rgb}{1,    1,    1}
\definecolor{LIPICSorange}      {rgb}{0.988,0.78 ,0.07}

\colorlet{tripColor1}{KITblue}
\colorlet{tripColor2}{KITred}
\colorlet{tripColor3}{KITpalegreen}
\colorlet{tripColor4}{KITorange}

\colorlet{rainbowColor1}{KITgreen}
\colorlet{rainbowColor2}{KITpalegreen}
\colorlet{rainbowColor3}{KITcyanblue}
\colorlet{rainbowColor4}{KITblue}
\colorlet{rainbowColor5}{KITred}
\colorlet{rainbowColor6}{KITorange}

\colorlet{mrColorRoute}{KITblue50}
\colorlet{mrColorTransfer}{KITblue70}
\colorlet{mrColorOther}{KITblue}
\colorlet{urColorRoute}{KITorange50}
\colorlet{urColorTransfer}{KITorange70}
\colorlet{urColorOther}{KITorange}
\colorlet{utbColor}{KITred}
\usepackage{tikz}
\usepackage{pgfplots}
\usepackage{sidecap}
\pgfplotsset{compat=1.16}

\usetikzlibrary{decorations.text}
\usetikzlibrary{decorations.pathmorphing}
\usetikzlibrary{pgfplots.statistics}
\usetikzlibrary{fit}

\colorlet{nodeColor}{black!80}
\colorlet{nodeColorFaded}{black!50}
\colorlet{edgeColor}{black!50}
\colorlet{edgeColorFaded}{black!30}
\colorlet{slowModeColor}{KITlilac}
\colorlet{fastModeColor}{KITcyanblue}
\colorlet{axisColor}{black!80}
\colorlet{legendColor}{black!80}
\colorlet{meanColor}{black!80}

\newcommand{\gs}{\hphantom{\tiny$\cdot$}}
\newcommand{\legend}[1]{\raisebox{0.07ex}{#1}}

\tikzstyle{vertex}=[circle,line width=.5pt,minimum size=0.1pt]
\tikzstyle{routeArrow}=[->, >=stealth]
\tikzstyle{route}=[routeArrow, line width=2.5pt, rounded corners = 20]
\tikzstyle{routeNoArrow}=[line width=2.5pt, rounded corners = 20]
\tikzstyle{edge}=[edgeColor, line width=1pt, rounded corners = 20]
\tikzstyle{edgeFaded}=[edgeColorFaded, line width=1pt, rounded corners = 20]
\tikzstyle{directedEdge}=[edge, routeArrow, shorten >=2.7pt]
\tikzstyle{linePlot}=[mark size=1.8pt, line width=1.2pt]
\tikzstyle{stop}=[rectangle, line width=.5pt, inner sep=5pt, draw=nodeColor!100,fill=nodeColor!15,rounded corners=0.1cm]
\tikzstyle{stopFaded}=[rectangle, line width=.5pt, inner sep=5pt, draw=nodeColor!50,fill=nodeColor!15,rounded corners=0.1cm]

\pgfplotsset{
    grid style={axisColor!20,line width = 0.2pt,dash pattern = on 2pt off 1pt},
    grid=both,
    axis line style={axisColor,line width = 0.4pt},
    major tick style={axisColor,line width = 0.4pt},
    minor tick style={axisColor,line width = 0.2pt},
    major tick length=3.5pt,
    minor tick length=2.0pt,
    ytick align=outside,
    xtick align=outside,
    ticklabel style = {font=\small},
    label style = {font=\small},
    boxplot/every median/.style={meanColor, line width = 1.5pt},
}

\begin{document}

\maketitle

\begin{abstract}
We study journey planning in multimodal networks consisting of public transit plus an unrestricted transfer mode (e.g., walking, cycling, e-scooter).
In order to provide good results in practice, algorithms must account for vehicle delays.
Delay-responsive algorithms receive a continuous stream of delay updates and must return optimal journeys in the currently known delay scenario.
Updates are incorporated in an update phase, which must be fast (e.g., a few seconds).
The fastest known approach for multimodal journey planning is ULTRA, which precomputes shortcuts representing transfers between vehicles.
This allows query algorithms to find Pareto-optimal journeys regarding arrival time and the number of public transit trips without any performance loss compared to pure public transit networks.
However, the precomputation phase does not account for delays and is too slow to rerun during the update phase.
We present Delay-ULTRA, a delay-responsive variant of ULTRA.
Since accounting for all theoretically possible delays would yield an impractically large set of shortcuts, our approach precomputes shortcuts that are provably sufficient as long as delays do not exceed a configurable limit (e.g., 5 minutes).
To handle delays above the limit (which are less frequent in practice), we propose a heuristic search for missing shortcuts that is fast enough to be run during the update phase.
Our experimental evaluation on real-world data shows that Delay-ULTRA has extremely low error rates, failing to find less than 0.02\% of optimal journeys on metropolitan and mid-sized country networks, and 0.16\% on the much larger Germany network.
Considering that the delay information that is available in realistic applications is never perfectly accurate or up to date, these error rates are negligible.
Query speed is at most twice as slow as ULTRA without delay information, and up to 8 times faster than the fastest exact algorithm, which does not require a preprocessing phase.
\end{abstract}

\newpage

\section{Introduction}
Modern transportation networks are increasingly multimodal, combining public and private transport.
Finding good travel itineraries in these networks can be challenging and requires the aid of multimodal journey planning algorithms.
In this work, we study bimodal networks that combine schedule-based public transit with an unrestricted \emph{transfer graph} representing one road-based mode (e.g., walking, cycling, e-scooter).
Given a point-to-point query, the objective is to find a Pareto set of journeys optimizing arrival time and the number of public transit trips.
The fastest known algorithms for this problem rely on ULTRA (UnLimited TRAnsfers)~\cite{Bau23}, a preprocessing technique that condenses the transfer graph into a set of \emph{shortcuts} between public transit vehicles.
However, its practical applicability is hampered by the fact that the shortcuts do not account for delays in the vehicle schedules.

We distinguish between two types of delay-robust algorithms:
\emph{Delay-anticipating} algorithms aim to find journeys that are robust regarding potential future delays that are not known at query time.
We refer to~\cite{Bas16b} for an overview of such algorithms.
In contrast, we study \emph{delay-responsive} algorithms, which aim to find journeys that are optimal in the currently known delay scenario.
These algorithms receive a continuous stream of \emph{delay updates}, which they need to periodically incorporate into their query data structures in an \emph{update phase}.
This phase needs to be fast (ideally no more than a few seconds) to ensure that the query algorithm runs on reasonably up-to-date information.

\subparagraph*{Related Work.}
An overview of the state of the art in journey planning is given in~\cite{Bas16b}.
Algorithms for pure public transit are divided into graph-based and timetable-based approaches.
The former model the network as a graph~\cite{Del08,Mue09} and apply variants of Dijkstra's algorithm~\cite{Dij59}.
This allows for fast delay updates, especially if the model is designed with updates in mind~\cite{Cio17}.
Timetable-based approaches such as RAPTOR~\cite{Del15b} achieve faster query times with cache-friendly data structures, which are typically lightweight enough that they can be rebuilt from scratch during the update phase.
The fastest such algorithm is Trip-Based Routing~(TB)~\cite{Wit15}, which computes transfers between vehicles in a preprocessing phase.
Delays can be incorporated by recomputing transfers for the affected vehicles~\cite{Wit21}.
Among techniques with heavy preprocessing, Bast et al.~\cite{Bas13b} show that Transfer Patterns~\cite{Bas10} answers almost all queries correctly even if the preprocessed data structures are not adjusted to incorporate delays.
However, the queries selected in their experiments are biased towards large stations and rush hours; it is unclear whether the error rate is also low for queries selected uniformly at random.
For Public Transit Labeling~\cite{Del15}, a dynamic version has been proposed~\cite{Emi19}, but it only optimizes arrival time.
Common to all these techniques is that they require a transitively closed transfer graph.
In particular, they assume that the set of potential transfers that are incident to a public transit vehicle is small and local, which makes it easy to enumerate when handling a delay update.
This is no longer the case with an unrestricted transfer graph.

For multimodal networks, the fastest algorithm without expensive preprocessing is MR~\cite{Del13}, which combines RAPTOR with Dijkstra searches on a contracted transfer graph.
It is delay-responsive but fairly slow compared to RAPTOR.
The delay-responsive graph model of~\cite{Cio17} has been extended to multimodal networks~\cite{Gia19}, but query times are not competitive with MR.
Faster query times are achieved with ULTRA~\cite{Bau23}, which exploits the observation that only a fairly small number of transfers between public transit trips occur in Pareto-optimal journeys.
These are precomputed and condensed into a \emph{shortcut graph}, which can be used by any existing public transit algorithm.
Combining ULTRA and TB yields the fastest known multimodal algorithm~\cite{Sau20b}.
However, the ULTRA preprocessing phase does not account for delays and is too slow to be rerun in the update phase.

\subparagraph*{Contribution and Outline.}
We present Delay-ULTRA, a variant of ULTRA that anticipates possible delays during the shortcut computation phase.
Doing this for arbitrarily high delays poses a conceptual challenge:
Nearly every journey is optimal in at least one theoretically possible delay scenario (e.g., one that applies extreme delays to all competing journeys).
Thus, almost every possible shortcut is required in theory.
In realistic scenarios, however, most vehicles are delayed by no more than a few minutes.
Delay-ULTRA exploits this with a three-phase approach:
The first preprocessing phase, which does not receive any delay information, computes a set of shortcuts that is provably sufficient for all delays below a specified threshold (e.g., 5 minutes).
During the update phase, when delays are known, irrelevant shortcuts are discarded and additional ones for higher delays are computed in a second step.
Finding all required shortcuts is not feasible within the few seconds allowed by the update phase, since this would require enumerating all potential incident shortcuts of the delayed vehicles.
Instead, we propose a heuristic \emph{replacement search} for shortcuts that have become infeasible due to delays.
Finally, the query algorithm runs TB on the filtered set of Delay-ULTRA shortcuts as well as the replacement shortcuts.

The remainder of this paper is organized as follows.
Section~\ref{sec:preliminaries} introduces notation and definitions.
Furthermore, it gives an overview of the RAPTOR and ULTRA algorithms, upon which we build.
In Section~\ref{sec:optimality-conditions}, we establish a characterization of the shortcuts that are required for delays up to a certain limit.
While this characterization is compact, it does not immediately imply an efficient algorithm for enumerating all such shortcuts.
Section~\ref{sec:optimality-tests} therefore develops formulas for computing them efficiently.
Based on these results, Section~\ref{sec:delay-ultra} outlines the Delay-ULTRA shortcut computation algorithm.
Section~\ref{sec:update} describes the update phase, including the replacement search.
In Section~\ref{sec:experiments}, we evaluate our algorithms on four real-world multimodal networks, using a synthetic delay model based on real-world punctuality data.
Our experiments show that the original ULTRA-TB is already fairly delay-robust, failing to find at most 1\% of optimal journeys.
Delay-ULTRA reduces the error rate by a factor of~4--15, yielding less than~0.02\% suboptimal journeys for metropolitan and mid-sized country networks, and~0.16\% for the much larger Germany network.
If we take into account that a realistic routing application cannot have access to perfectly accurate or up-to-date delay information, these error rates are negligible.
Our query algorithm, Delay-ULTRA-TB, retains a speedup of up to~8 compared to MR, and is at most two times slower than ULTRA-TB.
Finally, Section~\ref{sec:conclusion} summarizes our results and gives an outlook on future work.

\section{Preliminaries}
\label{sec:preliminaries}
\subparagraph*{Network.}
A public transit network is a 5-tuple~$(\stops,\stopEvents,\trips,\routes,\graph)$ consisting of a set of \emph{stops}~$\stops$, a set of \emph{stop events}~$\stopEvents$, a set of \emph{trips}~$\trips$, a set of \emph{routes}~$\routes$ and a directed, weighted~\emph{transfer graph}~$\graph=(\vertices,\edges)$.
A stop~$\astop\in\stops$ is a location where vehicles are entered or exited.
A stop event~$\stopEvent\in\stopEvents$ represents a vehicle arriving at a stop~$\astop(\stopEvent)\in\stops$ with the scheduled arrival time~$\arrivalTime(\stopEvent)$ and then departing from the same stop with the scheduled departure time~$\departureTime(\stopEvent)$.
Some works consider a \emph{departure buffer time} that passengers must observe before entering a vehicle.
As shown in~\cite{Bau23}, this can be handled implicitly by reducing~$\departureTime(\stopEvent)$ accordingly.
A trip~$\aTrip=\langle\stopEvent_0,\dots,\stopEvent_k\rangle\in\trips$ is a sequence of stop events performed by the same vehicle.
The~$i$-th stop event of~$\aTrip$ is denoted as~$\aTrip[i]$, and the number of stop events in~$\aTrip$ as~$|\aTrip|$.
The trip of a stop event~$\stopEvent$ is denoted as~$\aTrip(\stopEvent)$.
We denote the sets of stop events preceding and succeeding~$\aTrip[i]$ in~$\aTrip$ by~$\precedingEvents(\aTrip[i])$ and~$\succeedingEvents(\aTrip[i])$, respectively.
A~\emph{trip segment}~\mbox{$\aTripSegment{i}{j}:=\langle\stopEvent_i,\dots,\stopEvent_j\rangle$} is the contiguous subsequence of~\aTrip from~$\aTrip[i]$ to~$\aTrip[j]$.
The trips are partitioned into a set of routes~$\routes$ such that all trips of a route share the same stop sequence and no trip overtakes another along the sequence.
This induces a total ordering on the trips of a route:
for two trips~$\aTripA,\aTripB$ of the same route~$\aRoute\in\routes$, we write~$\aTripA\preceq\aTripB$ if~$\arrivalTime(\aTripA[i])\leq\arrivalTime(\aTripB[i])$ for every index~$i$ along~$\aRoute$.
Note that trips from different routes are not comparable via~$\preceq$.
The transfer graph~$\graph=(\vertices,\edges)$ consists of a set of \emph{vertices}~$\vertices$ with~$\stops\subseteq\vertices$, and a set of \emph{edges}~$\edges\subseteq\vertices\times\vertices$.
Traversing an edge~$\edge$ requires the \emph{transfer time}~$\transferTime(\edge)$.
The transfer time~$\transferTime(\aPath)$ of a path~$\aPath$ in~$\graph$ is the sum of the transfer times of its edges.\looseness=-1

\subparagraph*{Delays.}
A \emph{delay scenario}~\delayScenario assigns to every stop event~$\stopEvent\in\stopEvents$ an \emph{arrival} and~\emph{departure delay}~$\arrivalDelay(\stopEvent),\departureDelay(\stopEvent)\in\mathbb{N}_0$.
This yields delayed arrival and departure times~$\arrivalTime(\delayScenario,\stopEvent)=\arrivalTime(\stopEvent)+\arrivalDelay(\stopEvent)$ and~$\departureTime(\delayScenario,\stopEvent)=\departureTime(\stopEvent)+\departureDelay(\stopEvent)$.
Note that this definition considers the delays of all stop events independently of each other.
This allows impossible delay scenarios in which vehicles travel faster than is possible, or even backwards in time.
We permit these to avoid introducing dependencies between stop events (see Section~\ref{sec:time-travel}).

\subparagraph*{Journeys.}
A~\emph{journey} represents the movement of a passenger through the public transit network.
Formally, it is an alternating sequence~$\aJourney=\langle\aPath_0,\aNamedTripSegment{0}{i}{j},\dots,\aNamedTripSegment{k-1}{m}{n},\aPath_k\rangle$ of transfers and trip segments.
A~\emph{transfer} is a shortest path in~$\graph$; it is called \emph{empty} if it consists of a single vertex.
For source and target vertices~$\sourceVertex,\targetVertex\in\vertices$, we call~$\aJourney$ an~$\sourceVertex$-$\targetVertex$-journey if~$\aPath_0$ begins with~$\sourceVertex$ and~$\aPath_k$ ends with~$\targetVertex$.
We call~$\aPath_0$ the~\emph{initial transfer}, $\aPath_k$ the~\emph{final transfer} and all other transfers~\emph{intermediate transfers}.
For a delay scenario~\delayScenario, the latest possible departure time of~$\aJourney$ at~$\sourceVertex$ is~$\departureTime(\delayScenario,\aJourney):=\departureTime(\delayScenario,\aTrip_0[i])-\transferTime(\aPath_0)$ and the arrival time at~$\targetVertex$ is~$\arrivalTime(\delayScenario,\aJourney):=\arrivalTime(\delayScenario,\aTrip_{k-1}[n])+\transferTime(\aPath_{k})$.
The number of used trips is denoted as~$|\aJourney|:=k$.
We call~\aJourney \emph{feasible} for a departure time~\departureTime if~$\departureTime(\delayScenario,\aJourney)\geq\departureTime$ and all intermediate transfers are feasible.
An intermediate transfer between stop events~$\originEvent$ and~$\destinationEvent$ is feasible if~$\departureTime(\delayScenario,\destinationEvent)\geq\arrivalTime(\delayScenario,\originEvent)+\transferTime(\astop(\originEvent),\astop(\destinationEvent))$.

A journey~$\aJourney$ is evaluated according to its arrival time~$\arrivalTime(\delayScenario,\aJourney)$ and number of trips~$|\aJourney|$.
We say that~\aJourney \emph{weakly dominates} another journey~$\aJourney'$ if~$\arrivalTime(\delayScenario,\aJourney)\leq\arrivalTime(\delayScenario,\aJourney')$ and~$|\aJourney|\leq|\aJourney'|$.
If one criterion is also strictly smaller, $\aJourney$ \emph{strongly dominates}~$\aJourney'$.
A journey~$\aJourney$ is \emph{Pareto-optimal} if it is not strongly dominated by any journey that is feasible for~$\departureTime(\delayScenario,\aJourney)$.
A \emph{Pareto set}~\journeys is a set containing a minimal number of journeys such that every valid journey is weakly dominated by a journey in the set.

\subparagraph*{Problem Statement.}
A \emph{delay update}~$\delayUpdate$ represents changes in the delays of a single trip~$\aTrip$, starting at some index~$i$.
For each stop event~$\aTrip[j]$ with~$j\geq{}i$, $\delayUpdate$ specifies a new departure and arrival delay.
A delay-responsive algorithm receives a stream~$\langle\delayUpdate_1,\delayUpdate_2,\dots\rangle$ of delay updates and a stream~$\langle\query_0,\query_1,\dots\rangle$ of~\emph{queries}.
In the initial delay scenario~$\delayScenario_0$, all stop events are punctual.
After receiving update~$\delayUpdate_i$, its new delays are applied to the previous scenario~$\delayScenario_{i-1}$ to obtain the current scenario~$\delayScenario_i$.
A query~$\query$ consists of source and target vertices~$\sourceVertex,\targetVertex\in\vertices$, an earliest departure time~\departureTime and an~\emph{execution time}~$\queryExecutionTime\leq\departureTime$.
It must be answered with a Pareto set of~$\sourceVertex$-$\targetVertex$-journeys departing no later than~\departureTime in the delay scenario~$\delayScenario$ that is current at~$\queryExecutionTime$.

\subparagraph*{Partial Journeys.}
Since journeys that share a common suffix always have the same arrival time, there are often multiple equivalent Pareto-optimal journeys, any of which can be chosen for a Pareto set.
RAPTOR and ULTRA use the following rules to break some of these ties:
If two journeys meet during a transfer, the one with the lower arrival time at the meeting point is preferred.
If they meet during a trip segment, the one that entered the trip first is preferred.
We formalize these rules by introducing partial journeys, which start or end midway through a trip segment.
A~\emph{trip segment prefix}~$\aTripSegment{i}{\cdot}$ represents a passenger entering the trip~$\aTrip$ at~$\aTrip[i]$ but not yet exiting it, whereas a~\emph{trip segment suffix}~$\aTripSegment{\cdot}{j}$ represents a passenger entering~$\aTrip$ at an unspecified stop event and exiting at~$\aTrip[j]$.
A~\emph{journey prefix} is either a journey or a journey followed by a trip segment prefix.
Likewise, a~\emph{journey suffix} is either a journey or a journey preceded by a trip segment suffix.
In both cases, the incomplete trip segment counts as a used trip.
Journey prefixes and suffixes are collectively called~\emph{partial journeys}.
Two (partial) journeys~$\aJourney_1$ and~$\aJourney_2$ can be concatenated into a (partial) journey~$\aJourney_1\circ\aJourney_2$ if one of two conditions is fulfilled:
\begin{enumerate}
	\item $\aJourney_1$ ends with a final transfer~$\aPath_1$ to a vertex~$\aVertex$ and~$\aJourney_2$ starts with an initial transfer~$\aPath_2$ from the same vertex~$\aVertex$.
	Then~$\aJourney_1\circ\aJourney_2$ is obtained by replacing~$\aPath_1$ and~$\aPath_2$ with the intermediate transfer~$\aPath_1\circ\aPath_2$.
	\item $\aJourney_1$ ends with a trip segment prefix~$\aTripSegment{i}{\cdot}$, $\aJourney_2$ starts with a trip segment suffix~$\aTripSegment{\cdot}{j}$, and~$i<j$.
	Then~$\aJourney_1\circ\aJourney_2$ contains the proper trip segment~$\aTripSegment{i}{j}$ in their place.
\end{enumerate}

A journey~$\aJourney=\langle\aPath_0,\dots,\aNamedTripSegment{k}{i}{j},\aPath_{k+1},\dots\rangle$ has two~\emph{standard prefixes} with~$k$ trips:
the journey prefix~$\langle\aPath_0,\dots,\aNamedTripSegment{k}{i}{\cdot}\rangle$ and the proper journey~$\langle\aPath_0,\dots,\aNamedTripSegment{k}{i}{j},\aPath_{k+1}\rangle$.
Additionally, each journey~$\langle\aPath_0,\dots,\aNamedTripSegment{k}{i}{j},\aPath'_{k+1}\rangle$ where~$\aPath'_{k+1}$ is a prefix of~$\aPath_{k+1}$ is a~\emph{non-standard prefix} of~$\aJourney$.\looseness=-1

We extend the notion of dominance to journey prefixes~$\prefix{\aJourney}$ and~$\prefix{\aJourney'}$:
If~$\prefix{\aJourney}$ and~$\prefix{\aJourney'}$ are both proper journeys, i.e.,~end with final transfers, the original definition applies.
If~$\prefix{\aJourney}$ ends with a trip segment prefix~$\aTripSegment{i}{\cdot}$ and~$\prefix{\aJourney'}$ with a trip segment prefix~$\aTripSegment{j}{\cdot}$ of the same trip~$\aTrip$, then~$\prefix{\aJourney'}$ weakly dominates~$\prefix{\aJourney}$ if~$|\prefix{\aJourney'}|\leq|\prefix{\aJourney}|$ and~$j\leq{}i$.
If~$|\prefix{\aJourney'}|<|\prefix{\aJourney}|$ or~$j<i$ also holds, then~$\prefix{\aJourney'}$ strongly dominates~$\prefix{\aJourney}$.
In all other cases, $\prefix{\aJourney'}$ does not strongly or weakly dominate~$\prefix{\aJourney}$.
The definition of Pareto optimality carries over from proper journeys.
We call a journey~\emph{prefix-optimal} if all of its prefixes are Pareto-optimal.
Note that if a standard prefix that ends with a final transfer~$\aPath$ is Pareto-optimal, then so are all non-standard prefixes ending with a prefix~$\aPath'$ of~$\aPath$.
Thus, a journey is prefix-optimal iff all of its standard prefixes are Pareto-optimal.

\subparagraph*{Algorithms.}
MR~\cite{Del13} is an extension of RAPTOR~\cite{Del15b} that Pareto-optimizes arrival time and number of trips in networks with an unrestricted transfer graph.
It operates in \emph{rounds}, where round~$n$ finds journeys with~$n$ trips by extending Pareto-optimal journeys with~$n-1$ trips.
For each vertex~$\aVertex\in\vertices$ and each round~$n$, MR maintains a~\emph{tentative arrival time}~$\arrivalTime(\aVertex,n)$, which is the earliest arrival time among all journeys to~$\aVertex$ with at most~$n$ trips that have been found so far.
At the start of round~$n$, it is initialized with the arrival time~$\arrivalTime(\aVertex,n-1)$ from the previous round.
If~$\arrivalTime(\aVertex,n)$ is improved during round~$n$, then~$\aVertex$ is~\emph{marked} for the next round.\looseness=-1

Each round consists of a~\emph{route scanning phase} and a~\emph{transfer relaxation phase}.
The route scanning phase of round~$n$ collects all routes that visit marked stops and \emph{scans} them.
A route~$\aRoute$ is scanned by iterating across the stops visited by~$\aRoute$, starting at the first marked stop.
During the scan, the algorithm maintains an~\emph{active trip}~$\activeTrip$, which is the earliest trip of~$\aRoute$ that has been entered so far.
Let~$\astop$ be the~$i$-th stop of~$\aRoute$.
If~$\activeTrip$ has already been set, the algorithm checks whether exiting~$\activeTrip$ at~$\astop$ with arrival time~$\arrivalTime(\activeTrip[i])$ improves~$\arrivalTime(\astop,n)$.
If so, $\arrivalTime(\astop,n)$ is updated accordingly and~$\astop$ is marked.
Afterwards, the algorithm checks whether there is an earlier trip than~$\activeTrip$ that can be entered when arriving at~$\astop$ with arrival time~$\arrivalTime(\astop,n-1)$.
If so, $\activeTrip$ is updated accordingly.
The transfer relaxation phase runs Dijkstra's algorithm on a contracted version of the transfer graph, using~$\arrivalTime(\cdot,n)$ as the tentative distances.
The priority queue is initialized with all marked stops, and all stops that are settled by the search are themselves marked.

To answer a query with source stop~$\sourceVertex\in\stops$ and departure time~$\departureTime$, MR initializes~$\arrivalTime(\sourceVertex,0)$ with~$0$ and all other arrival times with~$\infty$.
Then it performs round~$0$ by marking~$\sourceVertex$ and relaxing its outgoing transfers.
Afterwards, new rounds are performed until no more stops have been marked.

ULTRA~\cite{Bau23} allows any query algorithm to handle an unrestricted transfer graph by precomputing a set~$\shortcutEdges\subseteq\stopEvents\times\stopEvents$ of \emph{shortcuts}.
A shortcut~$(\aTripA[i],\aTripB[j])\in\shortcutEdges$ represents a passenger exiting the trip~$\aTripA$ at the stop event~$\aTripA[i]$, taking a transfer from~$\astop(\aTripA[i])$ to~$\astop(\aTripB[j])$ and entering the trip~$\aTripB$ at the stop event~$\aTripB[j]$.
When combined with TB, ULTRA replaces the preprocessing phase of TB.
An ULTRA-TB query then explores intermediate transfers with~$\shortcutEdges$, and initial and final transfers with Bucket-CH~\cite{Kno07,Gei12}, a one-to-many search algorithm for road networks.
The shortcuts are computed by enumerating \emph{candidates}, which are journeys with exactly two trips connected by an intermediate transfer, but with empty initial and final transfers.
ULTRA produces correct results if it enumerates all prefix-optimal candidates\footnote{Since prefix optimality does not resolve all ties between equivalent journeys, ULTRA introduces the stronger notion of~\emph{canonicity} to break the remaining ties. While this reduces the number of shortcuts, it is not required for correctness since every canonical journey is prefix-optimal.} and generates shortcuts representing their intermediate transfers.
To determine whether a candidate~$\aCandidateJourney$ is prefix-optimal, ULTRA tests for each prefix~$\prefix{\aCandidateJourney}$ of~$\aCandidateJourney$ whether it is strongly dominated by another journey (prefix)~$\aWitnessJourney$.
If so, $\aWitnessJourney$ is called a~\emph{witness} for~$\prefix{\aCandidateJourney}$ and~$\aCandidateJourney$ is discarded.

\section{Characterizing Required Shortcuts}
\label{sec:optimality-conditions}
Since a set of shortcuts that is sufficient for all possible delay scenarios would be too large, we limit the precomputation to scenarios in which no delay exceeds a given \emph{delay limit}~$\maxDelay$.
Our goal is to enumerate all candidates that are prefix-optimal in at least one such delay scenario.
For the purpose of the shortcut computation, which we discuss throughout the following sections, we assume that all delay scenarios conform to this delay limit.
We discuss how to handle scenarios beyond the delay limit in Section~\ref{sec:update}.

Even with the delay limit, the number of possible delay scenarios is still astronomical, so it is not feasible to consider each one individually.
Instead, we develop a more succinct characterization of candidates that are prefix-optimal in at least one delay scenario.
To aid in this, Section~\ref{sec:optimality-conditions:notation} introduces a simplified notation for candidates and witnesses.
Section~\ref{sec:optimality-conditions:virtual} establishes a simple characterization under the assumption that candidates and witnesses do not have any stop events in common.
In Section~\ref{sec:optimality-conditions:shared}, we generalize this result to allow for shared stop events.
Based on our findings, Section~\ref{sec:optimality-conditions:definition} formally defines the shortcut computation problem in the presence of delays.

\subsection{Notation}
\label{sec:optimality-conditions:notation}
\begin{figure}
	\centering
	\begin{tikzpicture}
\node (es) at (0.00, 0.00) {};%
\node (eo) at (4.00, 0.00) {};%
\node (es2) at (1.00, 1.25) {};%
\node (eo2) at (5.00, 1.25) {};%
\node (v_inter) at (6.00, 0.00) {};%
\node (e1) at (6.00, -1.00) {};%
\node (ed) at (8.00, 0.00) {};%
\node (e2) at (10.00, 0.00) {};%
\node (et) at (12.00, 0.00) {};%
\node (ed2) at (7.00, 1.25) {};%
\node (et2) at (10.00, 1.25) {};%
\node (v_final) at (11.50, 0.90) {};%

\node (vs) at (0.00,-0.50) {};%
\node (vo) at (4.00,-0.50) {};%
\node (vd) at (8.00,-0.50) {};%
\node (vt) at (12.00,-0.50) {};%

\draw [edge]  (es) -- (es2);
\draw [edge]  (eo) -- (v_inter);
\draw [edge]  (eo2) -- (v_inter);
\draw [edge]  (v_inter) -- (e1);
\draw [edge]  (v_inter) -- (ed);
\draw [edge]  (v_inter) -- (ed2);
\draw [edge]  (e2) -- (v_final);
\draw [edge]  (et2) -- (v_final);
\draw [edge]  (v_final) -- (et);

\node [stop,fit=(es)(vs)] {};%
\node [stop,fit=(eo)(vo)] {};%
\node [stop,fit=(ed)(vd)] {};%
\node [stop,fit=(et)(vt)] {};%
\node [stop,fit=(es2)] {};%
\node [stop,fit=(eo2)] {};%
\node [stop,fit=(ed2)] {};%
\node [stop,fit=(et2)] {};%
\node [stop,fit=(e1)] {};%
\node [stop,fit=(e2)] {};%

\draw [tripColor1,route] (es) -- (eo);
\draw [tripColor2,route] (es2) -- (eo2);
\draw [tripColor3,route] (e1) -- (ed) -- (e2) -- (et);
\draw [tripColor4,route] (ed2) -- (et2);

\node (es_v) at (es) [vertex,draw=tripColor1,fill=tripColor1!15] {\gs};%
\node (eo_v) at (eo) [vertex,draw=tripColor1,fill=tripColor1!15] {\gs};%
\node (es2_v) at (es2) [vertex,draw=tripColor2,fill=tripColor2!15] {\gs};%
\node (eo2_v) at (eo2) [vertex,draw=tripColor2,fill=tripColor2!15] {\gs};%
\node (v_inter_v) at (v_inter) [vertex,draw=KITblack,fill=KITblack!15] {};%
\node (e1_v) at (e1) [vertex,draw=tripColor3,fill=tripColor3!15] {\gs};%
\node (ed_v) at (ed) [vertex,draw=tripColor3,fill=tripColor3!15] {\gs};%
\node (e2_v) at (e2) [vertex,draw=tripColor3,fill=tripColor3!15] {\gs};%
\node (et_v) at (et) [vertex,draw=tripColor3,fill=tripColor3!15] {\gs};%
\node (ed2_v) at (ed2) [vertex,draw=tripColor4,fill=tripColor4!15] {\gs};%
\node (et2_v) at (et2) [vertex,draw=tripColor4,fill=tripColor4!15] {\gs};%
\node (v_final_v) at (v_final) [vertex,draw=KITblack,fill=KITblack!15] {};%

\node at (es) [text=nodeColor!100] {\small{$\sourceEvent$}};%
\node at (eo) [text=nodeColor!100] {\small{$\originEvent$}};%
\node at (es2) [text=nodeColor!100] {\small{$\sourceEvent'$}};%
\node at (eo2) [text=nodeColor!100] {\small{$\originEvent'$}};%
\node at (e1) [text=nodeColor!100] {\small{$\stopEvent_1$}};%
\node at (ed) [text=nodeColor!100] {\small{$\destinationEvent$}};%
\node at (e2) [text=nodeColor!100] {\small{$\stopEvent_2$}};%
\node at (et) [text=nodeColor!100] {\small{$\targetEvent$}};%
\node at (ed2) [text=nodeColor!100] {\small{$\destinationEvent'$}};%
\node at (et2) [text=nodeColor!100] {\small{$\targetEvent'$}};%

\node at (vs) [text=nodeColor!100] {\small{$\sourceVertex$}};%
\node at (vo) [text=nodeColor!100] {\small{$\originVertex$}};%
\node at (vd) [text=nodeColor!100] {\small{$\destinationVertex$}};%
\node at (vt) [text=nodeColor!100] {\small{$\targetVertex$}};%
\end{tikzpicture}%
	\caption{%
        An example network with a candidate~$\aCandidateJourney=\candidateSequence$ and various witness types.
		Stops and vertices in gray, stop events as colored nodes, transfer edges as thin gray lines, trips as thick colored lines.
    }%
    \label{fig:candidates}%
\end{figure}

A candidate is a journey of the form
\[\aCandidateJourney=\langle\langle\sourceVertex\rangle,\aNamedTripSegment{1}{\sourceIndex}{\originIndex},\langle\originVertex,\dots,\destinationVertex\rangle,\aNamedTripSegment{2}{\destinationIndex}{\targetIndex},\langle\targetVertex\rangle\rangle.\]
An example is shown in Figure~\ref{fig:candidates}.
We define the~\emph{source event}~$\sourceEvent:=\aTrip_1[\sourceIndex]$, \emph{origin event}~$\originEvent:=\aTrip_1[\originIndex]$, \emph{destination event}~$\destinationEvent:=\aTrip_2[\destinationIndex]$ and~\emph{target event}~$\targetEvent:=\aTrip_2[\targetIndex]$.
Corresponding to these are the~\emph{source vertex}~$\sourceVertex:=\astop(\sourceEvent)$, \emph{origin vertex}~$\originVertex:=\astop(\originEvent)$, \emph{destination vertex}~$\destinationVertex:=\astop(\destinationEvent)$ and~\emph{target vertex}~$\targetVertex:=\astop(\targetEvent)$.
In the following, we use a more compact notation:
We write~$[\stopEvent_a,\stopEvent_b]$ for a trip segment from~$\stopEvent_a$ to~$\stopEvent_b$, whereas a trip segment prefix starting at~$\stopEvent_a$ is simply denoted as~$\stopEvent_a$.
Since transfers are shortest paths, they are uniquely determined by their endpoints.
We therefore omit all transfers from the notation except for non-empty initial or final transfers, which are represented by the source or target vertex, respectively.
Thus, a candidate is notated as~$\aCandidateJourney=\candidateSequence$.
Its standard prefixes are the~\emph{source prefix}~$\sourcePrefix=\sourcePrefixSequence$, the~\emph{origin prefix}~$\originPrefix=\originPrefixSequence$, the~\emph{destination prefix}~$\destinationPrefix=\destinationPrefixSequence$, and~$\aCandidateJourney$ itself.
The non-standard prefix that ends at a vertex~$\aVertex$ is denoted as~$\originPrefix(\aVertex)=\originStopPrefixSequence{\aVertex}$.
In particular, we call~$\originPrefix(\originVertex)$ the~\emph{origin stop prefix}.
For a journey (prefix)~$\aJourney$, we denote by~$\stopEvents(\aJourney)$ the sequence of stop events in~$\aJourney$ at which a trip is entered or exited.
For the candidate~$\aCandidateJourney$, this sequence is~$\stopEvents(\aCandidateJourney)=\langle\sourceEvent,\originEvent,\destinationEvent,\targetEvent\rangle$.
Standard prefixes of~$\aCandidateJourney$ are uniquely identified by their stop event sequence.
The standard prefix that ends with a stop event~$\stopEvent\in\stopEvents(\aCandidateJourney)$ is called the~\emph{$\stopEvent$-prefix}.

\subsection{Best-Case and Virtual Delay Scenarios}
\label{sec:optimality-conditions:virtual}
To simplify proofs, we define a partial order~$\leqEval$ on delay scenarios.
For two delay scenarios~$\delayScenario^1,\delayScenario^2$, we write~$\delayScenario^1\leqEval\delayScenario^2$ if~$\departureDelay^1(\stopEvent)\geq\departureDelay^2(\stopEvent)$ and~$\arrivalDelay^1(\stopEvent)\leq\arrivalDelay^2(\stopEvent)$ holds for every stop event~$\stopEvent\in\stopEvents$.
Then~$\delayScenario^1$ is ``better'' than~$\delayScenario^2$ in the following sense:
a journey~$\aJourney$ that is feasible in~$\delayScenario^1$ for a departure time~$\departureTime$ is also feasible for~$\departureTime$ in~$\delayScenario^2$, and its arrival time in~$\delayScenario^1$ is not higher than in~$\delayScenario^2$.
Consider the global \emph{best-case} delay scenario~$\bestDelayScenario$, in which all arrivals have delay~0 and all departures have delay~\maxDelay.
The other extreme is the \emph{worst-case} scenario~$\worstDelayScenario$, which assumes that all departures are punctual and all arrivals have maximum delay.
Then~$\bestDelayScenario\leqEval\delayScenario\leqEval\worstDelayScenario$ holds for every delay scenario~$\delayScenario$.
For a journey prefix~$\prefix{\aJourney}$, the best-case scenario~$\bestDelayScenario(\prefix{\aJourney})$ assumes the best case for all stop events in~$\stopEvents(\prefix{\aJourney})$ and the worst case otherwise.

Consider a candidate~$\aCandidateJourney$ and a witness~$\aWitnessJourney$ that does not use any stop events in~$\stopEvents(\aCandidateJourney)$.
If~$\aWitnessJourney$ strongly dominates a prefix of~$\aCandidateJourney$ in~$\bestDelayScenario(\aCandidateJourney)$, then it does so in every delay scenario.
Thus, if we (wrongly) assume that witnesses do not share stop events with~$\aCandidateJourney$, we obtain the very simple condition that~$\aCandidateJourney$ is prefix-optimal in at least one delay scenario iff it is prefix-optimal in~$\bestDelayScenario(\aCandidateJourney)$.
We formalize this assumption by introducing virtual delay scenarios.
A~\emph{virtual delay scenario}~$\virtualDelayScenario=(\candidateDelayScenario,\witnessDelayScenario)$ consists of a \emph{candidate scenario}~$\candidateDelayScenario$ and a \emph{witness scenario}~$\witnessDelayScenario$.
A candidate prefix~$\prefix{\aCandidateJourney}$ is strongly dominated by a witness~$\aWitnessJourney$ in~$\virtualDelayScenario$ if~$\aWitnessJourney$ as evaluated in~$\witnessDelayScenario$ strongly dominates~$\prefix{\aCandidateJourney}$ as evaluated in~$\candidateDelayScenario$.
Thus, stop events that are shared between both journeys can act as if they were not shared by assuming different delays in~$\candidateDelayScenario$ and~$\witnessDelayScenario$.

We define another partial order~$\leqDom$ specifically for virtual delay scenarios.
Given two virtual delay scenarios~$\virtualDelayScenario_1=(\candidateDelayScenario_1,\witnessDelayScenario_1)$ and~$\virtualDelayScenario_2=(\candidateDelayScenario_2,\witnessDelayScenario_2)$, we write~$\virtualDelayScenario_1\leqDom\virtualDelayScenario_2$ if~$\candidateDelayScenario_1\leqEval\candidateDelayScenario_2$ and~$\witnessDelayScenario_1\geqEval\witnessDelayScenario_2$.
If a witness~$\aWitnessJourney$ strongly dominates a candidate prefix~$\prefix{\aCandidateJourney}$ in~$\virtualDelayScenario_1$, it also does so in~$\virtualDelayScenario_2$.
Thus, if~$\prefix{\aCandidateJourney}$ is prefix-optimal in~$\virtualDelayScenario_2$, it is also prefix-optimal in~$\virtualDelayScenario_1$.
Since each proper delay scenario~$\delayScenario$ has an equivalent virtual delay scenario~$(\delayScenario,\delayScenario)$, virtual delay scenarios are a superset of proper delay scenarios.
We therefore extend the definition of~$\leqDom$ to proper delay scenarios as well.

For a candidate prefix~$\prefix{\aCandidateJourney}$, the \emph{virtual best-case scenario}~$\virtualDelayScenario(\prefix{\aCandidateJourney}):=(\bestDelayScenario,\bestDelayScenario(\prefix{\aCandidateJourney}))$ assumes the best case for~$\prefix{\aCandidateJourney}$ and for all witness events that are shared with it, and the worst case for everything else.
For a full candidate~$\aCandidateJourney$, $\virtualDelayScenario(\aCandidateJourney)$ is equivalent to~$\bestDelayScenario(\aCandidateJourney)$.
For the empty prefix~$\prefix{\aCandidateJourney}=\langle\rangle$, $\virtualDelayScenario(\langle\rangle)=(\bestDelayScenario,\worstDelayScenario)$ assumes the best case for~$\aCandidateJourney$ and the worst case for all witnesses.
Theorem~\ref{th:app:virtual} shows that we can use~$\virtualDelayScenario(\langle\rangle)$ to obtain a simple characterization of prefix-optimal candidates.

\begin{theorem}
	\label{th:app:virtual}
	A candidate~$\aCandidateJourney$ is prefix-optimal in at least one virtual delay scenario iff it is prefix-optimal in~$\virtualDelayScenario(\langle\rangle)=(\bestDelayScenario,\worstDelayScenario)$.
\end{theorem}
\begin{proof}
	Let~$\virtualDelayScenario=(\candidateDelayScenario,\witnessDelayScenario)$ be a virtual delay scenario.
	With
	\begin{align*}
		\witnessDelayScenario&\leqEval\worstDelayScenario=\bestDelayScenario(\langle\rangle)\quad\text{and}\\ \candidateDelayScenario&\geqEval\bestDelayScenario,
	\end{align*}
	it follows that~$\virtualDelayScenario\geqDom\virtualDelayScenario(\langle\rangle)$.
	Hence, if~$\aCandidateJourney$ is prefix-optimal in~$\virtualDelayScenario$, it is also prefix-optimal in~$\virtualDelayScenario(\langle\rangle)$.
\end{proof}

\subsection{Shared Stop Events}
\label{sec:optimality-conditions:shared}
Theorem~\ref{th:app:virtual} implies that a straightforward adaptation of ULTRA that explores candidates in~$\bestDelayScenario$ and witnesses in~$\worstDelayScenario$ will generate a sufficient set of shortcuts.
However, it will be impractically large because many of these shortcuts are only required in a virtual delay scenario, but not in any proper delay scenarios.
To avoid this, we investigate the effects of shared stop events.

\subparagraph*{Hook Witnesses.}
In Figure~\ref{fig:candidates}, consider the witness~$\aWitnessJourney=\langle[\sourceEvent,\originEvent],[\destinationEvent',\targetEvent']\rangle$.
If the origin event~$\originEvent$ is delayed, this may cause~$\aWitnessJourney$ to miss its intermediate transfer, even as the candidate~$\aCandidateJourney=\candidateSequence$ remains feasible.
This shows that it is not sufficient to consider the best case for shared stop events.
In the following, we will show that~$\originEvent$ is the only stop event of~$\aCandidateJourney$ where this is an issue.
For this purpose, we introduce the concept of hook witnesses.
We call a witness~$\aWitnessJourney$ for a journey (prefix)~$\aJourney$ a~\emph{hook witness} if~$\stopEvents(\aWitnessJourney)$ can be divided into a prefix that is shared with~$\stopEvents(\aJourney)$ (the \emph{handle}) and a suffix that is not (the \emph{hook}).
In Figure~\ref{fig:candidates}, $\langle[\sourceEvent,\originEvent],\stopEvent_1\rangle$ is a hook witness for~$\destinationPrefix$ while~$\langle[\sourceEvent,\originEvent],[\stopEvent_1,\targetEvent]\rangle$ is a non-hook witness for~$\aCandidateJourney$.
Note that if~$\stopEvents(\aWitnessJourney)$ includes a stop event~$\stopEvent\in\stopEvents(\aJourney)$, it includes the entire~$\stopEvent$-prefix of~$\aJourney$.
Lemma~\ref{th:app:hook} shows that it is sufficient to consider hook witnesses, because every non-hook witness has an equivalent hook witness that replaces everything up to the last shared stop event~$\stopEvent$ with the $\stopEvent$-prefix of~$\aJourney$.

\begin{lemma}
	\label{th:app:hook}
	A candidate~$\aCandidateJourney$ is prefix-optimal iff no prefix of~$\aCandidateJourney$ is strongly dominated by a hook witness.
\end{lemma}
\begin{proof}
	Let~$\aWitnessJourney$ be a non-hook witness for a prefix~$\prefix{\aCandidateJourney}$ of~$\aCandidateJourney$.
	We construct a hook witness~$\aJourney'$ that strongly dominates~$\prefix{\aCandidateJourney}$ as follows:
	Since~$\aWitnessJourney$ is not a hook witness, it must share at least one stop event with~$\prefix{\aCandidateJourney}$.
	Let~$\stopEvent$ be the last shared stop event, and let~$\aCandidateJourney_1$ and~$\aWitnessJourney_1$ be the~$\stopEvent$-prefixes of~$\prefix{\aCandidateJourney}$ and~$\aWitnessJourney$, respectively.
	Then there is a suffix~$\aWitnessJourney_2$ such that~$\aWitnessJourney=\aWitnessJourney_1\circ\aWitnessJourney_2$.
	Because~$\aWitnessJourney_1$ and~$\aCandidateJourney_1$ end with the same stop event, $\aJourney':=\aCandidateJourney_1\circ\aWitnessJourney_2$ is a valid and feasible journey (prefix).
	By construction, $\aJourney'$ is a hook witness and strongly dominates~$\prefix{\aCandidateJourney}$.
\end{proof}

A simple observation about hook witnesses is that they never use the last stop event of the journey (prefix) they strongly dominate:

\begin{lemma}
	\label{th:app:last-event}
	Let~$\prefix{\aJourney}$ be a journey prefix with stop event sequence~$\stopEvents(\prefix{\aJourney})=\langle\stopEvent_1,\dots,\stopEvent_k\rangle$, $\delayScenario$ a delay scenario, and~$\aWitnessJourney$ a hook witness that strongly dominates~$\prefix{\aJourney}$ in~$\delayScenario$.
	Then~$\stopEvent_k\notin\stopEvents(\aWitnessJourney)$.
\end{lemma}
\begin{proof}   
	Assume that~$\stopEvent_k\in\stopEvents(\aWitnessJourney)$.
	Because~$\aWitnessJourney$ is a hook witness, this implies that the non-shared suffix of~$\aWitnessJourney$ is empty, and thus~$\aWitnessJourney=\prefix{\aJourney}$.
	This contradicts the fact that~$\aWitnessJourney$ strongly dominates~$\prefix{\aJourney}$.
\end{proof}

\subparagraph*{Parameterized Scenarios.}
Let~$\aCandidateJourney=\candidateSequence$ be a candidate and~$\aWitnessJourney$ a hook witness.
In order to evaluate whether~$\aWitnessJourney$ strongly dominates a prefix of~$\aCandidateJourney$ in any delay scenario, we can assume the best case for~$\sourceEvent$, $\destinationEvent$ and~$\targetEvent$:
If~$\aWitnessJourney$ includes~$\targetEvent$, then~$\aWitnessJourney$ is identical to~$\aCandidateJourney$, which cannot strongly dominate itself.
If it includes~$\destinationEvent$, it must also include~$\originEvent$ and thus have the same intermediate transfer as~$\aCandidateJourney$.
If it includes~$\sourceEvent$, then it has the same initial transfer.
Thus, changing the delay of~$\sourceEvent$ or~$\destinationEvent$ to make~$\aWitnessJourney$ infeasible will also make~$\aCandidateJourney$ infeasible.
We formalize this by defining the \emph{parameterized best-case scenario}~$\parameterizedBestDelayScenario(\prefix{\aCandidateJourney},\delay)$ for a candidate prefix~$\prefix{\aCandidateJourney}$ with origin event~$\originEvent$ and a delay~$\delay$.
It is identical to the best-case scenario~$\bestDelayScenario(\prefix{\aCandidateJourney})$, except that the arrival delay of~$\originEvent$ is~$\delay$.
Note that~$\bestDelayScenario(\prefix{\aCandidateJourney})=\parameterizedBestDelayScenario(\prefix{\aCandidateJourney},0)$.

It is only necessary to consider parameterized scenarios in which the candidate is feasible.
The \emph{slack}~$\slack(\originEvent,\destinationEvent)$ of an intermediate transfer between two stop events~$\originEvent$ and~$\destinationEvent$ is the waiting time at~$\astop(\destinationEvent)$ before~$\destinationEvent$ departs, assuming both events are punctual.
Note that the slack may be negative; in this case, the transfer is only feasible if~$\destinationEvent$ is sufficiently delayed.
The~\emph{feasibility limit}
\[\feasibilityLimit(\aCandidateJourney):=\min(0,\slack(\originEvent,\destinationEvent))+\maxDelay\]
of the candidate~$\aCandidateJourney$ is the maximum delay~$\delay\leq\maxDelay$ such that~$\aCandidateJourney$ is feasible in~$\parameterizedBestDelayScenario(\aCandidateJourney,\delay)$.
Lemma~\ref{th:app:all-best-case} shows that it is sufficient to consider parameterized scenarios~$\parameterizedBestDelayScenario(\aCandidateJourney,\delay)$ with~$\delay\in[0,\feasibilityLimit(\aCandidateJourney)]$.

\begin{lemma}
	\label{th:app:all-best-case}
	A candidate~$\aCandidateJourney=\candidateSequence$ is prefix-optimal in at least one delay scenario iff there is a~$\delay\in[0,\feasibilityLimit(\aCandidateJourney)]$ such that~$\aCandidateJourney$ is prefix-optimal in~$\parameterizedBestDelayScenario(\aCandidateJourney,\delay)$.
\end{lemma}
\begin{proof}
	Assume that for every~$\delay\in[0,\feasibilityLimit(\aCandidateJourney)]$, a prefix of~$\aCandidateJourney$ is strongly dominated by a witness in~$\parameterizedBestDelayScenario(\aCandidateJourney,\delay)$.
	Let~$\delayScenario$ be a delay scenario.
	By Lemma~\ref{th:app:hook}, there is a prefix~$\prefix{\aCandidateJourney}$ of~$\aCandidateJourney$ that is strongly dominated by a hook witness~$\aWitnessJourney$ in~$\delayScenario^\textsf{o}:=\parameterizedBestDelayScenario(\aCandidateJourney,\arrivalDelay(\originEvent))$.
	We show that~$\aWitnessJourney$ is feasible and strongly dominates~$\prefix{\aCandidateJourney}$ in~$\delayScenario$.
	Consider the delay scenario~$\delayScenario^\text{max}$ with
	\begin{align*}
		\departureDelay^\text{max}(\stopEvent)&=\min(\departureDelay^\textsf{o}(\stopEvent),\departureDelay(\stopEvent)),\\ \arrivalDelay^\text{max}(\stopEvent)&=\max(\arrivalDelay^\textsf{o}(\stopEvent),\arrivalDelay(\stopEvent))
	\end{align*}
	for each stop event~$\stopEvent\in\stopEvents$.
	This scenario differs from~$\delayScenario^\textsf{o}$ only in the arrival delay of~$\targetEvent$ and in the departure delays of~$\sourceEvent$ and~$\destinationEvent$.
	By Lemma~\ref{th:app:last-event}, $\aWitnessJourney$ cannot use~$\targetEvent$.
	The departure delays of~$\sourceEvent$ and~$\destinationEvent$, as well as the arrival delay of~$\originEvent$, are identical in~$\delayScenario$ and~$\delayScenario^\text{max}$.
	Accordingly, the departure and intermediate transfer of~$\aCandidateJourney$ are feasible in~$\delayScenario^\text{max}$.        
	If~$\aWitnessJourney$ contains an intermediate transfer between some stop event~$\stopEvent'_\originIndex$ and~$\destinationEvent$, then~$\stopEvent'_\originIndex=\originEvent$ because~$\aWitnessJourney$ is a hook witness.
	Thus, the intermediate transfer is feasible in~$\delayScenario^\text{max}$.
	Likewise, if~$\aWitnessJourney$ uses~$\sourceEvent$, then its departure is feasible in~$\delayScenario^\text{max}$.
	Therefore, $\aWitnessJourney$ is feasible in~$\delayScenario^\text{max}$ and its arrival time remains unchanged from~$\delayScenario^\textsf{o}$.    
	Accordingly, $\aWitnessJourney$ strongly dominates~$\prefix{\aCandidateJourney}$ in the virtual delay scenario~$(\delayScenario^\textsf{o},\delayScenario^\text{max})$.
	Since~$\delayScenario^\text{max}\geqEval\delayScenario^\textsf{o}$ and~$\delayScenario\leqEval\delayScenario^\text{max}$, this is still the case in~$(\delayScenario^\text{max},\delayScenario)\geqDom(\delayScenario^\textsf{o},\delayScenario^\text{max})$.
	Since~$\delayScenario$ and~$\delayScenario^\text{max}$ are equivalent for~$\aCandidateJourney$, this implies that $\aWitnessJourney$ strongly dominates~$\prefix{\aCandidateJourney}$ in~$(\delayScenario,\delayScenario)$, which is equivalent to~$\delayScenario$.
\end{proof}

\subparagraph*{Hook Witness Classification.}
\begin{figure}
	\centering
	\begin{tabular}{c@{\hskip 1cm}c}
		\begin{subfigure}{0.4\textwidth}
			\centering%
			\begin{tikzpicture}
\draw [draw=none] (-0.35,-0.35) rectangle (5.35,1.35);

\node (es) at (0.00, 0.00) {};%
\node (es2) at (1.00, 1.00) {};%
\node (eo) at (2.00, 0.00) {};%
\node (ed) at (3.00, 0.00) {};%
\node (et) at (5.00, 0.00) {};%

\draw [edge, rounded corners=15]  (es) -- (0.00,1.00) -- (es2);
\draw [edgeFaded]  (eo) -- (ed);

\node [stop,fit=(es)] {};%
\node [stop,fit=(es2)] {};%
\node [stopFaded,fit=(eo)] {};%
\node [stopFaded,fit=(ed)] {};%
\node [stopFaded,fit=(et)] {};%

\draw [tripColor1,routeNoArrow] (es2) -- (es);
\draw [tripColor1!50,route] (es) -- (eo);
\draw [tripColor2!50,route] (ed) -- (et);

\node (es_v) at (es) [vertex,draw=tripColor1,fill=tripColor1!15] {\gs};%
\node (es2_v) at (es2) [vertex,draw=tripColor1,fill=tripColor1!15] {\gs};%
\node (eo_v) at (eo) [vertex,draw=tripColor1!50,fill=tripColor1!15] {\gs};%
\node (ed_v) at (ed) [vertex,draw=tripColor2!50,fill=tripColor2!15] {\gs};%
\node (et_v) at (et) [vertex,draw=tripColor2!50,fill=tripColor2!15] {\gs};%

\node at (es) [text=nodeColor] {\small{$\sourceEvent$}};%
\node at (eo) [text=nodeColorFaded] {\small{$\originEvent$}};%
\node at (ed) [text=nodeColorFaded] {\small{$\destinationEvent$}};%
\node at (et) [text=nodeColorFaded] {\small{$\targetEvent$}};%
\end{tikzpicture}%
			\caption{A $(\bot,\sourceEvent)$-witness (full).}%
			\label{fig:witnessTypes:dot-source}%
		\end{subfigure}&
		\begin{subfigure}{0.4\textwidth}
			\centering
			\begin{tikzpicture}
\draw [draw=none] (-0.35,-0.35) rectangle (5.35,1.35);

\node (es) at (0.00, 0.00) {};%
\node (eo) at (2.00, 0.00) {};%
\node (es2) at (1.00, 1.00) {};%
\node (eo2) at (2.00, 1.00) {};%
\node (ed) at (3.00, 0.00) {};%
\node (et) at (5.00, 0.00) {};%

\draw [edge]  (eo) -- (ed);
\draw [edge, rounded corners=15]  (es) -- (0.00, 1.00) -- (es2);
\draw [edge, rounded corners=15]  (eo2) -- (3.00, 1.00) -- (ed);

\node [stop,fit=(es)] {};%
\node [stop,fit=(eo)] {};%
\node [stop,fit=(es2)] {};%
\node [stop,fit=(eo2)] {};%
\node [stop,fit=(ed)] {};%
\node [stopFaded,fit=(et)] {};%

\draw [tripColor1,route] (es) -- (eo);
\draw [tripColor2!50,route] (ed) -- (et);
\draw [tripColor3,route] (es2) -- (eo2);

\node (es_v) at (es) [vertex,draw=tripColor1,fill=tripColor1!15] {\gs};%
\node (eo_v) at (eo) [vertex,draw=tripColor1,fill=tripColor1!15] {\gs};%
\node (ed_v) at (ed) [vertex,draw=tripColor2!50,fill=tripColor2!15] {\gs};%
\node (et_v) at (et) [vertex,draw=tripColor2!50,fill=tripColor2!15] {\gs};%
\node (es2_v) at (es2) [vertex,draw=tripColor3,fill=tripColor3!15] {\gs};%
\node (eo2_v) at (eo2) [vertex,draw=tripColor3,fill=tripColor3!15] {\gs};%

\node at (es) [text=nodeColor!100] {\small{$\sourceEvent$}};%
\node at (eo) [text=nodeColor!100] {\small{$\originEvent$}};%
\node at (ed) [text=nodeColorFaded] {\small{$\destinationEvent$}};%
\node at (et) [text=nodeColorFaded] {\small{$\targetEvent$}};%
\end{tikzpicture}%
			\caption{A $(\bot,\originEvent)$-witness (join).}%
			\label{fig:witnessTypes:dot-origin}%
		\end{subfigure}\\
		\newline
		\begin{subfigure}{0.4\textwidth}
            \vspace{5pt}
			\centering
			\begin{tikzpicture}
\draw [draw=none] (-0.35,-0.35) rectangle (5.35,1.35);

\node (es) at (0.00, 0.00) {};%
\node (eo) at (2.00, 0.00) {};%
\node (es2) at (1.00, 1.00) {};%
\node (eo2) at (2.00, 1.00) {};%
\node (ed) at (3.00, 0.00) {};%
\node (ed2) at (3.00, 1.00) {};%
\node (et) at (5.00, 0.00) {};%

\draw [edge]  (eo) -- (ed);
\draw [edge, rounded corners=15]  (es) -- (0.00, 1.00) -- (es2);
\draw [edge]  (eo2) -- (ed2);

\node [stop,fit=(es)] {};%
\node [stop,fit=(eo)] {};%
\node [stop,fit=(es2)] {};%
\node [stop,fit=(eo2)] {};%
\node [stop,fit=(ed)] {};%
\node [stop,fit=(ed2)] {};%
\node [stopFaded,fit=(et)] {};%

\draw [tripColor1,route] (es) -- (eo);
\draw [tripColor2,routeNoArrow] (ed2) -- (ed);
\draw [tripColor2!50,route] (ed) -- (et);
\draw [tripColor3,route] (es2) -- (eo2);

\node (es_v) at (es) [vertex,draw=tripColor1,fill=tripColor1!15] {\gs};%
\node (eo_v) at (eo) [vertex,draw=tripColor1,fill=tripColor1!15] {\gs};%
\node (ed2_v) at (ed2) [vertex,draw=tripColor2,fill=tripColor2!15] {\gs};%
\node (ed_v) at (ed) [vertex,draw=tripColor2,fill=tripColor2!15] {\gs};%
\node (et_v) at (et) [vertex,draw=tripColor2!50,fill=tripColor2!15] {\gs};%
\node (es2_v) at (es2) [vertex,draw=tripColor3,fill=tripColor3!15] {\gs};%
\node (eo2_v) at (eo2) [vertex,draw=tripColor3,fill=tripColor3!15] {\gs};%

\node at (es) [text=nodeColor!100] {\small{$\sourceEvent$}};%
\node at (eo) [text=nodeColor!100] {\small{$\originEvent$}};%
\node at (ed) [text=nodeColor!100] {\small{$\destinationEvent$}};%
\node at (et) [text=nodeColorFaded] {\small{$\targetEvent$}};%
\end{tikzpicture}%
			\caption{A $(\bot,\destinationEvent)$-witness (full).}%
			\label{fig:witnessTypes:dot-destination}%
		\end{subfigure}&
		\begin{subfigure}{0.4\textwidth}
            \vspace{5pt}
			\centering
			\begin{tikzpicture}
\draw [draw=none] (-0.35,-0.35) rectangle (5.35,1.35);

\node (es) at (0.00, 0.00) {};%
\node (eo) at (2.00, 0.00) {};%
\node (es2) at (1.00, 1.00) {};%
\node (eo2) at (2.00, 1.00) {};%
\node (ed) at (3.00, 0.00) {};%
\node (ed2) at (3.00, 1.00) {};%
\node (et) at (5.00, 0.00) {};%
\node (et2) at (4.00, 1.00) {};%

\draw [edge]  (eo) -- (ed);
\draw [edge, rounded corners=15]  (es) -- (0.00, 1.00) -- (es2);
\draw [edge]  (eo2) -- (ed2);
\draw [edge, rounded corners=15]  (et2) -- (5.00, 1.00) -- (et);

\node [stop,fit=(es)] {};%
\node [stop,fit=(eo)] {};%
\node [stop,fit=(es2)] {};%
\node [stop,fit=(eo2)] {};%
\node [stop,fit=(ed)] {};%
\node [stop,fit=(ed2)] {};%
\node [stop,fit=(et)] {};%
\node [stop,fit=(et2)] {};%

\draw [tripColor1,route] (es) -- (eo);
\draw [tripColor2,route] (ed) -- (et);
\draw [tripColor3,route] (es2) -- (eo2);
\draw [tripColor4,route] (ed2) -- (et2);

\node (es_v) at (es) [vertex,draw=tripColor1,fill=tripColor1!15] {\gs};%
\node (eo_v) at (eo) [vertex,draw=tripColor1,fill=tripColor1!15] {\gs};%
\node (ed_v) at (ed) [vertex,draw=tripColor2,fill=tripColor2!15] {\gs};%
\node (et_v) at (et) [vertex,draw=tripColor2,fill=tripColor2!15] {\gs};%
\node (es2_v) at (es2) [vertex,draw=tripColor3,fill=tripColor3!15] {\gs};%
\node (eo2_v) at (eo2) [vertex,draw=tripColor3,fill=tripColor3!15] {\gs};%
\node (ed2_v) at (ed2) [vertex,draw=tripColor4,fill=tripColor4!15] {\gs};%
\node (et2_v) at (et2) [vertex,draw=tripColor4,fill=tripColor4!15] {\gs};%

\node at (es) [text=nodeColor!100] {\small{$\sourceEvent$}};%
\node at (eo) [text=nodeColor!100] {\small{$\originEvent$}};%
\node at (ed) [text=nodeColor!100] {\small{$\destinationEvent$}};%
\node at (et) [text=nodeColor!100] {\small{$\targetEvent$}};%
\end{tikzpicture}%
			\caption{A $(\bot,\targetEvent)$-witness (full).}%
			\label{fig:witnessTypes:dot-target}%
		\end{subfigure}\\
		\newline
		\begin{subfigure}{0.4\textwidth}
            \vspace{5pt}
			\centering
			\begin{tikzpicture}
\draw [draw=none] (-0.35,-0.35) rectangle (5.35,1.35);

\node (es) at (0.00, 0.00) {};%
\node (eo) at (2.00, 0.00) {};%
\node (eo2) at (1.00, 0.00) {};%
\node (ed) at (3.00, 0.00) {};%
\node (et) at (5.00, 0.00) {};%

\draw [edge]  (eo) -- (ed);
\draw [edge, rounded corners=15]  (eo2) -- (1.00, 1.00) -- (3.00, 1.00) -- (ed);

\node [stop,fit=(es)] {};%
\node [stop,fit=(eo)] {};%
\node [stop,fit=(eo2)] {};%
\node [stop,fit=(ed)] {};%
\node [stopFaded,fit=(et)] {};%

\draw [tripColor1,route] (es) -- (eo2) -- (eo);
\draw [tripColor2!50,route] (ed) -- (et);

\node (es_v) at (es) [vertex,draw=tripColor1,fill=tripColor1!15] {\gs};%
\node (eo_v) at (eo) [vertex,draw=tripColor1,fill=tripColor1!15] {\gs};%
\node (ed_v) at (ed) [vertex,draw=tripColor2!50,fill=tripColor2!15] {\gs};%
\node (et_v) at (et) [vertex,draw=tripColor2!50,fill=tripColor2!15] {\gs};%
\node (eo2_v) at (eo2) [vertex,draw=tripColor1,fill=tripColor1!15] {\gs};%

\node at (es) [text=nodeColor!100] {\small{$\sourceEvent$}};%
\node at (eo) [text=nodeColor!100] {\small{$\originEvent$}};%
\node at (ed) [text=nodeColorFaded] {\small{$\destinationEvent$}};%
\node at (et) [text=nodeColorFaded] {\small{$\targetEvent$}};%
\end{tikzpicture}%
			\caption{A $(\sourceEvent,\originEvent)$-witness (join).}%
			\label{fig:witnessTypes:source-origin}%
		\end{subfigure}&
		\begin{subfigure}{0.4\textwidth}
            \vspace{5pt}
			\centering
			\begin{tikzpicture}
\draw [draw=none] (-0.35,-0.35) rectangle (5.35,1.35);

\node (es) at (0.00, 0.00) {};%
\node (eo) at (2.00, 0.00) {};%
\node (eo2) at (1.00, 0.00) {};%
\node (ed) at (3.00, 0.00) {};%
\node (ed2) at (3.00, 1.00) {};%
\node (et) at (5.00, 0.00) {};%

\draw [edge]  (eo) -- (ed);
\draw [edge, rounded corners=15]  (eo2) -- (1.00, 1.00) -- (ed2);

\node [stop,fit=(es)] {};%
\node [stop,fit=(eo)] {};%
\node [stop,fit=(eo2)] {};%
\node [stop,fit=(ed)] {};%
\node [stop,fit=(ed2)] {};%
\node [stopFaded,fit=(et)] {};%

\draw [tripColor1,route] (es) -- (eo2) -- (eo);
\draw [tripColor2,routeNoArrow] (ed2) -- (ed);
\draw [tripColor2!50,route] (ed) -- (et);

\node (es_v) at (es) [vertex,draw=tripColor1,fill=tripColor1!15] {\gs};%
\node (eo_v) at (eo) [vertex,draw=tripColor1,fill=tripColor1!15] {\gs};%
\node (ed2_v) at (ed2) [vertex,draw=tripColor2,fill=tripColor2!15] {\gs};%
\node (ed_v) at (ed) [vertex,draw=tripColor2,fill=tripColor2!15] {\gs};%
\node (et_v) at (et) [vertex,draw=tripColor2!50,fill=tripColor2!15] {\gs};%
\node (eo2_v) at (eo2) [vertex,draw=tripColor1,fill=tripColor1!15] {\gs};%

\node at (es) [text=nodeColor!100] {\small{$\sourceEvent$}};%
\node at (eo) [text=nodeColor!100] {\small{$\originEvent$}};%
\node at (ed) [text=nodeColor!100] {\small{$\destinationEvent$}};%
\node at (et) [text=nodeColorFaded] {\small{$\targetEvent$}};%
\end{tikzpicture}%
			\caption{A $(\sourceEvent,\destinationEvent)$-witness (full).}%
			\label{fig:witnessTypes:source-destination}%
		\end{subfigure}\\
		\newline
		\begin{subfigure}{0.4\textwidth}
            \vspace{5pt}
			\centering
			\begin{tikzpicture}
\draw [draw=none] (-0.35,-0.35) rectangle (5.35,1.35);

\node (es) at (0.00, 0.00) {};%
\node (eo) at (2.00, 0.00) {};%
\node (eo2) at (1.00, 0.00) {};%
\node (ed) at (3.00, 0.00) {};%
\node (ed2) at (3.00, 1.00) {};%
\node (et) at (5.00, 0.00) {};%
\node (et2) at (4.00, 1.00) {};%

\draw [edge]  (eo) -- (ed);
\draw [edge, rounded corners=15]  (eo2) -- (1.00, 1.00) -- (ed2);
\draw [edge, rounded corners=15]  (et2) -- (5.00, 1.00) -- (et);

\node [stop,fit=(es)] {};%
\node [stop,fit=(eo)] {};%
\node [stop,fit=(eo2)] {};%
\node [stop,fit=(ed)] {};%
\node [stop,fit=(ed2)] {};%
\node [stop,fit=(et)] {};%
\node [stop,fit=(et2)] {};%

\draw [tripColor1,route] (es) -- (eo2) -- (eo);
\draw [tripColor2,route] (ed) -- (et);
\draw [tripColor4,route] (ed2) -- (et2);

\node (es_v) at (es) [vertex,draw=tripColor1,fill=tripColor1!15] {\gs};%
\node (eo_v) at (eo) [vertex,draw=tripColor1,fill=tripColor1!15] {\gs};%
\node (ed_v) at (ed) [vertex,draw=tripColor2,fill=tripColor2!15] {\gs};%
\node (et_v) at (et) [vertex,draw=tripColor2,fill=tripColor2!15] {\gs};%
\node (eo2_v) at (eo2) [vertex,draw=tripColor1,fill=tripColor1!15] {\gs};%
\node (ed2_v) at (ed2) [vertex,draw=tripColor4,fill=tripColor4!15] {\gs};%
\node (et2_v) at (et2) [vertex,draw=tripColor4,fill=tripColor4!15] {\gs};%

\node at (es) [text=nodeColor!100] {\small{$\sourceEvent$}};%
\node at (eo) [text=nodeColor!100] {\small{$\originEvent$}};%
\node at (ed) [text=nodeColor!100] {\small{$\destinationEvent$}};%
\node at (et) [text=nodeColor!100] {\small{$\targetEvent$}};%
\end{tikzpicture}%
			\caption{A $(\sourceEvent,\targetEvent)$-witness (full).}%
			\label{fig:witnessTypes:source-target}%
		\end{subfigure}&
		\begin{subfigure}{0.4\textwidth}
            \vspace{5pt}
			\centering
			\begin{tikzpicture}
\draw [draw=none] (-0.35,-0.35) rectangle (5.35,1.35);

\node (es) at (0.00, 0.00) {};%
\node (eo) at (2.00, 0.00) {};%
\node (ed) at (3.00, 0.00) {};%
\node (ed2) at (3.00, 1.00) {};%
\node (et) at (5.00, 0.00) {};%

\draw [edge]  (eo) -- (ed);
\draw [edge, rounded corners=15]  (eo) -- (2.00, 1.00) -- (ed2);

\node [stop,fit=(es)] {};%
\node [stop,fit=(eo)] {};%
\node [stop,fit=(ed)] {};%
\node [stop,fit=(ed2)] {};%
\node [stopFaded,fit=(et)] {};%

\draw [tripColor1,route] (es) -- (eo);
\draw [tripColor2,routeNoArrow] (ed2) -- (ed);
\draw [tripColor2!50,route] (ed) -- (et);

\node (es_v) at (es) [vertex,draw=tripColor1,fill=tripColor1!15] {\gs};%
\node (eo_v) at (eo) [vertex,draw=tripColor1,fill=tripColor1!15] {\gs};%
\node (ed2_v) at (ed2) [vertex,draw=tripColor2,fill=tripColor2!15] {\gs};%
\node (ed_v) at (ed) [vertex,draw=tripColor2,fill=tripColor2!15] {\gs};%
\node (et_v) at (et) [vertex,draw=tripColor2!50,fill=tripColor2!15] {\gs};%

\node at (es) [text=nodeColor!100] {\small{$\sourceEvent$}};%
\node at (eo) [text=nodeColor!100] {\small{$\originEvent$}};%
\node at (ed) [text=nodeColor!100] {\small{$\destinationEvent$}};%
\node at (et) [text=nodeColorFaded] {\small{$\targetEvent$}};%
\end{tikzpicture}%
			\caption{A $(\originEvent,\destinationEvent)$-witness (split).}%
			\label{fig:witnessTypes:origin-destination}%
		\end{subfigure}\\
		\begin{subfigure}{0.4\textwidth}
            \vspace{5pt}
			\centering
			\begin{tikzpicture}
\draw [draw=none] (-0.35,-0.35) rectangle (5.35,1.35);

\node (es) at (0.00, 0.00) {};%
\node (eo) at (2.00, 0.00) {};%
\node (ed) at (3.00, 0.00) {};%
\node (ed2) at (3.00, 1.00) {};%
\node (et) at (5.00, 0.00) {};%
\node (et2) at (4.00, 1.00) {};%

\draw [edge]  (eo) -- (ed);
\draw [edge, rounded corners=15]  (eo) -- (2.00, 1.00) -- (ed2);
\draw [edge, rounded corners=15]  (et2) -- (5.00, 1.00) -- (et);

\node [stop,fit=(es)] {};%
\node [stop,fit=(eo)] {};%
\node [stop,fit=(ed)] {};%
\node [stop,fit=(ed2)] {};%
\node [stop,fit=(et)] {};%
\node [stop,fit=(et2)] {};%

\draw [tripColor1,route] (es) -- (eo);
\draw [tripColor2,route] (ed) -- (et);
\draw [tripColor4,route] (ed2) -- (et2);

\node (es_v) at (es) [vertex,draw=tripColor1,fill=tripColor1!15] {\gs};%
\node (eo_v) at (eo) [vertex,draw=tripColor1,fill=tripColor1!15] {\gs};%
\node (ed_v) at (ed) [vertex,draw=tripColor2,fill=tripColor2!15] {\gs};%
\node (et_v) at (et) [vertex,draw=tripColor2,fill=tripColor2!15] {\gs};%
\node (ed2_v) at (ed2) [vertex,draw=tripColor4,fill=tripColor4!15] {\gs};%
\node (et2_v) at (et2) [vertex,draw=tripColor4,fill=tripColor4!15] {\gs};%

\node at (es) [text=nodeColor!100] {\small{$\sourceEvent$}};%
\node at (eo) [text=nodeColor!100] {\small{$\originEvent$}};%
\node at (ed) [text=nodeColor!100] {\small{$\destinationEvent$}};%
\node at (et) [text=nodeColor!100] {\small{$\targetEvent$}};%
\end{tikzpicture}%
			\caption{A $(\originEvent,\targetEvent)$-witness (split).}%
			\label{fig:witnessTypes:origin-target}%
		\end{subfigure}&
		\begin{subfigure}{0.4\textwidth}
            \vspace{5pt}
			\centering
			\begin{tikzpicture}
\draw [draw=none] (-0.35,-0.35) rectangle (5.35,1.35);

\node (es) at (0.00, 0.00) {};%
\node (eo) at (2.00, 0.00) {};%
\node (ed) at (3.00, 0.00) {};%
\node (et) at (5.00, 0.00) {};%
\node (et2) at (4.00, 0.00) {};%

\draw [edge]  (eo) -- (ed);
\draw [edge, rounded corners=15]  (et2) -- (4.00, 1.00) -- (5.00, 1.00) -- (et);

\node [stop,fit=(es)] {};%
\node [stop,fit=(eo)] {};%
\node [stop,fit=(ed)] {};%
\node [stop,fit=(et)] {};%
\node [stop,fit=(et2)] {};%

\draw [tripColor1,route] (es) -- (eo);
\draw [tripColor2,route] (ed) -- (et2) -- (et);

\node (es_v) at (es) [vertex,draw=tripColor1,fill=tripColor1!15] {\gs};%
\node (eo_v) at (eo) [vertex,draw=tripColor1,fill=tripColor1!15] {\gs};%
\node (ed_v) at (ed) [vertex,draw=tripColor2,fill=tripColor2!15] {\gs};%
\node (et_v) at (et) [vertex,draw=tripColor2,fill=tripColor2!15] {\gs};%
\node (et2_v) at (et2) [vertex,draw=tripColor2,fill=tripColor2!15] {\gs};%

\node at (es) [text=nodeColor!100] {\small{$\sourceEvent$}};%
\node at (eo) [text=nodeColor!100] {\small{$\originEvent$}};%
\node at (ed) [text=nodeColor!100] {\small{$\destinationEvent$}};%
\node at (et) [text=nodeColor!100] {\small{$\targetEvent$}};%
\end{tikzpicture}%
			\caption{A $(\destinationEvent,\targetEvent)$-witness (full).}%
			\label{fig:witnessTypes:destination-target}%
		\end{subfigure}\\
	\end{tabular}
	\caption{
		Examples of the different hook witness types for a candidate~$\aCandidateJourney=\candidateSequence$.
		An~$(\stopEvent_a,\stopEvent_b)$-witness is a hook witness for the~$\stopEvent_b$-prefix of~$\aCandidateJourney$ whose shared prefix is the~$\stopEvent_a$-prefix of~$\aCandidateJourney$.
		If the shared prefix is empty, we write~$\stopEvent_a=\bot$.
        The suffix that follows after the witness rejoins the candidate is grayed out.
		Also listed are the types according to the join/split/full classification.
		Witnesses of the form~$(\cdot,\originEvent)$ are join witnesses, whereas witnesses of the form~$(\originEvent,\cdot)$ are split witnesses.
		All other hook witnesses are full witnesses.
	}
	\label{fig:witnessTypes}
\end{figure}

Lemmas~\ref{th:app:hook} and~\ref{th:app:all-best-case} significantly narrow the scope of the shortcut computation problem.
To determine whether a candidate~$\aCandidateJourney$ is prefix-optimal in at least one delay scenario, it is enough to consider hook witnesses and to keep the delays of all stop events except for the origin event~$\originEvent$ fixed.
To make further progress, we examine the effect that the delay of~$\originEvent$ has on different types of hook witnesses.
For this purpose, we introduce the following classification:
Given a candidate~$\aCandidateJourney=\candidateSequence$ and stop events~$\stopEvent_a,\stopEvent_b\in\stopEvents(\aCandidateJourney)$, we call a (partial) journey~$\aWitnessJourney$ an~$(\stopEvent_a,\stopEvent_b)$-witness if it is a hook witness for the~$\stopEvent_b$-prefix of~$\aCandidateJourney$ whose handle is the~$\stopEvent_a$-prefix of~$\aCandidateJourney$.
If~$\aWitnessJourney$ does not share any stop events with~$\aCandidateJourney$, we call it a~$(\bot,\stopEvent_b)$-witness.
Overall, this yields ten different hook witness types, which are illustrated in Figure~\ref{fig:witnessTypes}.

Consider a parameterized delay scenario~$\parameterizedBestDelayScenario(\aCandidateJourney,\delay)$ with origin delay~$\delay$.
The arrival time of the origin prefix~$\originPrefix=\originPrefixSequence$ increases with~$\delay$, whereas the arrival time of hook witnesses for~$\originPrefix$ is unaffected.
Accordingly, they only strongly dominate~$\aCandidateJourney$ if~$\delay$ is high enough.
We call these witnesses, which have the type~$(\bot,\originEvent)$ or~$(\sourceEvent,\originEvent)$, \emph{join} witnesses.
\emph{Split} witnesses are hook witnesses of the types~$(\originEvent,\destinationEvent)$ or~$(\originEvent,\targetEvent)$, whose handle ends with~$\originEvent$.
If~$\delay$ is too high, their intermediate transfer becomes infeasible.
All other hook witnesses are not affected by~$\delay$ and are therefore called~\emph{full witnesses}.
Lemma~\ref{th:app:group-dominance} characterizes the effect of each hook witness group on the values of~$\delay$ for which~$\aCandidateJourney$ is prefix-optimal:
If there is a full witness that strongly dominates~$\aCandidateJourney$ for any value of~$\delay$, then it does so for all values and~$\aCandidateJourney$ is not needed.
Otherwise, join witnesses establish an upper bound on~$\delay$, whereas split witnesses establish a lower bound.

\begin{lemma}
	\label{th:app:group-dominance}
	Let~$\aCandidateJourney=\candidateSequence$ be a candidate and~$\delay\in[0,\feasibilityLimit(\aCandidateJourney)]$ a delay such that a prefix~$\prefix{\aCandidateJourney}$ of~$\aCandidateJourney$ is strongly dominated by a witness~$\aWitnessJourney$ in~$\parameterizedBestDelayScenario(\aCandidateJourney,\delay)$.
	Then~$\aWitnessJourney$ strongly dominates~$\prefix{\aCandidateJourney}$ in every delay scenario~$\parameterizedBestDelayScenario(\aCandidateJourney,\delay')$ with
	\[\delay' \in\begin{cases}
		[0,\feasibilityLimit(\aCandidateJourney)] & \text{ if }\aWitnessJourney \text{ is a full witness,}\\
		[\delay,\feasibilityLimit(\aCandidateJourney)] & \text{ if }\aWitnessJourney \text{ is a join witness,} \\
		[0,\delay] & \text{ if }\aWitnessJourney \text{ is a split witness.} \\
	\end{cases}\]
\end{lemma}
\begin{proof}
	W.l.o.g.,~we assume that~$\prefix{\aCandidateJourney}$ is a standard prefix of~$\aCandidateJourney$.
	We show that~$\aWitnessJourney$ is feasible and strongly dominates~$\prefix{\aCandidateJourney}$ in~$\parameterizedBestDelayScenario(\aCandidateJourney,\delay')$.
	Only the delay of~$\originEvent$ differs between the two scenarios.
	This can have two effects:
	\begin{enumerate}
		\item If~$\delay'>\delay$ and~$\aWitnessJourney$ includes an intermediate transfer from~$\originEvent$ to some other stop event~$\stopEvent'_\destinationIndex$, this intermediate transfer can become infeasible.
		If~$\stopEvent'_\destinationIndex=\destinationEvent$, the intermediate transfer is the same as that of~$\aCandidateJourney$, which is feasible in~$\parameterizedBestDelayScenario(\aCandidateJourney,\delay')$.
		Otherwise, $\aWitnessJourney$ is a split witness, which contradicts~$\delay'>\delay$.
		\item If~$\delay'<\delay$ and~$\prefix{\aCandidateJourney}$ is the origin prefix~$\originPrefixSequence$, then the arrival time of~$\prefix{\aCandidateJourney}$ decreases in~$\parameterizedBestDelayScenario(\aCandidateJourney,\delay')$.
		If~$\aWitnessJourney$ does not use~$\originEvent$ as well, it may no longer strongly dominate~$\prefix{\aCandidateJourney}$.
		However, in this case~$\aWitnessJourney$ is a join witness, which contradicts~$\delay'<\delay$.
	\end{enumerate}
\end{proof}

\subparagraph*{Final Characterization.}
\begin{figure}
	\centering
	\begin{tikzpicture}
	\node (es) at (0.00, 0.00) {};%
	\node (eo) at (4.00, 0.00) {};%
	\node (es2) at (1.00, 1.25) {};%
	\node (eo2) at (4.00, 1.25) {};%
	\node (ed) at (6.00, 0.00) {};%
	\node (ed2) at (6.00, -1.25) {};%
	\node (et) at (10.00, 0.00) {};%
	\node (et2) at (9.00, -1.25) {};%
	
	\draw [edge] (eo) -- (ed) node [midway, anchor=south] {$1$};
	\draw [edge, rounded corners=15]  (es) -- (0.00, 1.25) -- (es2) node [midway, anchor=south] {$1$};
	\draw [edge, rounded corners=15]  (eo2) -- (6.00, 1.25) -- (ed) node [midway, anchor=west] {$1$};
	\draw [edge, rounded corners=15]  (eo) -- (4.00, -1.25) -- (ed2) node [midway, anchor=south] {$1$};
	\draw [edge, rounded corners=15]  (et2) -- (10.00, -1.25) -- (et) node [midway, anchor=west] {$1$};
	
	\node [stop,fit=(es)] {};%
	\node [stop,fit=(eo)] {};%
	\node [stop,fit=(es2)] {};%
	\node [stop,fit=(eo2)] {};%
	\node [stop,fit=(ed)] {};%
	\node [stop,fit=(ed2)] {};%
	\node [stop,fit=(et)] {};%
	\node [stop,fit=(et2)] {};%
	
	\draw [tripColor1,route] (es) -- (eo);
	\draw [tripColor2,route] (ed) -- (et);
	\draw [tripColor3,route] (es2) -- (eo2);
	\draw [tripColor4,route] (ed2) -- (et2);
	
	\node [align=left,text=tripColor1] at (2.00, 0.70) {$0\rightarrow20$};%
	\node [align=left,text=tripColor2] at (8.00, 0.70) {$19\rightarrow50$};%
	\node [align=left,text=tripColor3] at (2.50, 1.95) {$6\rightarrow17$};%
	\node [align=left,text=tripColor4] at (7.50, -0.55) {$24\rightarrow40$};%
	\node [align=left,text=tripColor1] at (2.00, 0.30) {$(5\rightarrow20+\delay)$};%
	\node [align=left,text=tripColor2] at (8.00, 0.30) {$(24\rightarrow50)$};%
	\node [align=left,text=tripColor3] at (2.50, 1.55) {$(6\rightarrow22)$};%
	\node [align=left,text=tripColor4] at (7.50, -0.95) {$(24\rightarrow45)$};%
		
	\node (es_v) at (es) [vertex,draw=tripColor1,fill=tripColor1!15] {\gs};%
	\node (eo_v) at (eo) [vertex,draw=tripColor1,fill=tripColor1!15] {\gs};%
	\node (ed_v) at (ed) [vertex,draw=tripColor2,fill=tripColor2!15] {\gs};%
	\node (et_v) at (et) [vertex,draw=tripColor2,fill=tripColor2!15] {\gs};%
	\node (es2_v) at (es2) [vertex,draw=tripColor3,fill=tripColor3!15] {\gs};%
	\node (eo2_v) at (eo2) [vertex,draw=tripColor3,fill=tripColor3!15] {\gs};%
	\node (ed2_v) at (ed2) [vertex,draw=tripColor4,fill=tripColor4!15] {\gs};%
	\node (et2_v) at (et2) [vertex,draw=tripColor4,fill=tripColor4!15] {\gs};%
	
	\node at (es) [text=nodeColor!100] {\small{$\sourceEvent$}};%
	\node at (eo) [text=nodeColor!100] {\small{$\originEvent$}};%
	\node at (ed) [text=nodeColor!100] {\small{$\destinationEvent$}};%
	\node at (et) [text=nodeColor!100] {\small{$\targetEvent$}};%
	\node at (es2) [text=nodeColor!100] {\small{$\sourceEvent'$}};%
	\node at (eo2) [text=nodeColor!100] {\small{$\originEvent'$}};%
	\node at (ed2) [text=nodeColor!100] {\small{$\destinationEvent'$}};%
	\node at (et2) [text=nodeColor!100] {\small{$\targetEvent'$}};%
\end{tikzpicture}%
	\caption{%
		An example of how the origin delay interval is calculated for the candidate~$\aCandidateJourney=\candidateSequence$, assuming a delay limit of~$\maxDelay=5$.
		Each connection between two consecutive stop events~$\stopEvent_a$ and~$\stopEvent_b$ is labeled with their scheduled departure and arrival times, in the format~$\departureTime(\stopEvent_a)\to\arrivalTime(\stopEvent_b)$.
		The departure and arrival times in the parameterized delay scenario~$\bestDelayScenario(\aCandidateJourney,\delay)$ are given below in parentheses.
		The intermediate transfer from~$\originEvent$ to~$\destinationEvent$ has a slack of~$\slack(\originEvent,\destinationEvent)=-2$, which yields a feasibility limit is~$\feasibilityLimit(\aCandidateJourney)=3$.
		There are two witnesses to consider: the join witness~$\aJourney^\textsf{j}=\langle[\sourceEvent',\originEvent'],\destinationVertex\rangle$ and the split witness~$\aJourney^\textsf{s}=\langle[\sourceEvent,\originEvent],[\destinationEvent',\targetEvent']\rangle$.
		The join witness~$\aJourney^\textsf{j}$ has an arrival time of~23 at~$\destinationVertex$, whereas the origin prefix~$\originPrefix=\originPrefixSequence$ has an arrival time of~$21+\delay$.
		This yields a join limit of~$\joinLimit(\aCandidateJourney)=2$.
		The split witness~$\aJourney^\textsf{s}$ strongly dominates~$\aCandidateJourney$.
		It is feasible if the arrival time at~$\astop(\destinationEvent')$ is at most~24, which corresponds to a split limit of~$\splitLimit(\aCandidateJourney)=3$.
		Because there are no full witnesses, the minimum origin delay is~$\minOriginDelay(\aCandidateJourney)=\splitLimit(\aCandidateJourney)=3$ and the maximum origin delay is~$\maxOriginDelay(\aCandidateJourney)=\min(\feasibilityLimit(\aCandidateJourney),\joinLimit(\aCandidateJourney))=2$.
		Thus, the origin delay interval~$\originDelayInterval(\aCandidateJourney)=[3,2]$ is empty, which means that there is no delay scenario in which~$\aCandidateJourney$ is prefix-optimal.
	}%
	\label{fig:delay:originDelayInterval}%
\end{figure}

Given a group~$X$ of witnesses, a delay scenario~$\delayScenario$ is called \emph{$X$-avoiding} for~$\aCandidateJourney$ if no prefix of~$\aCandidateJourney$ is strongly dominated by an~$X$-witness in~$\delayScenario$.
A delay~$\delay$ is called~$X$-avoiding if~$\parameterizedBestDelayScenario(\aCandidateJourney,\delay)$ is~$X$-avoiding.
We call the lowest split-avoiding delay in~$[0,\maxDelay+1]$ the~\emph{split limit}~$\splitLimit(\aCandidateJourney)$, and the highest join-avoiding delay in~$(-\infty,\maxDelay]$ the~\emph{join limit}~$\joinLimit(\aCandidateJourney)$.
The \emph{minimum} and \emph{maximum origin delay}~$\minOriginDelay(\aCandidateJourney)$ and~$\maxOriginDelay(\aCandidateJourney)$ additionally take full witnesses into account:
\begin{align*}
	\minOriginDelay(\aCandidateJourney)&:=\begin{cases}
		\maxDelay+1 &\text{if }\bestDelayScenario(\aCandidateJourney)\text{ is not full-avoiding,}\\
		\splitLimit(\aCandidateJourney) &\text{otherwise.}
	\end{cases}\\
	\maxOriginDelay(\aCandidateJourney)&:=\begin{cases}
		-\infty &\text{if }\bestDelayScenario(\aCandidateJourney)\text{ is not full-avoiding,}\\
		\min(\feasibilityLimit(\aCandidateJourney),\joinLimit(\aCandidateJourney)) &\text{otherwise.}
	\end{cases}
\end{align*}
Together, they form the~\emph{origin delay interval}~$\originDelayInterval(\aCandidateJourney):=[\minOriginDelay(\aCandidateJourney),\maxOriginDelay(\aCandidateJourney)]$.
Figure~\ref{fig:delay:originDelayInterval} gives an example of how it is calculated.
Based on this, Theorem~\ref{th:app:prefix-optimality-conditions} establishes our final characterization of candidates that are prefix-optimal in at least one delay scenario.

\begin{theorem}
	\label{th:app:prefix-optimality-conditions}
	A candidate~$\aCandidateJourney$ is prefix-optimal in at least one delay scenario iff~$\originDelayInterval(\aCandidateJourney)\neq\emptyset$.
\end{theorem}
\begin{proof}
	By Lemma~\ref{th:app:all-best-case}, $\aCandidateJourney$ is prefix-optimal in at least one delay scenario iff there is a~$\delta\in[0,\feasibilityLimit(\aCandidateJourney)]$ such that~$\aCandidateJourney$ in~$\parameterizedBestDelayScenario(\aCandidateJourney,\delta)$.
	By Lemma~\ref{th:app:hook}, this is the case iff~$\parameterizedBestDelayScenario(\aCandidateJourney,\delta)$ is hook-avoiding for~$\aCandidateJourney$.
	By Lemma~\ref{th:app:group-dominance}, this is the case iff~$\delta\in\originDelayInterval(\aCandidateJourney)$.
	Hence, $\aCandidateJourney$ is prefix-optimal in at least one delay scenario iff~$\originDelayInterval(\aCandidateJourney)\neq\emptyset$.
\end{proof}

\subsection{Problem Definition}
\label{sec:optimality-conditions:definition}
Based on the condition established by Theorem~\ref{th:app:prefix-optimality-conditions}, we formally define the shortcut computation problem.
For a potential shortcut~$\edge=(\originEvent,\destinationEvent)$, let~$\candidateJourneys(\edge)$ be the set of candidates with an intermediate transfer from~$\originEvent$ to~$\destinationEvent$.
The shortcut is necessary if the union of the origin delay intervals for all candidates in~$\candidateJourneys(\edge)$ is not empty.
For simplicity, we do not consider the union (which may not form an interval itself) but rather the smallest interval containing all origin delay intervals.
Let~$\minOriginDelay(\edge)$ and~$\maxOriginDelay(\edge)$ be the lowest minimum and highest maximum origin delays among~$\candidateJourneys(\edge)$, respectively.
Then the interval is given by~$\originDelayInterval(\edge):=[\minOriginDelay(\edge),\maxOriginDelay(\edge)]$.
We will see in Section~\ref{sec:optimality-tests} that computing~$\minOriginDelay(\edge)$ exactly is expensive, so we only ask for a lower bound~$\minOriginDelayBound(\edge)$.
If the corresponding interval~$\originDelayIntervalBound(\edge):=[\minOriginDelayBound(\edge),\maxOriginDelay(\edge)]$ is empty, then~$\edge$ is not required in any delay scenario.
While the converse is not necessarily true, superfluous shortcuts only affect the performance of the query algorithm, not its correctness.

Given a network~$(\stops,\stopEvents,\trips,\routes,\graph)$ and a delay limit~\maxDelay, the~\textsc{DelayShortcut} problem asks for the set~$\shortcutEdges=\{\edge\in\stopEvents\times\stopEvents\mid{}\originDelayIntervalBound(\edge)\neq\emptyset\}$ of relevant shortcuts, as well as the origin delay interval~$\originDelayIntervalBound(\edge)$ for each shortcut~$\edge\in\shortcutEdges$.
The latter can be used to discard irrelevant shortcuts when the delay scenario is revealed in the update phase:
a shortcut~$\edge=(\originEvent,\destinationEvent)$ can be discarded in a delay scenario~$\delayScenario$ if~$\arrivalDelay(\originEvent)\notin\originDelayIntervalBound(\edge)$.

\section{Efficient Candidate Testing}
\label{sec:optimality-tests}
Theorem~\ref{th:app:prefix-optimality-conditions} implies an algorithmic framework for solving the~\textsc{DelayShortcut} problem: generate all possible candidates and test for each candidate~$\aCandidateJourney$ whether the origin delay interval~$\originDelayInterval(\aCandidateJourney)$ is empty.
To perform this test, the algorithm needs to search for witnesses that strongly dominate a prefix of~$\aCandidateJourney$ in one of the relevant delay scenarios.
These scenarios are different for each candidate, so a naive approach would perform a new search each time.
To obtain a more efficient algorithm, we exploit the fact that candidates with a common prefix share the same set of witnesses for this prefix.

Given an origin stop prefix~$\originPrefix(\originVertex)=\originStopPrefixSequence{\originVertex}$, we define the subproblem~\textsc{DelayShortcut}-$\originPrefix(\originVertex)$, which only allows candidates that begin with~$\originPrefix(\originVertex)$.
The overall~\textsc{DelayShortcut} problem can be broken down into solving~\textsc{DelayShortcut}-$\originPrefix(\originVertex)$ for every possible origin stop prefix~$\originPrefix(\originVertex)$ and merging the results.
This has the advantage that the origin event~$\originEvent$ is now fixed by the input, which means that a shortcut~$\edge=(\originEvent,\destinationEvent)$ is uniquely identified by~$\destinationEvent$.
Thus, the set~$\candidateJourneys(\edge)$ of candidates containing~$\edge$ becomes the set~$\candidateJourneys(\destinationEvent)$ of candidates with destination prefix~$\destinationPrefix$.
Likewise, the minimum and maximum origin delays~$\minOriginDelay(\edge)=\minOriginDelay(\destinationEvent)$ and~$\maxOriginDelay(\edge)=\maxOriginDelay(\destinationEvent)$, the origin delay interval~$\originDelayInterval(\edge)=\originDelayInterval(\destinationEvent)$ and their lower bounds also depend only on~$\destinationEvent$.

In this section, we show how an individual~\textsc{DelayShortcut}-$\originPrefix(\originVertex)$ problem can be solved with only two witness searches: one in~$\bestDelayScenario(\sourcePrefix)$ and one in~$\bestDelayScenario(\originPrefix)$.
For this purpose, we define various arrival times:
\begin{itemize}
\item For a vertex~$\aVertex$, the \emph{candidate arrival time} of~$\originPrefix(\aVertex)=\originStopPrefixSequence{\aVertex}$ in~$\bestDelayScenario$ is given by~$\candidateArrivalTime(\aVertex):=\arrivalTime(\originEvent)+\transferTime(\originVertex,\aVertex)$.
\item For a candidate prefix~$\prefix{\aCandidateJourney}$, a vertex~$\aVertex$ and a number of trips~$n\leq{}2$, the \emph{witness arrival time}~$\witnessArrivalTime(\prefix{\aCandidateJourney},\aVertex,n)$ is the earliest arrival time among $\sourceVertex$-$\aVertex$-journeys with at most~$n$ trips that depart no earlier than~$\departureTime(\sourceEvent)+\maxDelay$ in~$\bestDelayScenario(\prefix{\aCandidateJourney})$.
\end{itemize}
In the following, we establish formulas for the various components of the origin delay interval that depend only on~$\candidateArrivalTime(\cdot)$, $\witnessArrivalTime(\sourcePrefix,\cdot,\cdot)$, $\witnessArrivalTime(\originPrefix,\cdot,\cdot)$ and derived values.
This is done in three steps: calculating the join and feasibility limits (Section~\ref{sec:optimality-tests:join-feasibility}), examining full witnesses (Section~\ref{sec:optimality-tests:full-witnesses}), and calculating the split limits (Section~\ref{sec:optimality-tests:split}).
Additionally, Section~\ref{sec:time-travel} discusses how to exclude impossible delay scenarios in which trips travel backwards in time.

Throughout this section, we assume that the source prefix~$\sourcePrefix=\langle\aTrip_1[\sourceIndex]\rangle$ is Pareto-optimal in every delay scenario~$\delayScenario$.
To see why this assumption is realistic, consider the example in Figure~\ref{fig:witnessTypes:dot-source}.
A witness that strongly dominates~$\sourcePrefix$ can only be a~$(\bot,\sourceEvent)$-witness.
This means it must have the form~$\aWitnessJourney=\langle\aTrip_1[i]\rangle$ with~$i<\sourceIndex$, i.e., it takes a transfer to an earlier stop along~$\aTrip$ and enters the trip there.
If~$\aWitnessJourney$ is feasible in~$\delayScenario$, then~$\departureTime(\aTrip_1[i],\delayScenario)\geq\departureTime(\aTrip_1[\sourceIndex],\delayScenario)$ must hold.
This requires that the trip moves from index~$i$ to~$\sourceIndex$ instantaneously (or even backwards in time), which is unrealistic.

\subsection{Join and Feasibility Limit}
\label{sec:optimality-tests:join-feasibility}
The join limit~$\joinLimit(\aCandidateJourney)$ of a candidate~$\aCandidateJourney=\candidateSequence$ is the highest delay~$\delay\in[-\infty,\maxDelay]$ such that~$\parameterizedBestDelayScenario(\aCandidateJourney,\delay)$ is join-avoiding.
Since join witnesses are witnesses for the origin prefix~$\originPrefix=\originPrefixSequence$, they do not use~$\destinationEvent$ or~$\targetEvent$.
Accordingly, we can consider the parameterized scenario~$\parameterizedBestDelayScenario(\sourcePrefix,\delay)$ for the source prefix~$\sourcePrefix$ instead of~$\parameterizedBestDelayScenario(\aCandidateJourney,\delay)$.
We define the join limit for a vertex~$\aVertex$ as
\begin{equation}\label{eqn:join-limit}
	\joinLimit(\aVertex):=\witnessArrivalTime(\sourcePrefix,\aVertex,1)-\candidateArrivalTime(\aVertex).
\end{equation}

Lemma~\ref{th:app:join-limit} shows the relation to the join limit of a candidate.

\begin{lemma}
	\label{th:app:join-limit}
	For a candidate~$\aCandidateJourney=\candidateSequence$, $\joinLimit(\aCandidateJourney)=\joinLimit(\destinationVertex)$.
\end{lemma}
\begin{proof}
	Let~$\delayScenario:=\bestDelayScenario(\sourcePrefix)=\parameterizedBestDelayScenario(\sourcePrefix,\maxDelay)$ and let~$\aWitnessJourney$ be a~$\sourceVertex$-$\destinationVertex$-journey with at most one trip such that~$\arrivalTime(\aWitnessJourney,\delayScenario)=\witnessArrivalTime(\sourcePrefix,\destinationVertex,1)$.
	Since~$\candidateArrivalTime(\destinationVertex)=\arrivalTime(\originPrefix,\bestDelayScenario)$ and~$\aWitnessJourney$ has minimal arrival time among journeys with at most one trip, it follows that
	\begin{align*}
		\joinLimit(\destinationVertex)&=\arrivalTime(\aWitnessJourney,\delayScenario)-\arrivalTime(\originPrefix,\bestDelayScenario)\leq\arrivalTime(\originPrefix,\delayScenario)-\arrivalTime(\originPrefix,\bestDelayScenario)=\maxDelay.
	\end{align*}
	There are two possible cases:
	\begin{enumerate}
		\item If~$\aWitnessJourney$ uses~$\originEvent$, then~$\arrivalTime(\aWitnessJourney,\delayScenario)=\arrivalTime(\originPrefix,\delayScenario)$, so it follows that~$\joinLimit(\destinationVertex)=\maxDelay$.
		Since no join witness for~$\originPrefix$ has an earlier arrival time than~$\aWitnessJourney$ in~$\delayScenario$, it follows that~$\maxDelay$ is join-avoiding.
		\item If~$\aWitnessJourney$ does not use~$\originEvent$, it is a join witness.
		Consider the delay scenario~$\parameterizedBestDelayScenario(\sourcePrefix,\delay)$ for some delay~$\delay$.
		The arrival time of~$\aWitnessJourney$ is equal to~$\arrivalTime(\aWitnessJourney,\delayScenario)=\witnessArrivalTime(\sourcePrefix,\destinationVertex,1)$ and the arrival time of~$\originPrefix$ is~$\candidateArrivalTime(\destinationVertex)+\delay$.
		Thus, the scenario is join-avoiding iff~$\delay\leq\witnessArrivalTime(\sourcePrefix,\destinationVertex,1)-\candidateArrivalTime(\destinationVertex)=\joinLimit(\destinationVertex)$.
	\end{enumerate}
	In both cases, the highest join-avoiding delay in~$[-\infty,\maxDelay]$ is~$\joinLimit(\destinationVertex)$.
\end{proof}

The following two lemmas establish pruning rules based on the join limit of a vertex.

\begin{lemma}
	\label{th:app:join-limit-negative}
	Let~$\aCandidateJourney=\candidateSequence$ be a candidate with~$\joinLimit(\aCandidateJourney)\geq{}0$.
	For each vertex~$\aVertex\in\vertices$ visited by the intermediate transfer of~$\aCandidateJourney$, $\joinLimit(\aVertex)\geq{}0$.
\end{lemma}
\begin{proof}
	Assume that~$\joinLimit(\aVertex)<0$, i.e.,~$\witnessArrivalTime(\sourcePrefix,\aVertex,1)<\candidateArrivalTime(\aVertex)$.
	The transfer time between~$\aVertex$ and~$\destinationVertex$ is given by~$\transferTime(\aVertex,\destinationVertex)$.
	By the triangle inequality, $\witnessArrivalTime(\sourcePrefix,\destinationVertex,1)\leq\witnessArrivalTime(\sourcePrefix,\aVertex,1)+\transferTime(\aVertex,\destinationVertex)$.
	Since~$\aVertex$ is visited by the intermediate transfer of~$\aCandidateJourney$, it lies on a shortest~$\originVertex$-$\destinationVertex$-path.
	Hence, $\candidateArrivalTime(\destinationVertex)=\candidateArrivalTime(\aVertex)+\transferTime(\aVertex,\destinationVertex)$.
	It follows that~$\witnessArrivalTime(\sourcePrefix,\destinationVertex,1)<\candidateArrivalTime(\destinationVertex)$ and therefore~$\joinLimit(\aCandidateJourney)=\joinLimit(\destinationVertex)<0$, a contradiction.
\end{proof}

\begin{lemma}
	\label{th:app:join-limit-virt}
	Let~$\aVertex\in\vertices$ be a vertex.
	If~$\joinLimit(\aVertex)\geq{}0$, then~$\originPrefix(\aVertex)=\originStopPrefixSequence{\aVertex}$ is prefix-optimal in~$\virtualDelayScenario(\sourcePrefix)$.
\end{lemma}
\begin{proof}
	By our assumption at the beginning of Section~\ref{sec:optimality-tests}, $\sourcePrefix$ is Pareto-optimal in every delay scenario.
	It remains to be shown that there is no witness~$\aWitnessJourney$ that strongly dominates~$\originPrefix(\aVertex)$ in~$\virtualDelayScenario(\sourcePrefix)=(\bestDelayScenario,\bestDelayScenario(\sourcePrefix))$.
	The arrival time of~$\originPrefix(\aVertex)$ in~$\bestDelayScenario$ is given by~$\candidateArrivalTime(\aVertex)$.
	The arrival time of~$\aWitnessJourney$ in~$\bestDelayScenario(\sourcePrefix)$ is at most~$\witnessArrivalTime(\sourcePrefix,\aVertex,1)$.
	Since~$\joinLimit(\aVertex)=\witnessArrivalTime(\sourcePrefix,\aVertex,1)-\candidateArrivalTime(\aVertex)\geq{}0$, it follows that~$\aWitnessJourney$ does not strongly dominate~$\originPrefix(\aVertex)$.
\end{proof}

For the feasibility limit~$\feasibilityLimit(\aCandidateJourney)$, applying the definition of~$\candidateArrivalTime(\cdot)$ yields
\begin{equation}\label{eqn:feasibility-limit}
	\feasibilityLimit(\destinationEvent)=\min(0,\departureTime(\destinationEvent)-\candidateArrivalTime(\astop(\destinationEvent)))+\maxDelay.
\end{equation}

From this point onwards, we can restrict the set of relevant destination events to those with non-negative feasibility and join limits:
\[
\feasibleDestinationEvents:=\{\destinationEvent\in\stopEvents\mid\min(\feasibilityLimit(\destinationEvent),\joinLimit(\astop(\destinationEvent)))\geq0\}.
\]
All other destination events can be discarded because they cannot occur in prefix-optimal candidates.

\subsection{Examining Full Witnesses}
\label{sec:optimality-tests:full-witnesses}
For each destination event~$\destinationEvent\in\feasibleDestinationEvents$, we need to determine the minimum and maximum origin delay~$\minOriginDelay(\destinationEvent)$ and~$\maxOriginDelay(\destinationEvent)$.
These depend on the minimum and maximum origin delays of the candidates that use~$\destinationEvent$, which are contained in~$\candidateJourneys(\destinationEvent)$.
Candidates~$\aCandidateJourney$ for which~$\bestDelayScenario(\aCandidateJourney)$ is not full-avoiding do not contribute to~$\minOriginDelay(\destinationEvent)$ or~$\maxOriginDelay(\destinationEvent)$, so the next step is to examine full witnesses in order to discard these candidates.
Unfortunately, the witness arrival times~$\witnessArrivalTime(\cdot,\cdot,\cdot)$ do not distinguish between split and full witnesses.
We can circumvent this issue by considering the set~$\optimalCandidateJourneys(\destinationEvent)$ of candidates~$\aCandidateJourney\in\candidateJourneys(\destinationEvent)$ for which~$\minOriginDelay(\aCandidateJourney)\leq\maxDelay$, i.e.,~$\bestDelayScenario(\aCandidateJourney)$ is full-avoiding and~$\bestDelayScenario(\aCandidateJourney,\maxDelay)$ is split-avoiding.
Consider a candidate~$\aCandidateJourney\in\candidateJourneys(\destinationEvent)$.
If~$\aCandidateJourney$ is not contained in~$\optimalCandidateJourneys(\destinationEvent)$, we know that~$\originDelayInterval(\aCandidateJourney)=\emptyset$, so we can discard~$\aCandidateJourney$.
Otherwise, we can calculate~$\minOriginDelay(\aCandidateJourney)$ and~$\maxOriginDelay(\aCandidateJourney)$ from the split, feasibility and join limits.

\subparagraph*{Using Virtual Delay Scenarios.}
Our objective is to compute the set
\[\optimalDestinationEvents:=\{\destinationEvent\in\feasibleDestinationEvents\mid\optimalCandidateJourneys(\destinationEvent)\neq\emptyset\}.\]
For a destination event~$\destinationEvent\in\feasibleDestinationEvents$, it is difficult to compute~$\optimalCandidateJourneys(\destinationEvent)$ exactly.
Instead, we compute the set~$\virtualCandidateJourneys(\destinationEvent)$ of candidates in~$\candidateJourneys(\destinationEvent)$ that are prefix-optimal in~$\virtualDelayScenario(\sourcePrefix)$.
This allows us to use the witness arrival times for~$\bestDelayScenario(\sourcePrefix)$.
Lemma~\ref{th:app:optimal-virtual-candidates} shows that~$\virtualCandidateJourneys(\destinationEvent)$ is a superset of~$\optimalCandidateJourneys(\destinationEvent)$.

\begin{lemma}
	\label{th:app:optimal-virtual-candidates}
	For every destination event~$\destinationEvent\in\feasibleDestinationEvents$, $\optimalCandidateJourneys(\destinationEvent)\subseteq\virtualCandidateJourneys(\destinationEvent)$.
\end{lemma}
\begin{proof}
	Consider a candidate~$\aCandidateJourney\notin\virtualCandidateJourneys(\destinationEvent)$.
	Then a prefix~$\prefix{\aCandidateJourney}$ of~$\aCandidateJourney$ is strongly dominated by a witness~$\aWitnessJourney$ in~$\virtualDelayScenario(\sourcePrefix)$.
	Since~$\virtualDelayScenario(\sourcePrefix)\leqDom\parameterizedBestDelayScenario(\aCandidateJourney,\delay)$ for every~$\delay\in[0,\maxDelay]$, $\prefix{\aCandidateJourney}$ is also strongly dominated by~$\aWitnessJourney$ in~$\bestDelayScenario(\aCandidateJourney)$ and~$\parameterizedBestDelayScenario(\aCandidateJourney,\maxDelay)$.
	Since~$\joinLimit(\astop(\destinationEvent))\geq{}0$, $\aWitnessJourney$ is not a join witness.
	If~$\aWitnessJourney$ is a split witness, then~$\parameterizedBestDelayScenario(\aCandidateJourney,\maxDelay)$ is not split-avoiding, so~$\splitLimit(\aCandidateJourney)>\maxDelay$.
	If~$\aWitnessJourney$ a full witness, then~$\bestDelayScenario(\aCandidateJourney)$ is not full-avoiding.
	In both cases, it follows that~$\aCandidateJourney\notin\optimalCandidateJourneys(\destinationEvent)$.
\end{proof}

The reason why the two sets are not equal is that~$\bestDelayScenario(\sourcePrefix)$ assumes the worst case for~$\destinationEvent$.
This underestimates~$(\destinationEvent,\targetEvent)$-witnesses, which are the only witnesses that use~$\destinationEvent$ (see Figure~\ref{fig:witnessTypes:destination-target}).
Lemma~\ref{th:app:destination-witnesses} shows that this is the only difference between~$\optimalCandidateJourneys(\destinationEvent)$ and~$\virtualCandidateJourneys(\destinationEvent)$.

\begin{lemma}
	\label{th:app:destination-witnesses}
	Let~$\destinationEvent\in\feasibleDestinationEvents$ be a destination event and~$\aCandidateJourney$ a candidate in~$\virtualCandidateJourneys(\destinationEvent)\setminus\optimalCandidateJourneys(\destinationEvent)$.
	Then~$\aCandidateJourney$ is strongly dominated by an~$(\destinationEvent,\targetEvent)$-witness in~$\bestDelayScenario(\aCandidateJourney)$.
\end{lemma}
\begin{proof}
	Since~$\aCandidateJourney\notin\optimalCandidateJourneys(\destinationEvent)$, we know that~$\minOriginDelay(\aCandidateJourney)>\maxDelay$.
	There are two possible reasons for this:
	\begin{enumerate}
		\item There is a split witness~$\aWitnessJourney$ that strongly dominates a prefix~$\prefix{\aCandidateJourney}$ of~$\aCandidateJourney$ in~$\parameterizedBestDelayScenario(\aCandidateJourney,\maxDelay)$.
		Since~$\aWitnessJourney$ does not use~$\destinationEvent$ or~$\targetEvent$, this is also the case in~$\virtualDelayScenario(\sourcePrefix)$, which contradicts~$\aCandidateJourney\in\virtualCandidateJourneys(\destinationEvent)$.
		\item There is a full witness~$\aWitnessJourney$ that strongly dominates a prefix~$\prefix{\aCandidateJourney}$ of~$\aCandidateJourney$ in~$\bestDelayScenario(\aCandidateJourney)$.
		If~$\aWitnessJourney$ is a~$(\bot,\cdot)$- or~$(\sourceEvent,\cdot)$-witness, it does not use~$\originEvent$, $\destinationEvent$ or~$\targetEvent$.
		Then it still strongly dominates~$\prefix{\aCandidateJourney}$ in~$\virtualDelayScenario(\sourcePrefix)$, which contradicts~$\aCandidateJourney\in\virtualCandidateJourneys(\destinationEvent)$.
		Hence, $\aWitnessJourney$ is a~$(\destinationEvent,\targetEvent)$-witness.
	\end{enumerate}
\end{proof}

In order to evaluate~$(\destinationEvent,\targetEvent)$-witnesses exactly, we would need to perform an additional witness search in~$\bestDelayScenario(\destinationPrefix)$.
Because this would have to be done individually for each possible destination event in~$\feasibleDestinationEvents$, this would be expensive.
However, Lemma~\ref{th:app:full-witness-condition-candidate} shows that this is unnecessary for the purpose of computing~$\optimalDestinationEvents$.
Intuitively, $(\destinationEvent,\targetEvent)$-witnesses are themselves candidates with a different target event than~$\aCandidateJourney$.
If there is a~$(\destinationEvent,\targetEvent)$-witness that is not strongly dominated in~$\virtualDelayScenario(\sourcePrefix)$, then it is also a candidate in~$\optimalCandidateJourneys(\destinationEvent)$.
Thus, if~$\virtualCandidateJourneys(\destinationEvent)$ is not empty, then neither is~$\optimalCandidateJourneys(\destinationEvent)$, which means we can compute~$\optimalDestinationEvents$ as
\[\optimalDestinationEvents=\{\destinationEvent\in\feasibleDestinationEvents\mid\virtualCandidateJourneys(\destinationEvent)\neq\emptyset\}.\]

\begin{lemma}
	\label{th:app:full-witness-condition-candidate}
	For every destination event~$\destinationEvent\in\feasibleDestinationEvents$, $\virtualCandidateJourneys(\destinationEvent)=\emptyset$ iff $\optimalCandidateJourneys(\destinationEvent)=\emptyset$.
\end{lemma}
\begin{proof}
	By Lemma~\ref{th:app:optimal-virtual-candidates}, $\optimalCandidateJourneys(\destinationEvent)\subseteq\virtualCandidateJourneys(\destinationEvent)$.
	It therefore remains to be shown that if~$\optimalCandidateJourneys(\destinationEvent)$ is empty, then so is~$\virtualCandidateJourneys(\destinationEvent)$.
	Assume that~$\optimalCandidateJourneys(\destinationEvent)$ is empty but there is at least one candidate~$\aCandidateJourney=\candidateSequence\in\virtualCandidateJourneys(\destinationEvent)$.
	By Lemma~\ref{th:app:destination-witnesses}, $\aCandidateJourney$ is strongly dominated by an~$(\destinationEvent,\targetEvent)$-witness~$\aWitnessJourney=\langle[\sourceEvent,\originEvent],[\destinationEvent,\targetEvent']\rangle$ in~$\bestDelayScenario(\aCandidateJourney)$.
	W.l.o.g., choose~$\aWitnessJourney$ such that it is not strongly dominated by another~$(\destinationEvent,\targetEvent)$-witness in~$\bestDelayScenario(\aCandidateJourney)$.
	Since such a witness would not use~$\targetEvent'$, it follows that~$\aWitnessJourney$ is also not strongly dominated by an~$(\destinationEvent,\targetEvent')$-witness in~$\bestDelayScenario(\aWitnessJourney)$.
	
	Note that~$\aWitnessJourney\in\candidateJourneys(\destinationEvent)$ is itself a candidate.
	However, $\aWitnessJourney$ is not included in~$\optimalCandidateJourneys(\destinationEvent)$, which is empty by our assumption, and therefore not in~$\virtualCandidateJourneys(\destinationEvent)$ by Lemma~\ref{th:app:optimal-virtual-candidates}.
	Thus, $\aWitnessJourney$ is strongly dominated by another witness~$\overline{\aWitnessJourney}$ in~$\virtualDelayScenario(\sourcePrefix)$.
	Hence, $\overline{\aWitnessJourney}$ as evaluated in~$\bestDelayScenario(\sourcePrefix)$ strongly dominates~$\aWitnessJourney$ as evaluated in~$\bestDelayScenario(\aWitnessJourney)$, which in turn dominates~$\aWitnessJourney$ as evaluated in~$\bestDelayScenario(\aCandidateJourney)$.
	Since~$\aWitnessJourney$ strongly dominates~$\aCandidateJourney$ in~$\bestDelayScenario(\aCandidateJourney)$, it follows that~$\overline{\aWitnessJourney}$ strongly dominates~$\aCandidateJourney$ in~$\virtualDelayScenario(\sourcePrefix)$.
	This contradicts~$\aCandidateJourney\in\virtualCandidateJourneys(\destinationEvent)$.
\end{proof}

\subparagraph*{Witness Indices.}
\begin{figure}
	\centering
	\begin{tikzpicture}
\node (e0) at (0.00, 0.00) {};%
\node (e1) at (2.00, 0.00) {};%
\node (e2) at (4.00, 0.00) {};%
\node (e3) at (6.00, 0.00) {};%
\node (e4) at (8.00, 0.00) {};%
\node (e5) at (10.00, 0.00) {};%

\node (v0) at (0.00, -0.50) {};%
\node (v1) at (2.00, -0.50) {};%
\node (v2) at (4.00, -0.50) {};%
\node (v3) at (6.00, -0.50) {};%
\node (v4) at (8.00, -0.50) {};%
\node (v5) at (10.00, -0.50) {};%

\node [anchor=east] at (-0.75, 1.00) {$\witnessArrivalTime(\sourcePrefix,\aVertex_i,2)$:};%
\node at (0.00, 1.00) {8:50};%
\node at (2.00, 1.00) {8:54};%
\node (wit2_2) at (4.00, 1.00) {9:12};%
\node at (6.00, 1.00) {9:03};%
\node (wit2_4) at (8.00, 1.00) {9:21};%
\node at (10.00, 1.00) {9:07};%

\node [anchor=east] at (-0.75, 0.50) {$\arrivalTime(\stopEvent_i)$:};%
\node at (0.00, 0.50) {9:00};%
\node at (2.00, 0.50) {9:05};%
\node (arr_2) at (4.00, 0.50) {9:08};%
\node at (6.00, 0.50) {9:15};%
\node (arr_4) at (8.00, 0.50) {9:20};%
\node at (10.00, 0.50) {9:24};%

\node [anchor=east] at (-0.75, -1.00) {$\departureTime(\stopEvent_i)$:};%
\node at (0.00, -1.00) {9:02};%
\node (dep_1) at (2.00, -1.00) {9:06};%
\node at (4.00, -1.00) {9:10};%
\node at (6.00, -1.00) {9:16};%
\node (dep_4) at (8.00, -1.00) {9:21};%
\node (dep_5) at (10.00, -1.00) {9:25};%

\node [anchor=east] at (-0.75, -1.50) {$\witnessArrivalTime(\sourcePrefix,\aVertex_i,1)$:};%
\node at (0.00, -1.50) {9:30};%
\node (wit1_1) at (2.00, -1.50) {8:54};%
\node at (4.00, -1.50) {9:12};%
\node at (6.00, -1.50) {9:18};%
\node (wit1_4) at (8.00, -1.50) {9:21};%
\node (wit1_5) at (10.00, -1.50) {9:07};%

\node at (4.00, -2.00) {$\uparrow$};%
\node at (8.00, -2.00) {$\uparrow$};%
\node at (4.00, -2.50) {$\entryIndex(\sourcePrefix,\aTrip)=\witnessIndex(\sourcePrefix,\aTrip)=2$};%
\node at (8.00, -2.50) {$\exitIndex(\sourcePrefix,\aTrip)=4$};%

\node at (-1.50, 1.75) {\textcolor{KITorange}{$\candidateExitEvents(\sourcePrefix,\aTrip)=\{\stopEvent_2,\stopEvent_4\}$}};%
\node at (-1.50, -2.25) {\textcolor{KITlilac}{$\witnessEntryEvents(\sourcePrefix,\aTrip)=\{\stopEvent_1,\stopEvent_4,\stopEvent_5\}$}};%

\node [stop,fit=(e0)(v0)] {};%
\node [stop,fit=(e1)(v1)] {};%
\node [stop,fit=(e2)(v2)] {};%
\node [stop,fit=(e3)(v3)] {};%
\node [stop,fit=(e4)(v4)] {};%
\node [stop,fit=(e5)(v5)] {};%

\node [rectangle, line width=.5pt, inner sep=-2pt,rounded corners=0.1cm,fit=(arr_2)(wit2_2),draw=KITorange,fill=KITorange!15] {};%
\node [rectangle, line width=.5pt, inner sep=-2pt,rounded corners=0.1cm,fit=(arr_4)(wit2_4),draw=KITorange,fill=KITorange!15] {};%
\node [rectangle, line width=.5pt, inner sep=-2pt,rounded corners=0.1cm,fit=(dep_1)(wit1_1),draw=KITlilac,fill=KITlilac!15] {};%
\node [rectangle, line width=.5pt, inner sep=-2pt,rounded corners=0.1cm,fit=(dep_4)(wit1_4),draw=KITlilac,fill=KITlilac!15] {};%
\node [rectangle, line width=.5pt, inner sep=-2pt,rounded corners=0.1cm,fit=(dep_5)(wit1_5),draw=KITlilac,fill=KITlilac!15] {};%

\node at (wit2_2) [text=nodeColor!100] {\small{9:12}};%
\node at (wit2_4) [text=nodeColor!100] {\small{9:21}};%
\node at (arr_2) [text=nodeColor!100] {\small{9:08}};%
\node at (arr_4) [text=nodeColor!100] {\small{9:20}};%
\node at (dep_1) [text=nodeColor!100] {\small{9:06}};%
\node at (dep_4) [text=nodeColor!100] {\small{9:21}};%
\node at (dep_5) [text=nodeColor!100] {\small{9:25}};%
\node at (wit1_1) [text=nodeColor!100] {\small{8:54}};%
\node at (wit1_4) [text=nodeColor!100] {\small{9:21}};%
\node at (wit1_5) [text=nodeColor!100] {\small{9:07}};%

\draw [tripColor1,route] (e0) -- (e1) -- (e2) -- (e3) -- (e4) -- (e5);

\node (e0_v) at (e0) [vertex,draw=tripColor1,fill=tripColor1!15] {\gs};%
\node (e1_v) at (e1) [vertex,draw=tripColor1,fill=tripColor1!15] {\gs};%
\node (e2_v) at (e2) [vertex,draw=tripColor1,fill=tripColor1!15] {\gs};%
\node (e3_v) at (e3) [vertex,draw=tripColor1,fill=tripColor1!15] {\gs};%
\node (e4_v) at (e4) [vertex,draw=tripColor1,fill=tripColor1!15] {\gs};%
\node (e5_v) at (e5) [vertex,draw=tripColor1,fill=tripColor1!15] {\gs};%

\node at (e0) [text=nodeColor!100] {\small{$\stopEvent_0$}};%
\node at (e1) [text=nodeColor!100] {\small{$\stopEvent_1$}};%
\node at (e2) [text=nodeColor!100] {\small{$\stopEvent_2$}};%
\node at (e3) [text=nodeColor!100] {\small{$\stopEvent_3$}};%
\node at (e4) [text=nodeColor!100] {\small{$\stopEvent_4$}};%
\node at (e5) [text=nodeColor!100] {\small{$\stopEvent_5$}};%

\node at (v0) [text=nodeColor!100] {\small{$\aVertex_0$}};%
\node at (v1) [text=nodeColor!100] {\small{$\aVertex_1$}};%
\node at (v2) [text=nodeColor!100] {\small{$\aVertex_2$}};%
\node at (v3) [text=nodeColor!100] {\small{$\aVertex_3$}};%
\node at (v4) [text=nodeColor!100] {\small{$\aVertex_4$}};%
\node at (v5) [text=nodeColor!100] {\small{$\aVertex_5$}};%
\end{tikzpicture}
	\caption{A trip~$\aTrip$ with entry index~$2$, exit index~$4$ and witness index~$2$.}
	\label{fig:reachedIndex}
\end{figure}
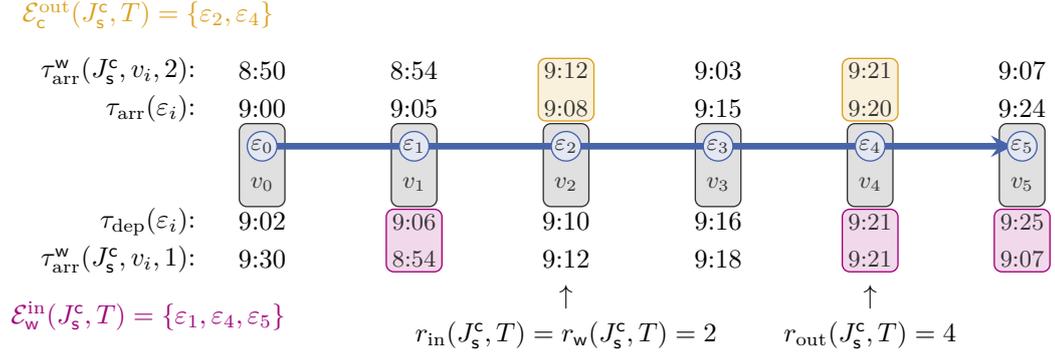

To compute~$\virtualCandidateJourneys(\destinationEvent)$ efficiently, we adapt the concept of reached indices from TB~\cite{Wit15} (see Figure~\ref{fig:reachedIndex}).
Consider a stop event~$\stopEvent$.
If~$\departureTime(\stopEvent)\geq\witnessArrivalTime(\sourcePrefix,\astop(\stopEvent),1)$, we call~$\stopEvent$ a \emph{w-in event} since there is a witness in~$\bestDelayScenario(\sourcePrefix)$ that is able to enter its second trip at~$\stopEvent$.
If~$\arrivalTime(\stopEvent)\leq\witnessArrivalTime(\sourcePrefix,\astop(\stopEvent),2)$, we call~$\stopEvent$ a~\emph{c-out event} since candidates ending with~$\stopEvent$ are not strongly dominated by a witness in~$\virtualDelayScenario(\sourcePrefix)$.
For a trip~$\aTrip$, let~$\witnessEntryEvents(\sourcePrefix,\aTrip)$ denote the set of w-in events and~$\candidateExitEvents(\sourcePrefix,\aTrip)$ the set of c-out events in~$\aTrip$.
We define the~\emph{entry index}~$\entryIndex(\sourcePrefix,\aTrip)$ and~\emph{exit index}~$\exitIndex(\sourcePrefix,\aTrip)$ as
\begin{align*}
	\entryIndex(\sourcePrefix,\aTrip)&:=\begin{cases}
		|\aTrip| & \text{if } \witnessEntryEvents(\sourcePrefix,\aTrip)=\emptyset,\\
		\min\limits_{\aTrip[i]\in\witnessEntryEvents(\sourcePrefix,\aTrip)} i + 1 & \text{otherwise,}
	\end{cases}\\
	\exitIndex(\sourcePrefix,\aTrip)&:=\begin{cases}
		|\aTrip| & \text{if } \candidateExitEvents(\sourcePrefix,\aTrip)=\emptyset,\\
		\max\limits_{\aTrip[i]\in\candidateExitEvents(\sourcePrefix,\aTrip)} i & \text{otherwise.}
	\end{cases}
\end{align*}
The entry and exit indices are combined into the~\emph{witness index}
\[\witnessIndex(\sourcePrefix,\aTrip):=\min(\entryIndex(\sourcePrefix,\aTrip),\exitIndex(\sourcePrefix,\aTrip)).\]

Based on these definitions, we can make two observations for a destination event $\aTrip[i]\in\feasibleDestinationEvents$:
The destination prefix~$\destinationPrefix=\langle[\sourceEvent,\originEvent],\aTrip[i]\rangle$ is Pareto-optimal in~$\virtualDelayScenario(\sourcePrefix)$ iff~$\aTrip[i]$ is not preceded by any w-in events in~$\aTrip$, i.e.,~$i<\entryIndex(\sourcePrefix,\aTrip)$.
Furthermore, at least one candidate in~$\candidateJourneys(\aTrip[i])$ is Pareto-optimal in~$\virtualDelayScenario(\sourcePrefix)$ iff~$\aTrip[i]$ is succeeded by at least one c-out event, i.e.,~$i<\exitIndex(\sourcePrefix,\aTrip)$.
It follows that~$\aTrip[i]$ is contained in~$\optimalDestinationEvents$ iff~$i<\witnessIndex(\sourcePrefix,\aTrip)$, as shown by Lemma~\ref{th:app:full-witness-condition}.

\begin{lemma}
	\label{th:app:full-witness-condition}
	The set~$\optimalDestinationEvents$ is equal to~$\{\aTrip[i]\in\feasibleDestinationEvents\mid{}i<\witnessIndex(\sourcePrefix,\aTrip)\}$.
\end{lemma}
\begin{proof}
	Consider a destination event~$\aTrip[i]\in\feasibleDestinationEvents$.
	By Lemma~\ref{th:app:join-limit-virt}, it follows from~$\joinLimit(\astop(\aTrip[i]))\geq{}0$ that~$\originPrefix$ is prefix-optimal in~$\virtualDelayScenario(\sourcePrefix)$.
	Hence, a candidate~$\aCandidateJourney\in\candidateJourneys(\aTrip[i])$ is prefix-optimal in~$\virtualDelayScenario(\sourcePrefix)$ iff~$\aCandidateJourney$ and the destination prefix~$\langle[\sourceEvent,\originEvent],\aTrip[i]\rangle$ are Pareto-optimal in~$\virtualDelayScenario(\sourcePrefix)$.
	It follows that~$\virtualCandidateJourneys(\aTrip[i])\neq\emptyset$ iff~$i<\witnessIndex(\sourcePrefix,\aTrip)$.
	By Lemma~\ref{th:app:full-witness-condition-candidate}, this is equivalent to~$\optimalCandidateJourneys(\aTrip[i])\neq\emptyset$.
\end{proof}

For a destination event~$\aTrip[i]$, let~$\virtualCandidateJourneys(\aTrip[i])$ denote the set of target events that occur in the candidates from~$\virtualTargetEvents(\aTrip[i])$.
Lemma~\ref{th:app:virtual-target-events} shows that these are exactly the c-out events succeeding~$\aTrip[i]$ in~$\aTrip$.

\begin{lemma}
	\label{th:app:virtual-target-events}
	For a destination event~$\aTrip[i]\in\optimalDestinationEvents$, $\virtualTargetEvents(\aTrip[i])=\candidateExitEvents(\sourcePrefix,\aTrip)\cap\succeedingEvents(\aTrip[i])$.
\end{lemma}
\begin{proof}
	Consider a target event~$\aTrip[j]\in\succeedingEvents(\aTrip[i])$ and the corresponding candidate~$\aCandidateJourney=\langle[\sourceEvent,\originEvent],\aTripSegment{i}{j}\rangle$.
	Since~$\aTrip[i]\in\optimalDestinationEvents$, the destination prefix~$\destinationPrefix=\destinationPrefixSequence$ is prefix-optimal in~$\virtualDelayScenario(\sourcePrefix)$.
	Thus, $\aTrip[j]\in\virtualTargetEvents(\aTrip[i])$ iff~$\aCandidateJourney$ is Pareto-optimal in~$\virtualDelayScenario(\sourcePrefix)$.
	By definition, this is the case iff~$\aTrip[j]\in\candidateExitEvents(\sourcePrefix,\aTrip)$.
\end{proof}

\subsection{Split Limit}
\label{sec:optimality-tests:split}
\begin{figure}
	\centering
	\begin{tikzpicture}
\node (es) at (0.00, 0.00) {};%
\node (eo) at (4.00, 0.00) {};%
\node (v_inter) at (6.00, 0.00) {};%
\node (e1) at (6.00, -1.00) {};%
\node (ed) at (8.00, -1.00) {};%
\node (et) at (12.00, -1.00) {};%
\node (ed2) at (7.00, 1.25) {};%
\node (et2) at (10.00, 1.25) {};%
\node (v_final) at (11.50, 0.00) {};%

\node (vs) at (0.00,-0.50) {};%
\node (vo) at (4.00,-0.50) {};%
\node (vd) at (8.00,-1.50) {};%
\node (vt) at (12.00,-1.50) {};%

\draw [edge]  (eo) -- (v_inter);
\draw [edge]  (v_inter) -- (e1);
\draw [edge]  (v_inter) -- (ed);
\draw [edge]  (v_inter) -- (ed2);
\draw [edge]  (et2) -- (v_final);
\draw [edge]  (v_inter) -- (v_final);
\draw [edge]  (v_final) -- (et);

\node [stop,fit=(es)(vs)] {};%
\node [stop,fit=(eo)(vo)] {};%
\node [stop,fit=(ed)(vd)] {};%
\node [stop,fit=(et)(vt)] {};%
\node [stop,fit=(ed2)] {};%
\node [stop,fit=(et2)] {};%
\node [stop,fit=(e1)] {};%

\draw [tripColor1,route] (es) -- (eo);
\draw [tripColor2,route] (e1) -- (ed) -- (et);
\draw [tripColor4,route] (ed2) -- (et2);

\node (es_v) at (es) [vertex,draw=tripColor1,fill=tripColor1!15] {\gs};%
\node (eo_v) at (eo) [vertex,draw=tripColor1,fill=tripColor1!15] {\gs};%
\node (v_inter_v) at (v_inter) [vertex,draw=KITblack,fill=KITblack!15] {};%
\node (e1_v) at (e1) [vertex,draw=tripColor2,fill=tripColor2!15] {\gs};%
\node (ed_v) at (ed) [vertex,draw=tripColor2,fill=tripColor2!15] {\gs};%
\node (et_v) at (et) [vertex,draw=tripColor2,fill=tripColor2!15] {\gs};%
\node (ed2_v) at (ed2) [vertex,draw=tripColor4,fill=tripColor4!15] {\gs};%
\node (et2_v) at (et2) [vertex,draw=tripColor4,fill=tripColor4!15] {\gs};%
\node (v_final_v) at (v_final) [vertex,draw=KITblack,fill=KITblack!15] {};%

\node at (es) [text=nodeColor!100] {\small{$\sourceEvent$}};%
\node at (eo) [text=nodeColor!100] {\small{$\originEvent$}};%
\node at (e1) [text=nodeColor!100] {\small{$\stopEvent_1$}};%
\node at (ed) [text=nodeColor!100] {\small{$\destinationEvent$}};%
\node at (et) [text=nodeColor!100] {\small{$\targetEvent$}};%
\node at (ed2) [text=nodeColor!100] {\small{$\destinationEvent'$}};%
\node at (et2) [text=nodeColor!100] {\small{$\targetEvent'$}};%

\node at (vs) [text=nodeColor!100] {\small{$\sourceVertex$}};%
\node at (vo) [text=nodeColor!100] {\small{$\originVertex$}};%
\node at (vd) [text=nodeColor!100] {\small{$\destinationVertex$}};%
\node at (vt) [text=nodeColor!100] {\small{$\targetVertex$}};%
\end{tikzpicture}%
	\caption{%
		An example network showcasing the three split witness types for a candidate~$\aCandidateJourney=\candidateSequence$.
		The partial journey~$\langle[\sourceEvent,\originEvent],\stopEvent_1\rangle$ is a~$(\originEvent,\destinationEvent)$-witness.
		The journey~$\langle[\sourceEvent,\originEvent],\targetVertex\rangle$ that takes a direct transfer from~$\originVertex$ to~$\targetVertex$ is a direct~$(\originEvent,\targetEvent)$-witness, whereas the journey~$\langle[\sourceEvent,\originEvent],[\destinationEvent',\targetEvent']\rangle$ that takes a second trip is an indirect~$(\originEvent,\targetEvent)$-witness.
	}%
	\label{fig:delay:splitWitnesses}%
\end{figure}
The split limit~$\splitLimit(\aCandidateJourney)$ of a candidate~$\aCandidateJourney=\candidateSequence$ is the lowest delay~$\delay\in[0,\maxDelay+1]$ such that~$\parameterizedBestDelayScenario(\aCandidateJourney,\delay)$ is split-avoiding.
Since split witnesses do not use the destination event~$\destinationEvent$ or target event~$\targetEvent$, we can replace~$\parameterizedBestDelayScenario(\aCandidateJourney,\delay)$ with the parameterized scenario~$\parameterizedBestDelayScenario(\sourcePrefix,\delay)$ for the source prefix~$\sourcePrefix$.
As shown in Figures~\ref{fig:witnessTypes:origin-destination} and~\ref{fig:witnessTypes:origin-target}, split witnesses are either~$(\originEvent,\destinationEvent)$-witnesses for the destination prefix~$\destinationPrefix$ or~$(\originEvent,\targetEvent)$-witnesses for the candidate~$\aCandidateJourney$ itself.
We divide the latter further into~\emph{direct~$(\originEvent,\targetEvent)$-witnesses}, which use only one trip, and~\emph{indirect~$(\originEvent,\targetEvent)$-witnesses}, which use two.
The three split witness types are shown in Figure~\ref{fig:delay:splitWitnesses}.

Because each split witness type has a different effect on the split limit, we divide it into three components.
The~\emph{\destinationIndex-split limit}~$\destinationSplitLimit(\destinationEvent)$ considers~$(\originEvent,\destinationEvent)$-witnesses, the~\emph{1\targetIndex-split limit}~$\directTargetSplitLimit(\targetEvent)$ considers direct~$(\originEvent,\targetEvent)$-witnesses, and the~\emph{2\targetIndex-split limit}~$\indirectTargetSplitLimit(\targetEvent)$ considers indirect~$(\originEvent,\targetEvent)$-witnesses.
The overall split limit~$\splitLimit(\aCandidateJourney)$ is the maximum of all three.
In the following, we establish individual formulas for the three split limits.

Direct~$(\originEvent,\targetEvent)$-witnesses differ from the other two types in the effect that the arrival delay of the origin event~$\originEvent$ has on them:
For a direct~$(\originEvent,\targetEvent)$-witness~$\aWitnessJourney$, the origin delay does not affect whether~$\aWitnessJourney$ is feasible because it does not have an intermediate transfer, but it directly influences the arrival time and therefore whether~$\aWitnessJourney$ strongly dominates~$\aCandidateJourney$.
Thus, the 1\targetIndex-split limit is the lowest delay~$\delay\in[0,\maxDelay+1]$ such that no direct~$(\originEvent,\targetEvent)$-witness has a lower arrival time than~$\aCandidateJourney$ in~$\parameterizedBestDelayScenario(\aCandidateJourney,\delay)$.

For witnesses~$\aWitnessJourney$ of the other two types, the origin delay does not affect whether~$\aWitnessJourney$ strongly dominates~$\aCandidateJourney$ or~$\destinationPrefix$, but it determines whether the intermediate transfer is feasible.
To quantify this effect, we define the~\emph{maximum witness delay}~$\maxWitnessDelay(\aWitnessJourney)$.
This is the highest delay~$\delay\in[0,\maxDelay]$ such that~$\aWitnessJourney$ is feasible in~$\parameterizedBestDelayScenario(\sourcePrefix,\delay)$, or~$-1$ if~$\aWitnessJourney$ is not feasible in any of these scenarios.
If~$\aWitnessJourney$ contains an intermediate transfer from~$\originEvent$ to a destination event~$\destinationEvent'$, then~$\maxWitnessDelay(\aWitnessJourney)$ is given by
\begin{equation}\label{eqn:max-witness-delay}
\maxWitnessDelay(\destinationEvent'):=\max(-1,\min(\departureTime(\destinationEvent')-\candidateArrivalTime(\astop(\destinationEvent')),\maxDelay)).
\end{equation}
Otherwise, $\maxWitnessDelay(\aWitnessJourney)=\maxDelay$.
Let~$X$ be a type of split witness (either $(\originEvent,\destinationEvent)$- or indirect~$(\originEvent,\targetEvent)$-witnesses) and let~$\journeys_X$ be the set of $X$-witnesses that strongly dominate the respective candidate prefix (either~$\destinationPrefix$ or~$\aCandidateJourney$) in~$\bestDelayScenario(\sourcePrefix)$, assuming they are feasible.
If~$\journeys_X$ is empty, then the~$X$-split limit is 0.
Otherwise, it is the highest maximum witness delay among~$\journeys_X$ plus~1.

\subparagraph*{\destinationIndex-Split Limit.}
The set of strongly dominating~$(\originEvent,\destinationEvent)$-witnesses is given by
\[\journeys_\destinationIndex=\{\langle[\sourceEvent,\originEvent],\destinationEvent'\rangle \mid \destinationEvent'\in\precedingEvents(\destinationEvent) \}.\]
Accordingly, the~\destinationIndex-split limit can be calculated as
\begin{align*}
	\destinationSplitLimit(\destinationEvent)&=\begin{cases}
		0 &\text{if }\precedingEvents(\destinationEvent)=\emptyset,\\
		\max\limits_{\aTrip[j]\in\precedingEvents(\destinationEvent)\cap\optimalDestinationEvents}\maxWitnessDelay(\aTrip[j])+1 &\text{otherwise.}
	\end{cases}
\end{align*}
This requires iterating over the preceding stop events of~$\destinationEvent$ and evaluating the maximum witness delay.
Additionally, Lemma~\ref{th:app:destination-split-limit} shows that stop events that are not in~$\optimalDestinationEvents$ can be skipped.

\begin{lemma}
	\label{th:app:destination-split-limit}
	Let~$\destinationEvent=\aTrip[i]\in\optimalDestinationEvents$ be a destination event.
	Then
	\[\destinationSplitLimit(\destinationEvent)=\begin{cases}
		0 &\text{if }\precedingEvents(\destinationEvent)\cap\optimalDestinationEvents=\emptyset,\\
		\max\limits_{\aTrip[k]\in\precedingEvents(\destinationEvent)\cap\optimalDestinationEvents}\maxWitnessDelay(\aTrip[k])+1 &\text{otherwise.}
	\end{cases}\]
\end{lemma}
\begin{proof}
	If~$\precedingEvents(\destinationEvent)=\emptyset$ or~$\destinationSplitLimit(\destinationEvent)=0$, the claim is trivially true.
We therefore assume that~$i>0$ and~$\destinationSplitLimit(\destinationEvent)\geq{}1$.
Let~$j<i$ be the smallest index such that~$\destinationSplitLimit(\destinationEvent)=\maxWitnessDelay(\aTrip[j])+1$.
We define~$\destinationEvent':=\aTrip[j]$ and~$\destinationVertex':=\astop(\destinationEvent')$.
We show that~$\destinationEvent'\in\optimalDestinationEvents$ in two steps: $\destinationEvent'\in\feasibleDestinationEvents$ and~$\virtualCandidateJourneys(\destinationEvent')\neq\emptyset$ (which is equivalent to~$\optimalCandidateJourneys(\destinationEvent')\neq\emptyset$ by Lemma~\ref{th:app:full-witness-condition-candidate}).

\begin{description}
	\item[Step 1:]
	A comparison of Equations~\ref{eqn:feasibility-limit} and~\ref{eqn:max-witness-delay} shows that~$\feasibilityLimit(\destinationEvent')\geq\maxWitnessDelay(\destinationEvent')$.
	Since we know that~$\maxWitnessDelay(\destinationEvent')=\destinationSplitLimit(\destinationEvent)-1\geq{}0$, it follows that~$\feasibilityLimit(\destinationEvent')\geq{}0$.
	By Lemma~\ref{th:app:full-witness-condition}, $\destinationEvent\in\optimalDestinationEvents$ implies that~$\entryIndex(\sourcePrefix,\aTrip)>i$.
	Because~$j<i$, this means that~$\destinationEvent'=\aTrip[j]$ is not a w-in event, i.e., $\departureTime(\destinationEvent')<\witnessArrivalTime(\sourcePrefix,\destinationVertex',1)$.
	Furthermore, $\maxWitnessDelay(\destinationEvent')\geq{}0$ implies that~$\departureTime(\destinationEvent')\geq\candidateArrivalTime(\destinationVertex')$.
	With Lemma~\ref{th:app:join-limit}, this yields
	\[
	\joinLimit(\destinationVertex')=\witnessArrivalTime(\sourcePrefix,\destinationVertex',1)-\candidateArrivalTime(\destinationVertex')>\departureTime(\destinationEvent')-\candidateArrivalTime(\destinationVertex')\geq{}0.
	\]
	It follows from~$\feasibilityLimit(\destinationEvent')\geq{}0$ and~$\joinLimit(\destinationVertex')\geq{}0$ that~$\destinationEvent'\in\feasibleDestinationEvents$.
	\item[Step 2:]
	Because~$\destinationEvent$ is contained in~$\optimalDestinationEvents$, the set~$\virtualCandidateJourneys(\destinationEvent)$ must contain at least one candidate~$\aCandidateJourney=\candidateSequence$, which is prefix-optimal in~$\virtualDelayScenario(\sourcePrefix)$.
	We show that~$\overline{\aCandidateJourney}=\langle[\sourceEvent,\originEvent],[\destinationEvent',\targetEvent]\rangle\in\candidateJourneys(\destinationEvent')$ is also prefix-optimal in~$\virtualDelayScenario(\sourcePrefix)$ by considering the prefixes individually:
	\begin{itemize}
		\item \textbf{Origin prefix:} It follows from~$\joinLimit(\destinationVertex')\geq{}0$ and Lemma~\ref{th:app:join-limit-virt} that the origin prefix~$\langle[\sourceEvent,\originEvent],\destinationVertex'\rangle$ is prefix-optimal in~$\virtualDelayScenario(\sourcePrefix)$.
		\item \textbf{Destination prefix:} Assume that the destination prefix~$\overline{\destinationPrefix}=\langle[\sourceEvent,\originEvent],\aTrip[j]\rangle$ is not Pareto-optimal in~$\virtualDelayScenario(\sourcePrefix)$.
		Then there is a witness~$\aWitnessJourney=\langle[\sourceEvent,\originEvent],\aTrip[k]\rangle$ with~$k<j$ that strongly dominates it.
		However, if~$\aWitnessJourney$ is feasible in~$\bestDelayScenario(\sourcePrefix)$, it follows that~$\maxWitnessDelay(\aTrip[k])=\maxDelay\geq\maxWitnessDelay(\aTrip[j])$, which contradicts our choice of~$j$.
		\item \textbf{Candidate:} The candidate~$\overline{\aCandidateJourney}$ itself is Pareto-optimal in~$\virtualDelayScenario(\sourcePrefix)$ since it shares the same target event as~$\aCandidateJourney$.  	
	\end{itemize}
	Because~$\virtualCandidateJourneys(\destinationEvent')$ contains at least~$\overline{\aCandidateJourney}$, it is not empty.
\end{description}
\end{proof}

\subparagraph*{1\targetIndex-Split Limit.}
Lemma~\ref{th:app:direct-target-split-limit} establishes a simple formula for the~1$\targetIndex$-split limit.

\begin{lemma}
	\label{th:app:direct-target-split-limit}
	Let~$\aCandidateJourney=\candidateSequence\in\virtualCandidateJourneys(\destinationEvent)$ be a candidate.
	Then
	\[\directTargetSplitLimit(\targetEvent)=\min(\maxDelay+1,\arrivalTime(\targetEvent)-\candidateArrivalTime(\targetVertex)).\]
\end{lemma}
\begin{proof}
	The only possible direct~$(\originEvent,\targetEvent)$-witness for~$\aCandidateJourney$ is the journey~$\aWitnessJourney=\langle[\sourceEvent,\originEvent],\targetVertex\rangle$ that takes a final transfer directly from~$\originVertex$ to~$\targetVertex$.
	Since~$\aWitnessJourney$ does not use an intermediate transfer, it is feasible in~$\parameterizedBestDelayScenario(\sourcePrefix,\delay)$ for any~$\delay\in[0,\maxDelay+1]$.
	Its arrival time in~$\parameterizedBestDelayScenario(\sourcePrefix,\delay)$ is given by~$\candidateArrivalTime(\targetVertex)+\delay$, while the arrival time of~$\aCandidateJourney$ is given by~$\arrivalTime(\targetEvent)$.
	Hence, $\aCandidateJourney$ is strongly dominated by a direct~$(\originEvent,\targetEvent)$-witness iff~$\delay<\arrivalTime(\targetEvent)-\candidateArrivalTime(\targetVertex)$.
	Since~$\aCandidateJourney\in\virtualCandidateJourneys(\destinationEvent)$, it follows that~$\aCandidateJourney$ is not strongly dominated by~$\aWitnessJourney$ in~$\virtualDelayScenario(\sourcePrefix)$.
	Since~$\virtualDelayScenario(\originPrefix)\leqDom\virtualDelayScenario(\sourcePrefix)$, it follows that~$\arrivalTime(\targetEvent)-\candidateArrivalTime(\targetVertex)\geq{}0$.
\end{proof}

\subparagraph*{2\targetIndex-Split Limit.}
Determining the set of all indirect~$(\originEvent,\targetEvent)$-witnesses that strongly dominate~$\aCandidateJourney$ in~$\bestDelayScenario(\sourcePrefix)$ is challenging because not all of them are Pareto-optimal in~$\bestDelayScenario(\sourcePrefix)$.
Finding all of them would potentially require examining many different delay scenarios, which would be expensive.
To avoid this, we do not compute the~2\targetIndex-split limit~$\indirectTargetSplitLimit(\targetEvent)$ exactly.
Instead, we compute a lower bound~$\indirectTargetSplitLimitBound(\targetEvent)$ that only considers the~indirect~$(\originEvent,\targetEvent)$-witness with the lowest arrival time~$\bestDelayScenario(\sourcePrefix)$.
Other~indirect~$(\originEvent,\targetEvent)$-witnesses, which may have a higher arrival time but a lower maximum witness delay, are ignored.
For a vertex~$\aVertex$, let~$\aWitnessJourney(\aVertex)$ be the journey that minimizes the arrival time~$\witnessArrivalTime(\originPrefix,\aVertex,2)$ in~$\bestDelayScenario(\originPrefix)$ and let~$\maxWitnessDelay(\aVertex)$ be its maximum witness delay.
Given destination and target events~$\destinationEvent\in\optimalDestinationEvents$ and~$\targetEvent\in\virtualTargetEvents(\destinationEvent)$, Lemma~\ref{th:app:target-split-limit-lower-bound} shows that a lower bound for the 2\targetIndex-split limit is given by
\begin{equation}\label{eqn:2target-split-limit}
\indirectTargetSplitLimitBound(\targetEvent)=\begin{cases}
	0 & \text{if } \arrivalTime(\targetEvent)\leq\witnessArrivalTime(\originPrefix,\astop(\targetEvent),2),\\
	\maxWitnessDelay(\astop(\targetEvent))+1 & \text{otherwise.}
\end{cases}
\end{equation}

\begin{lemma}
	\label{th:app:target-split-limit-lower-bound}
	Let~$\aCandidateJourney=\candidateSequence$ be a candidate for which~$\bestDelayScenario(\aCandidateJourney)$ is full-avoiding.
	Then
	\[\indirectTargetSplitLimitBound(\targetEvent)\leq\indirectTargetSplitLimit(\targetEvent).
	\]
\end{lemma}
\begin{proof}
	If~$\arrivalTime(\targetEvent)\leq\witnessArrivalTime(\originPrefix,\astop(\targetEvent),2)$, then~$\indirectTargetSplitLimitBound(\targetEvent)=0$, so the claim is trivially true.
	We therefore assume that~$\witnessArrivalTime(\originPrefix,\astop(\targetEvent),2)<\arrivalTime(\targetEvent)$.
	The former is the arrival time of~$\aWitnessJourney$ in~$\bestDelayScenario(\originPrefix)$ and the latter the arrival time of~$\aCandidateJourney$ in~$\bestDelayScenario$, so it follows that~$\aWitnessJourney$ strongly dominates~$\aCandidateJourney$ in~$\virtualDelayScenario(\originPrefix)=(\bestDelayScenario,\bestDelayScenario(\originPrefix))$.
	Since~$\bestDelayScenario\leqEval\bestDelayScenario(\aCandidateJourney)\leqEval\bestDelayScenario(\originPrefix)$, it follows that~$\virtualDelayScenario(\originPrefix)\leqDom\bestDelayScenario(\aCandidateJourney)$.
	Hence, $\aWitnessJourney$ also strongly dominates~$\aCandidateJourney$ in~$\bestDelayScenario(\aCandidateJourney)$.
	Since~$\bestDelayScenario(\aCandidateJourney)$ is full-avoiding, $\aWitnessJourney$ must be an indirect~$(\originEvent,\targetEvent)$-witness.
	By definition of~$\maxWitnessDelay(\astop(\targetEvent))=\maxWitnessDelay(\aWitnessJourney)$, $\aWitnessJourney$ is still feasible in~$\bestDelayScenario(\sourcePrefix,\maxWitnessDelay(\astop(\targetEvent)))$ and strongly dominates~$\aCandidateJourney$.
	Hence, $\indirectTargetSplitLimit(\targetEvent)>\maxWitnessDelay(\astop(\targetEvent))$.
\end{proof}

\subparagraph*{Aggregating the Split Limits.}
The formulas established for the split limits allows us to calculate a lower bound for the minimum origin delay~$\minOriginDelay(\aCandidateJourney)$ of a candidate~$\aCandidateJourney$.
To obtain a lower bound for the minimum origin delay~$\minOriginDelay(\destinationEvent)$ of a destination event~$\destinationEvent\in\optimalDestinationEvents$, the values for all candidates in~$\optimalCandidateJourneys(\destinationEvent)$ must be aggregated.
For this purpose, we define the \emph{aggregated \targetIndex-split limit}
\begin{equation}\label{eqn:agg-target-split-limit}
	\aggTargetSplitLimitBound(\destinationEvent):=\min\limits_{\targetEvent\in\virtualTargetEvents(\destinationEvent)}\max(\directTargetSplitLimit(\targetEvent),\indirectTargetSplitLimitBound(\targetEvent)).
\end{equation}
Lemma~\ref{th:app:min-origin-delay-bound} shows that~$\minOriginDelayBound(\destinationEvent):=\max(\destinationSplitLimit(\destinationEvent),\aggTargetSplitLimitBound(\destinationEvent))$ is a lower bound for~$\minOriginDelay(\destinationEvent)$.

\begin{lemma}
	\label{th:app:min-origin-delay-bound}
	For a destination event~$\destinationEvent\in\optimalDestinationEvents$, $\minOriginDelayBound(\destinationEvent)\leq\minOriginDelay(\destinationEvent)$.
\end{lemma}
\begin{proof}
	Let~$\optimalTargetEvents(\destinationEvent)$ be the set of target events that occur in~$\optimalCandidateJourneys(\destinationEvent)$.
	We observe that
	\begin{align*}
		&\minOriginDelay(\destinationEvent)=\min\limits_{\aCandidateJourney\in\optimalCandidateJourneys(\destinationEvent)}\minOriginDelay(\aCandidateJourney)=\max\left(\destinationSplitLimit(\destinationEvent),\min\limits_{\targetEvent\in\optimalTargetEvents(\destinationEvent)}\max(\directTargetSplitLimit(\targetEvent),\indirectTargetSplitLimitBound(\targetEvent))\right).
	\end{align*}
	It remains to be shown that
	\begin{align*}
		\aggTargetSplitLimitBound(\destinationEvent)&=\min\limits_{\targetEvent\in\virtualTargetEvents(\destinationEvent)}\max(\directTargetSplitLimit(\targetEvent),\indirectTargetSplitLimitBound(\targetEvent))\leq\min\limits_{\targetEvent\in\optimalTargetEvents(\destinationEvent)}\max(\directTargetSplitLimit(\targetEvent),\indirectTargetSplitLimit(\targetEvent)).
	\end{align*}
	This is the because since~$\virtualTargetEvents(\destinationEvent)\supseteq\optimalTargetEvents(\destinationEvent)$ by Lemma~\ref{th:app:optimal-virtual-candidates} and~$\indirectTargetSplitLimitBound(\targetEvent)\leq\indirectTargetSplitLimit(\targetEvent)$ by Lemma~\ref{th:app:target-split-limit-lower-bound}.
\end{proof}

\subsection{Preventing Time Travel}
\label{sec:time-travel}
So far, we have focused on finding shortcuts that are required in at least one delay scenario.
However, many delay scenarios cannot occur in reality because they require public transit vehicles to travel faster than their maximum speed, or even backwards in time.
Unfortunately, prohibiting time travel introduces dependencies between the delays of stop events.
This makes it much more complicated to characterize the conditions under which a candidate is prefix-optimal.
We therefore eschew a full characterization and limit ourselves to two optimizations that prevent some types of time travel within the two candidate trips.

\subparagraph{First Trip.}
Let~$\originPrefix=\originPrefixSequence$ be an origin prefix with~$\originEvent=\aTrip[i]$.
By Lemma~\ref{th:app:join-limit-virt}, the join limit~$\joinLimit(\destinationVertex)$ is non-negative iff~$\originPrefix$ is not strongly dominated by a join witness in the virtual delay scenario~$\virtualDelayScenario(\sourcePrefix)$.
Consider a join witness of the form~$\aWitnessJourney=\langle[\sourceEvent,\aTrip[j]],\destinationVertex\rangle$ with~$j<i$.
Then~$\virtualDelayScenario(\sourcePrefix)$ assumes maximum arrival delay for~$\aTrip[j]$, but no arrival delay for~$\aTrip[i]$.
If the difference in the arrival times of both stop events is less than~$\maxDelay$, this is physically impossible because~$\aTrip$ would have to travel backwards in time.
To avoid this type of time travel, we define the~\emph{time travel delay scenario}~$\timeTravelDelayScenario(\aTrip[i])$.
For stop events~$\aTrip[i]$ and~$\aTrip[j]$ of the same trip~$\aTrip$, we define the arrival time difference
\[\arrivalTimeDifference(\aTrip,i,j):=\min(\maxDelay,\arrivalTime(\aTrip[i])-\arrivalTime(\aTrip[j])).\]
Then the departure and arrival delay of a stop event~$\stopEvent$ in~$\timeTravelDelayScenario(\aTrip[i])$ are given by
\begin{align}\label{eqn:time-travel-scenario}
	\begin{split}
		\departure{\timeTravelDelayScenario}(\aTrip[i])(\stopEvent)&=\departure{\bestDelayScenario}(\sourcePrefix)(\stopEvent)\\
		\arrival{\timeTravelDelayScenario}(\aTrip[i])(\stopEvent)&=\begin{cases}
			\arrivalTimeDifference(\aTrip,i,j) & \text{if }\stopEvent=\aTrip[j]\text{ with }j<i,\\
			\maxDelay & \text{otherwise.}
		\end{cases}
	\end{split}
\end{align}

For each vertex~$\aVertex$, let~$\timeTravelArrivalTime(\aTrip[i],\aVertex)$ be the earliest arrival time at~$\aVertex$ among journeys with at most one trip that depart no earlier than~$\departureTime(\sourcePrefix,\bestDelayScenario)$ in~$\timeTravelDelayScenario(\aTrip[i])$.

\begin{lemma}
	\label{th:time-travel-scenario}
	Let~$\aCandidateJourney=\candidateSequence$ be a candidate and~$\aVertex\in\stops$ a vertex visited by the intermediate transfer of~$\aCandidateJourney$.
	If~$\candidateArrivalTime(\aVertex)>\timeTravelArrivalTime(\aTrip[i],\aVertex)$, then~$\aCandidateJourney$ is not prefix-optimal in any delay scenario without time travel.
\end{lemma}
\begin{proof}
	If~$\candidateArrivalTime(\aVertex)>\witnessArrivalTime(\sourcePrefix,\aVertex,1)$, i.e.,~$\joinLimit(\aVertex)<0$, then~$\joinLimit(\aCandidateJourney)<0$ by Lemma~\ref{th:app:join-limit-negative} and therefore~$\aCandidateJourney$ is not prefix-optimal in any delay scenario.
	Since~$\timeTravelDelayScenario(\aTrip[i])\leqEval\bestDelayScenario(\sourcePrefix)$, we know that~$\timeTravelArrivalTime(\aTrip[i],\aVertex)\leq\witnessArrivalTime(\sourcePrefix,\aVertex,1)$.
	Assume therefore that~$\timeTravelArrivalTime(\aTrip[i],\aVertex)<\candidateArrivalTime(\aVertex)\leq\witnessArrivalTime(\sourcePrefix,\aVertex,1)$.
	Let~$\aTrip[i]:=\originEvent$ and let~$\aWitnessJourney$ be a $\sourceVertex$-$\aVertex$-journey with at most one trip and arrival time~$\timeTravelArrivalTime(\aTrip[i],\aVertex)$ at~$\aVertex$ that departs no later than~$\aCandidateJourney$ in~$\timeTravelDelayScenario(\aTrip[i])$.
	Since the arrival time of~$\aWitnessJourney$ is earlier in~$\timeTravelDelayScenario(\aTrip[i])$ than in~$\bestDelayScenario(\sourcePrefix)$, the final stop event of~$\aWitnessJourney$ must be some~$\aTrip[j]$ with~$j<i$ and~$\arrivalTime(\aTrip[i])-\arrivalTime(\aTrip[j])<\maxDelay$.    
	In any delay scenario~$\delayScenario$ without time travel, $\arrivalTime(\delayScenario,\aTrip[j])\leq\arrivalTime(\delayScenario,\aTrip[i])$ must hold.
	This is equivalent to~$\arrivalDelay(\aTrip[j])-\arrivalDelay(\aTrip[i])\leq\arrivalTime(\aTrip[i])-\arrivalTime(\aTrip[j])=\arrival{\timeTravelDelayScenario}(\aTrip[i])(\aTrip[j])$.
	It follows that
	\begin{alignat*}{3}
		&\arrivalTime(\delayScenario,\aWitnessJourney)&&=&&\ \arrivalTime(\timeTravelDelayScenario(\aTrip[i]),\aWitnessJourney)+\arrivalDelay(\aTrip[j])-\arrival{\timeTravelDelayScenario}(\aTrip[i])(\aTrip[j])\\
		&&&=&&\ \timeTravelArrivalTime(\aTrip[i],\aVertex)+\arrivalDelay(\aTrip[j])-\arrival{\timeTravelDelayScenario}(\aTrip[i])(\aTrip[j])\\
		&&&<&&\ \candidateArrivalTime(\aVertex)+\arrivalDelay(\aTrip[j])-\arrival{\timeTravelDelayScenario}(\aTrip[i])(\aTrip[j])\\
		&&&\leq&&\ \candidateArrivalTime(\aVertex)+\arrivalDelay(\aTrip[i])\\
		&&&=&&\ \arrivalTime(\delayScenario,\originPrefix(\aVertex)).
	\end{alignat*}
	Thus, $\aCandidateJourney$ is not prefix-optimal in~$\delayScenario$ because the prefix~$\originPrefix(\aVertex)$ is strongly dominated by~$\aWitnessJourney$.
\end{proof}

Lemma~\ref{th:time-travel-scenario} shows that if~$\candidateArrivalTime(\aVertex)>\timeTravelArrivalTime(\aTrip[i],\aVertex)$, then no candidate starting with~$\originPrefix(\aVertex)$ can be prefix-optimal in any delay scenario without time travel.
Because~$\timeTravelArrivalTime(\aTrip[i],\aVertex)\leq\witnessArrivalTime(\sourcePrefix,\aVertex,1)$, this also implies~$\joinLimit(\aVertex)<0$, so Lemma~\ref{th:time-travel-scenario} is a stronger version of Lemma~\ref{th:app:join-limit-negative}.

\subparagraph{Second Trip.}
Consider a candidate~$\aCandidateJourney=\candidateSequence$.
If~$\aCandidateJourney$ is Pareto-optimal in a delay scenario~$\delayScenario$, the arrival time of the earliest witness at~$\targetVertex$ in~$\delayScenario$ is an upper bound for the arrival time of~$\targetEvent$.
Because the second trip of~$\aCandidateJourney$ cannot travel backwards in time, this is also an upper bound for the arrival time of the origin prefix~$\originPrefix$ at~$\destinationVertex$.
We define the~\emph{time travel limit}
\[\timeTravelLimit(\destinationEvent,\targetEvent):=\witnessArrivalTime(\sourcePrefix,\astop(\targetEvent),2) -\candidateArrivalTime(\astop(\destinationEvent)).\]

\begin{lemma}
	\label{th:time-travel-limit}
	Let~$\aCandidateJourney=\candidateSequence$ be a candidate and~$\delayScenario$ a delay scenario without time travel.
	If~$\aCandidateJourney$ is prefix-optimal in~$\delayScenario$, then~$\arrivalTime(\delayScenario,\originEvent)\leq\timeTravelLimit(\destinationEvent,\targetEvent)$.
\end{lemma}
\begin{proof}
	Assume that~$\aCandidateJourney=\candidateSequence$ is prefix-optimal in~$\delayScenario$.
	Then~$\aCandidateJourney\in\optimalCandidateJourneys(\destinationEvent)\subseteq\virtualCandidateJourneys(\destinationEvent)$ and therefore~$\arrivalTime(\delayScenario,\targetEvent)\leq\witnessArrivalTime(\sourcePrefix,\targetVertex,2)$ must hold.
	Since no time travel occurs in~$\delayScenario$, it follows that~$\departureTime(\delayScenario,\destinationEvent)\leq\arrivalTime(\delayScenario,\targetEvent)$.
	Finally, since~$\aCandidateJourney$ is feasible, $\arrivalTime(\delayScenario,\originPrefix)=\candidateArrivalTime(\destinationVertex)+\arrivalDelay(\originEvent)\leq\departureTime(\delayScenario,\destinationEvent)$ must hold.
	Altogether, this yields~$\arrivalDelay(\originEvent)\leq\witnessArrivalTime(\sourcePrefix,\targetVertex,2)-\candidateArrivalTime(\destinationVertex)=\timeTravelLimit(\destinationEvent,\targetEvent)$.
\end{proof}

For a destination event~$\destinationEvent\in\optimalDestinationEvents$, we define the~\emph{maximum witness time}
\begin{equation}\label{eqn:max-witness-time}
	\maxWitnessTime(\destinationEvent):=\max\limits_{\targetEvent\in\virtualTargetEvents(\destinationEvent)}\witnessArrivalTime(\sourcePrefix,\astop(\targetEvent),2)
\end{equation}
and the~\emph{aggregated time travel limit}
\begin{equation}\label{eqn:agg-time-travel-limit}
	\aggTimeTravelLimit(\destinationEvent):=\max\limits_{\targetEvent\in\virtualTargetEvents(\destinationEvent)}\timeTravelLimit(\destinationEvent,\targetEvent)=\maxWitnessTime(\destinationEvent)-\candidateArrivalTime(\astop(\destinationEvent)).
\end{equation}
If the delay of~$\originEvent$ in a delay scenario~$\delayScenario$ without time travel exceeds~$\aggTimeTravelLimit(\destinationEvent)$, then by Lemma~\ref{th:time-travel-limit} the shortcut~$(\originEvent,\destinationEvent)$ is not required in~$\delayScenario$.
We therefore redefine the maximum origin delay as
\[\maxOriginDelay(\destinationEvent):=\min\{\feasibilityLimit(\destinationEvent),\joinLimit(\astop(\destinationEvent)),\aggTimeTravelLimit(\destinationEvent)\}.\]

\section{Delay-ULTRA Shortcut Computation}
\label{sec:delay-ultra}
The Delay-ULTRA shortcut computation retains the basic outline of ULTRA:
For each source stop~$\sourceVertex$ (in parallel), it enumerates all prefix-optimal candidates originating at~$\sourceVertex$.
For each possible candidate departure time~$\atime_j$ at~$\sourceVertex$ (in descending order), candidates departing at~$\atime_j$ and witnesses departing within~$[\atime_j,\atime_{j+1})$ are explored.
ULTRA does this with a \emph{run} of the MR algorithm, which is limited to the first two rounds.
Delay-ULTRA must additionally take different delay scenarios into account.
Section~\ref{sec:optimality-tests} showed that three delay scenarios are sufficient to examine a candidate~$\aCandidateJourney$:
The candidate itself is explored in~$\bestDelayScenario$.
Witnesses are explored in~$\bestDelayScenario(\sourcePrefix)$ and~$\bestDelayScenario(\originPrefix)$, yielding arrival times~$\witnessArrivalTime(\sourcePrefix,\cdot,\cdot)$ and~$\witnessArrivalTime(\originPrefix,\cdot,\cdot)$, respectively.
While~$\bestDelayScenario$ is the same for all candidates, the witness delay scenarios depend on the candidate prefixes.
Thus, instead of a single witness search per run, Delay-ULTRA performs one for each source and origin prefix.

\subparagraph*{Self-Pruning.}
To avoid redundant work, Delay-ULTRA performs the witness searches in a staggered manner and extends the \emph{self-pruning} approach used by ULTRA.
A run with departure time~$\atime_j$ begins with a two-round MR witness search in~$\worstDelayScenario$.
For each vertex~$\aVertex$ and round~$n\leq{}2$, this finds the earliest arrival time~$\witnessArrivalTime(\atime_j,\aVertex,n)$ among journeys to~$\aVertex$ with at most~$n$ trips departing no earlier than~$\atime_j$.
The original self-pruning rule is applied here:
Because~$\witnessArrivalTime(\atime_j,\aVertex,n)\leq\witnessArrivalTime(\atime_{j+1},\aVertex,n)$ holds for each vertex~$\aVertex$ and round~$n$, the search for~$\atime_j$ initializes the former with the latter and only explores journeys departing before~$\atime_{j+1}$.
Since~$\witnessArrivalTime(\sourcePrefix,\aVertex,n)\leq\witnessArrivalTime(\atime_j,\aVertex,n)$ holds as well, the witness search in~$\bestDelayScenario(\sourcePrefix)$ reuses the results for~$\worstDelayScenario$ in the same manner.
Moreover, since the two scenarios only differ in the departure delay of~$\sourceEvent$, the witness search in~$\bestDelayScenario(\sourcePrefix)$ only explores journeys starting with~$\sourcePrefix$.
By the same argument, the witness search for~$\bestDelayScenario(\originPrefix)$ reuses the results for~$\bestDelayScenario(\sourcePrefix)$ and only explores journeys starting with~$\originPrefix$.
Thus, each witness search is pruned with results from the previous ones.

A similar approach is used for the entry index~$\entryIndex(\sourcePrefix,\aTrip)$ of each trip~$\aTrip$.
Whereas~$\entryIndex(\sourcePrefix,\aTrip)$ is based on witness arrival times in~$\bestDelayScenario(\sourcePrefix)$, we define~$\entryIndex(\atime_j,\aTrip)$ as the entry index based on witnesses in~$\worstDelayScenario$ that depart no earlier than~$\atime_j$.
To compute it, we augment the MR witness search as follows:
whenever a stop event~$\aTrip[i]$ is entered during a route scan, $\entryIndex(\atime_j,\aTrip')$ is updated to~$\min(\entryIndex(\atime_j,\aTrip'),i)$ for all trips~$\aTrip\preceq\aTrip'$.
Based on the inequalities~$\entryIndex(\sourcePrefix,\aTrip)\geq\entryIndex(\atime_j,\aTrip)\geq\entryIndex(\atime_{j+1},\aTrip)$, self-pruning is applied to the entry indices as well.

\subsection{Overview}
\begin{algorithm}[t]
	\caption{Delay-ULTRA shortcut computation.}\label{alg:overview}
	\Input{Public transit network~$(\stops,\stopEvents,\trips,\routes,\graph)$, source stop~$\sourceVertex\in\stops$}
	\Output{Shortcuts~$\shortcutEdges$}
	$\departures\leftarrow$ Collect departures of trips at~$\sourceVertex$\label{alg:overview:departure}\;
	\ForEach{$(\departureTime,\sourceEvents)\in\departures$ in descending order of~$\departureTime$\label{alg:overview:departure-loop}}{
		$\witnessArrivalTime(\departureTime,\cdot,\cdot),\entryIndex(\departureTime,\cdot)\leftarrow$ Find witnesses in~$\worstDelayScenario$\label{alg:overview:regular-witnesses}\;
		\ForEach{$\sourceEvent\in\sourceEvents$\label{alg:overview:source-loop}}{
			\tcp{Solve~\textsc{DelayShortcut}-$\originPrefix(\astop(\originEvent))$}
			$\witnessArrivalTime(\sourcePrefix,\cdot,\cdot),\entryIndex(\sourcePrefix,\cdot)\leftarrow$ Find witnesses in~$\bestDelayScenario(\sourcePrefix)$\label{alg:overview:shared-source}\;
			\ForEach{$\originEvent\in\succeedingEvents(\sourceEvent)$\label{alg:overview:origin-loop}}{
				$\candidateArrivalTime(\cdot),\joinLimit(\cdot)\leftarrow$ Explore intermediate transfers of candidates\label{alg:overview:intermediate}\;
				$\exitIndex(\sourcePrefix,\cdot)\leftarrow$ Compute exit indices\label{alg:overview:exit-indices}\;
				$\optimalDestinationEvents,\feasibilityLimit(\cdot)\leftarrow$ Find destination events\label{alg:overview:register}\;
				$\destinationSplitLimit(\cdot)\leftarrow$ Compute \destinationIndex-split limits\label{alg:overview:destination-split-limit}\;
				$\witnessArrivalTime(\originPrefix,\cdot,\cdot),\maxWitnessDelay(\cdot)\leftarrow$ Find witnesses in~$\bestDelayScenario(\originPrefix)$\label{alg:overview:shared-origin}\;
				$\newShortcutEdges,\minOriginDelayBound(\cdot),\maxOriginDelay(\cdot)\leftarrow$ Scan second trips of candidates\label{alg:overview:route2}\;
				Merge~$\newShortcutEdges$ into~$\shortcutEdges$\label{alg:overview:merge}\;
			}
		}
	}
\end{algorithm}
Algorithm~\ref{alg:overview} shows a high-level overview of the shortcut computation for a source stop~$\sourceVertex$.
Further details are provided in Section~\ref{sec:delay-ultra:details}.
Line~\ref{alg:overview:departure} collects all stop events departing at~$\sourceVertex$ in a set~$\departures$, grouped by their departure time.
For a tuple~$(\departureTime,\sourceEvents)\in\departures$, $\sourceEvents$ contains all stop events departing at~$\sourceVertex$ with scheduled departure time~$\departureTime-\maxDelay$.
These tuples are processed in descending order of departure time.
Line~\ref{alg:overview:regular-witnesses} computes~$\witnessArrivalTime(\departureTime,\cdot,\cdot)$ and~$\entryIndex(\departureTime,\cdot)$ with a witness search in~$\worstDelayScenario$.
For each source event~$\sourceEvent\in\sourceEvents$, the algorithm then explores candidates beginning with the source prefix~$\sourcePrefix=\sourcePrefixSequence$.
Line~\ref{alg:overview:shared-source} computes~$\witnessArrivalTime(\sourcePrefix,\cdot,\cdot)$ and~$\entryIndex(\sourcePrefix,\cdot)$ with a self-pruned witness search in~$\bestDelayScenario(\sourcePrefix)$.
For each event~$\originEvent$ succeeding~$\sourceEvent$ in~$\aTrip(\sourceEvent)$, the algorithm now solves the subproblem~\textsc{DelayShortcut}-$\originPrefix(\originVertex)$ for the origin stop prefix~$\originPrefix(\originVertex)=\originStopPrefixSequence{\originVertex}$ with origin stop~$\originVertex=\astop(\originEvent)$.

Line~\ref{alg:overview:intermediate} performs a Dijkstra search from~$\originVertex$ in~$\bestDelayScenario$.
For each vertex~$\aVertex$, this computes the candidate arrival time~$\candidateArrivalTime(\aVertex)$ and the join limit~$\joinLimit(\aVertex)$ by Equation~\ref{eqn:join-limit}.
For each trip~$\aTrip$ that can be entered by a candidate, line~\ref{alg:overview:exit-indices} computes the exit index~$\exitIndex(\sourcePrefix,\aTrip)$.
For each event~$\destinationEvent:=\aTrip[i]$ with~$i<\witnessIndex(\sourcePrefix,\aTrip)$, line~\ref{alg:overview:register} computes the feasibility limit~$\feasibilityLimit(\destinationEvent)$ by Equation~\ref{eqn:feasibility-limit}.
If it is non-negative, then~$\destinationEvent$ is added to~$\optimalDestinationEvents$.
Line~\ref{alg:overview:destination-split-limit} computes the \destinationIndex-split limits for the events in~$\optimalDestinationEvents$ by Lemma~\ref{th:app:destination-split-limit}.
For the \targetIndex-split limits, line~\ref{alg:overview:shared-origin} computes the witness arrival times~$\witnessArrivalTime(\originPrefix,\cdot,\cdot)$ and the maximum witness delays~$\maxWitnessDelay(\cdot)$ with a modified MR search in~$\bestDelayScenario(\originPrefix)$.
For each destination event~$\destinationEvent\in\optimalDestinationEvents$, line~\ref{alg:overview:route2} calculates the aggregated \targetIndex-split limit~$\aggTargetSplitLimitBound(\destinationEvent)$ by Equation~\ref{eqn:agg-target-split-limit} and the aggregated time travel limit~$\aggTimeTravelLimit(\destinationEvent)$ by Equation~\ref{eqn:agg-time-travel-limit}.
From these, the maximum origin delay~$\maxOriginDelay(\destinationEvent)$ and a lower bound~$\minOriginDelayBound(\destinationEvent)$ for the minimum origin delay are computed.
If~$\minOriginDelayBound(\destinationEvent)\leq\maxOriginDelay(\destinationEvent)$, the shortcut~$(\originEvent,\destinationEvent)$ is generated and merged into the result set~$\shortcutEdges$ in line~\ref{alg:overview:merge}.

\subsection{Details}
\label{sec:delay-ultra:details}
We now discuss further details that are not shown in Algorithm~\ref{alg:overview}.

\subparagraph*{Exploring Intermediate Transfers.}
Line~\ref{alg:overview:intermediate} of Algorithm~\ref{alg:overview} explores the intermediate transfers of candidates beginning with the origin stop prefix~$\originPrefix(\originVertex)$.
This is done with a Dijkstra search starting from~$\originVertex$ in the best-case delay scenario~$\bestDelayScenario$.
For each vertex~$\aVertex$ visited by the search, this yields the candidate arrival time~$\candidateArrivalTime(\aVertex)$ and the join limit~$\joinLimit(\aVertex)$ by Equation~\ref{eqn:join-limit}.
By Lemma~\ref{th:app:join-limit-negative}, if~$\joinLimit(\aVertex)<0$, then all candidates beginning with~$\originPrefix(\aVertex)$ are irrelevant because they have a negative join limit as well.
In this case, the search is pruned at~$\aVertex$, i.e., the outgoing edges are not explored.
For an even stronger pruning rule, a Dijkstra search from~$\originVertex$ in the time travel delay scenario~$\timeTravelDelayScenario(\originEvent)$ can be performed before line~\ref{alg:overview:intermediate} to compute~$\timeTravelArrivalTime(\originEvent,\aVertex)$ for each vertex~$\aVertex$.
Then by Lemma~\ref{th:time-travel-scenario}, the candidate search can be pruned at a vertex~$\aVertex$ if~$\timeTravelArrivalTime(\originEvent,\aVertex)<0$.\looseness=-1

\subparagraph*{Exit Indices.}
Line~\ref{alg:overview:exit-indices} of Algorithm~\ref{alg:overview} collects all routes that can be entered by a candidate and calculates the exit indices of their trips.
A route~$\aRoute$ can be entered if its last trip~$\aTrip_\text{max}$ can be entered, i.e., there is at least one stop event~$\aTrip_\text{max}[i]$ such that~$\candidateArrivalTime(\astop(\aTrip_\text{max}[i]))\leq\departureTime(\aTrip_\text{max}[i])+\maxDelay$.
For each trip~$\aTrip$ of~$\aRoute$, the exit index~$\exitIndex(\sourcePrefix,\aTrip)$ is then calculated.
Since this value is based on the delay scenario~$\bestDelayScenario(\sourcePrefix)$, which does not depend on the choice of~$\originEvent$, the value of~$\exitIndex(\sourcePrefix,\aTrip)$ cannot change between iterations of the loop in line~\ref{alg:overview:origin-loop}.
To avoid unnecessary recalculations, the algorithm maintains a timestamp for each route~$\aRoute$ that indicates whether the exit indices for the trips of~$\aRoute$ were already calculated for the current choice of~$\sourceEvent$.\looseness=-1

\subparagraph*{Finding Destination Events.}
\begin{algorithm}[t]
	\caption{Find destination events.}\label{alg:registerDestination}
	\Input{Route~$\aRoute$, first reachable stop index~$j$}
	\Output{Destination events~\optimalDestinationEvents,
		feasibility limits~$\feasibilityLimit(\cdot)$}
	$\tripBegin\leftarrow$ first trip of~$\aRoute$\label{alg:registerDestination:setBegin}\;
	$\tripEnd\leftarrow$ last trip of~$\aRoute$\label{alg:registerDestination:setEnd}\;
	\For{$j\leq{}i<\witnessIndex(\sourcePrefix,\tripBegin)$ in ascending order}{
		\While{$\witnessIndex(\sourcePrefix,\tripEnd)\leq{}i$\label{alg:registerDestination:updateEnd}}{$\tripEnd\leftarrow$ predecessor of~$\tripEnd$ in~$\aRoute$}
		\For{$\tripEnd\succeq\aTrip\succeq\tripBegin$ in descending order}{
			$\destinationEvent\leftarrow\aTrip[i]$\;
			$\destinationVertex\leftarrow\astop(\destinationEvent)$\;
			$\feasibilityLimit(\destinationEvent)\leftarrow\min(0,\departureTime(\destinationEvent)-\candidateArrivalTime(\destinationVertex))+\maxDelay$\label{alg:registerDestination:feasibilityLimit}\;
			\lIf{$\feasibilityLimit(\destinationEvent)<0$}{\Break\label{alg:registerDestination:canEnter}}
			$\optimalDestinationEvents\leftarrow\optimalDestinationEvents\cup\{\destinationEvent\}$\label{alg:registerDestination:addEvent}\;
		}
	}
\end{algorithm}

Line~\ref{alg:overview:register} of Algorithm~\ref{alg:overview} computes the set~$\optimalDestinationEvents$ of relevant destination events and the feasibility limit~$\feasibilityLimit(\destinationEvent)$ for each destination event~$\destinationEvent\in\optimalDestinationEvents$.
By Lemma~\ref{th:app:full-witness-condition}, a stop event~$\aTrip[i]\in\feasibleDestinationEvents$ is contained in~$\optimalDestinationEvents$ iff~$i<\witnessIndex(\sourcePrefix,\aTrip)$.
To find the stop events fulfilling this condition efficiently, each route collected in line~\ref{alg:overview:exit-indices} is processed individually.
Detailed pseudocode for the examination of a route~$\aRoute$ is given in Algorithm~\ref{alg:registerDestination}.
Lines~\ref{alg:registerDestination:setBegin} and~\ref{alg:registerDestination:setEnd} initialize the variables~$\tripBegin$ and~$\tripEnd$, which initially refer to the first and last trip of~$\aRoute$, respectively.
The algorithm then iterates over the stops of~$\aRoute$, starting at the first stop that was reached by the Dijkstra search in line~\ref{alg:overview:intermediate} of Algorithm~\ref{alg:overview} and ending at~$\witnessIndex(\sourcePrefix,\tripBegin)-1$.
For each stop index~$i$, line~\ref{alg:registerDestination:updateEnd} decreases~$\tripEnd$ until~$\witnessIndex(\sourcePrefix,\tripEnd)>i$ holds.
This ensures that the trips after~$\tripEnd$, which contain no further stop events in~$\optimalDestinationEvents$, are ignored.
Then, in descending order from~$\tripEnd$ to~$\tripBegin$, the algorithm examines for each trip~$\aTrip$ the destination event~$\destinationEvent:=\aTrip[i]$ departing at the current stop.
Line~\ref{alg:registerDestination:feasibilityLimit} computes the feasibility limit~$\feasibilityLimit(\destinationEvent)$ according to Equation~\ref{eqn:feasibility-limit}.
If this is below~$0$, then~$\destinationEvent\notin\feasibleDestinationEvents$, so it can be discarded.
Since the feasibility limit will also be below~$0$ for all of the trips preceding~$\aTrip$, the loop is exited altogether.
Otherwise, $\destinationEvent$ is added to~$\optimalDestinationEvents$ in line~\ref{alg:registerDestination:addEvent}.\looseness=-1

\subparagraph*{$\destinationIndex$-Split Limit.}
Line~\ref{alg:overview:destination-split-limit} of Algorithm~\ref{alg:overview} computes the~$\destinationIndex$-split limits for the events in~$\optimalDestinationEvents$.
By Equation~\ref{th:app:destination-split-limit}, the~$\destinationIndex$-split limit of a destination event~$\aTrip[i]$ depends on the maximum witness delays of the events in~$\optimalDestinationEvents$ that precede it in~$\aTrip$.
For each trip~$\aTrip$ with at least one stop event in~$\optimalDestinationEvents$, the algorithm sweeps over its stop events from first to last and maintains a value~$\destinationSplitLimit(\aTrip)$, which is initalized with~$0$.
If a stop event~$\aTrip[i]$ is not in~$\optimalDestinationEvents$, it is skipped.
Otherwise, its~$\destinationIndex$-split limit~$\destinationSplitLimit(\aTrip[i])$ is set to~$\destinationSplitLimit(\aTrip)$.
Afterwards, the maximum witness delay~$\maxWitnessDelay(\aTrip[i])$ is calculated according to Equation~\ref{eqn:max-witness-delay} and~$\destinationSplitLimit(\aTrip)$ is set to the maximum of itself and~$\maxWitnessDelay(\aTrip[i])+1$.

\subparagraph*{Modified MR Search for~$(\originEvent,\targetEvent)$-Witnesses.}
Line~\ref{alg:overview:shared-origin} of Algorithm~\ref{alg:overview} performs a witness search in~$\bestDelayScenario(\originPrefix)$ to calculate the witness arrival times~$\witnessArrivalTime(\originPrefix,\cdot,\cdot)$ and the maximum witness delays~$\maxWitnessDelay(\cdot)$.
This is done with an MR search that is modified as follows.
Since~$\witnessArrivalTime(\originPrefix,\aVertex,1)=\min(\witnessArrivalTime(\sourcePrefix,\aVertex,1),\candidateArrivalTime(\aVertex))$ holds for each vertex~$\aVertex$, the first MR round is skipped and~$\witnessArrivalTime(\originPrefix,\aVertex,1)$ is initialized accordingly.
Afterwards, a second MR round is performed to find second trips and final transfers.
For each vertex~$\aVertex$, this search calculates the earliest arrival time~$\witnessArrivalTime(\originPrefix,\aVertex,2)$ at~$\aVertex$ among journeys with at most~$2$ trips in~$\bestDelayScenario(\originPrefix)$, as well as maximum witness delay~$\maxWitnessDelay(\aVertex)$ of the corresponding witness.
This search finds exactly those witnesses that include an intermediate transfer from~$\originEvent$ to a destination event~$\destinationEvent'$.
Since all other witnesses have maximum witness delay~$\maxDelay$, $\maxWitnessDelay(\aVertex)$ is initialized with~$\maxDelay$ for each vertex~$\aVertex$ at the start of the search.

To calculate the maximum witness delays, the scanning procedure for a route~$\aRoute$ is modified as follows.
In addition to the active trip~$\activeTrip$, it maintains an \emph{active maximum witness delay}~$\maxWitnessDelay(\activeTrip)$, which is initialized with~$\maxDelay$.
If~$\activeTrip$ is improved by entering~$\aRoute$ at a destination event~$\destinationEvent'$, then~$\maxWitnessDelay(\activeTrip)$ is set to~$\maxWitnessDelay(\destinationEvent')$, which is calculated according to Equation~\ref{eqn:max-witness-delay}.
If exiting~$\aRoute$ at a stop~$\astop$ improves the arrival time~$\witnessArrivalTime(\originPrefix,\astop,2)$, then~$\maxWitnessDelay(\astop)$ is set to~$\maxWitnessDelay(\activeTrip)$.
During the Dijkstra search, when an edge~$(\aVertex,\bVertex)$ is relaxed and~$\witnessArrivalTime(\originPrefix,\bVertex,2)$ is improved, then~$\maxWitnessDelay(\bVertex)$ is set to~$\maxWitnessDelay(\aVertex)$.
When a new iteration of the loop in line~\ref{alg:overview:origin-loop} is started, i.e., a new origin event is processed, $\witnessArrivalTime(\originPrefix,\cdot,\cdot)$ and~$\maxWitnessDelay(\cdot)$ are cleared.

\subparagraph*{Second Candidate Trip Scan.}
\begin{algorithm}[t]
	\caption{Scan second candidate trip.}\label{alg:candidateTrip2}
	\Input{Trip~$\aTrip$, destination events~$\optimalDestinationEvents$}
	\Output{New shortcut edges~$\newShortcutEdges$, lower bounds~$\minOriginDelayBound(\cdot)$ for the minimum origin delay, maximum origin delays~$\maxOriginDelay(\cdot)$}
	$\aggTargetSplitLimitBound(\aTrip)\leftarrow\infty$\label{alg:candidateTrip2:initTargetSplitLimit}\;
	$\maxWitnessTime(\aTrip)\leftarrow-\infty$\label{alg:candidateTrip2:initMaxWitnessTime}\;
	$j\leftarrow\min\{i\mid\aTrip[i]\in\optimalDestinationEvents\}$\label{alg:candidateTrip2:calculateStart}\;
	\For{$\exitIndex(\sourcePrefix,\aTrip)>i\geq{}j$ in descending order}{
		$\targetEvent\leftarrow\aTrip[i+1]$\;
		$\targetVertex\leftarrow\astop(\targetEvent)$\;
		\If{$\arrivalTime(\targetEvent)\leq\witnessArrivalTime(\sourcePrefix,\targetVertex,2)$\label{alg:candidateTrip2:checkTargetEvent}}{
			Calculate~$\directTargetSplitLimit(\targetEvent)$ by Lemma~\ref{th:app:direct-target-split-limit}\label{alg:candidateTrip2:1targetSplitLimit}\;
			Calculate~$\indirectTargetSplitLimitBound(\targetEvent)$ by Equation~\ref{eqn:2target-split-limit}\label{alg:candidateTrip2:2targetSplitLimit}\;
			$\aggTargetSplitLimitBound(\aTrip)\leftarrow\min(\aggTargetSplitLimitBound(\aTrip),\max(\directTargetSplitLimit(\targetEvent),\indirectTargetSplitLimitBound(\targetEvent)))$\label{alg:candidateTrip2:aggregateTargetSplitLimit}\;
			$\maxWitnessTime(\aTrip)\leftarrow\max(\maxWitnessTime(\aTrip),\witnessArrivalTime(\sourcePrefix,\targetVertex,2))$\label{alg:candidateTrip2:aggregateMaxWitnessTime}\;
		}
		$\destinationEvent\leftarrow\aTrip[i]$\label{alg:candidateTrip2:setDestinationEvent}\;
		$\destinationVertex\leftarrow\astop(\destinationEvent)$\;
		$\minOriginDelayBound(\destinationEvent)\leftarrow\max(\destinationSplitLimit(\destinationEvent),\aggTargetSplitLimitBound(\aTrip))$\label{alg:candidateTrip2:setMinOriginDelay}\;
		$\aggTimeTravelLimit(\destinationEvent)\leftarrow\maxWitnessTime(\aTrip)-\candidateArrivalTime(\destinationVertex)$\label{alg:candidateTrip2:calculateAggTimeTravelLimit}\;
		$\maxOriginDelay(\destinationEvent)\leftarrow\min\{\feasibilityLimit(\destinationEvent),\joinLimit(\destinationVertex),\aggTimeTravelLimit(\destinationEvent)\}$\label{alg:candidateTrip2:setMaxOriginDelay}\;
		\If{$\minOriginDelayBound(\destinationEvent)\leq\maxOriginDelay(\destinationEvent)$}{
			$\newShortcutEdges\leftarrow\newShortcutEdges\cup\{(\originEvent,\destinationEvent)\}$\label{alg:candidateTrip2:generateShortcut}\;
		}
	}
\end{algorithm}

Line~\ref{alg:overview:route2} of Algorithm~\ref{alg:overview} scans all trips that contain at least one destination event~$\destinationEvent\in\optimalDestinationEvents$.
This scan calculates a lower bound~$\minOriginDelayBound(\destinationEvent)$ for the minimum origin delay and the maximum origin delay~$\maxOriginDelay(\destinationEvent)$ and generates the shortcut~$(\originEvent,\destinationEvent)$ if it is required.
Detailed pseudocode for the scan of a trip~$\aTrip$ is given in Algorithm~\ref{alg:candidateTrip2}.
For each stop event~$\aTrip[i]\in\optimalDestinationEvents$, the algorithm must calculate a lower bound~$\aggTargetSplitLimitBound(\aTrip[i])$ for the aggregated target split limit by Equation~\ref{eqn:agg-target-split-limit} and the aggregated time travel limit~$\aggTimeTravelLimit(\aTrip[i])$ by Equation~\ref{eqn:agg-time-travel-limit}.
The latter in turn requires the maximum witness time~$\maxWitnessTime(\aTrip[i])$, which is calculated by Equation~\ref{eqn:max-witness-time}.
Computing these values requires aggregating over the set~$\virtualTargetEvents(\aTrip[i])$ of relevant target events.
According to Lemma~\ref{th:app:virtual-target-events}, this is given by~$\candidateExitEvents(\sourcePrefix,\aTrip)\cap\succeedingEvents(\aTrip[i])$, i.e., the set of c-out events succeeding~$\aTrip[i]$ in~$\aTrip$.

In order to consider exactly the set~$\virtualTargetEvents(\aTrip[i])$ for each destination event~$\aTrip[i]$, the procedure performs one sweep across the stop events of~$\aTrip$ in reverse, starting at the stop index~$\exitIndex(\sourcePrefix,\aTrip)-1$ and ending at the first stop index~$j$ for which~$\aTrip[j]\in\optimalDestinationEvents$.
The algorithm maintains the values~$\aggTargetSplitLimitBound(\aTrip)$ and~$\maxWitnessTime(\aTrip)$, which fulfill the invariants that~$\aggTargetSplitLimitBound(\aTrip)=\aggTargetSplitLimitBound(\aTrip[i])$ and~$\maxWitnessTime(\aTrip)=\maxWitnessTime(\aTrip[i])$ at the point when~$\aTrip[i]$ is scanned.
They are initialized with~$\infty$ in line~\ref{alg:candidateTrip2:initTargetSplitLimit} and~$-\infty$ in line~\ref{alg:candidateTrip2:initMaxWitnessTime}, respectively.
In the step for destination stop index~$i$, the target event~$\targetEvent:=\aTrip[i+1]$ is examined.
Line~\ref{alg:candidateTrip2:checkTargetEvent} tests whether~$\targetEvent$ is a c-out event.
If so, then~$\directTargetSplitLimit(\targetEvent)$ and~$\indirectTargetSplitLimitBound(\targetEvent)$ are calculated in lines~\ref{alg:candidateTrip2:1targetSplitLimit}--\ref{alg:candidateTrip2:2targetSplitLimit}.
Lines~\ref{alg:candidateTrip2:aggregateTargetSplitLimit} and~\ref{alg:candidateTrip2:aggregateMaxWitnessTime} updates~$\aggTargetSplitLimitBound(\aTrip)$ and~$\maxWitnessTime(\aTrip)$, respectively, to uphold the invariants.
Line~\ref{alg:candidateTrip2:setMinOriginDelay} incorporates~$\aggTargetSplitLimitBound(\aTrip)$ into~$\minOriginDelayBound(\destinationEvent)$.
Line~\ref{alg:candidateTrip2:calculateAggTimeTravelLimit} calculates~$\aggTimeTravelLimit(\destinationEvent)$ from~$\maxWitnessTime(\aTrip)$ and line~\ref{alg:candidateTrip2:setMaxOriginDelay} incorporates this into~$\maxOriginDelay(\destinationEvent)$.
If~$\minOriginDelayBound(\destinationEvent)\leq\maxOriginDelay(\destinationEvent)$, the shortcut~$(\originEvent,\destinationEvent)$ is generated in line~\ref{alg:candidateTrip2:generateShortcut}.

\subparagraph*{Merging Shortcuts.}
Line~\ref{alg:overview:merge} of Algorithm~\ref{alg:overview} merges the newly generated shortcuts~$\newShortcutEdges$ into the result set~$\shortcutEdges$.
If a shortcut~$\edge=(\originEvent,\destinationEvent)\in\newShortcutEdges$ is already contained in~$\shortcutEdges$, then~$\minOriginDelayBound(\edge)$ is set to~$\min(\minOriginDelayBound(\edge),\minOriginDelayBound(\destinationEvent))$ and~$\maxOriginDelay(\edge)$ to~$\max(\maxOriginDelay(\edge),\maxOriginDelay(\destinationEvent))$.
Otherwise, they are initialized with~$\minOriginDelayBound(\destinationEvent)$ and~$\maxOriginDelay(\destinationEvent)$ and the shortcut is added to~$\shortcutEdges$.

\section{Update Phase}
\label{sec:update}
The update phase incorporates the delay update stream~$\langle\delayUpdate_1,\delayUpdate_2,\dots\rangle$ into the query data structures.
We present two variants of the update phase: a basic one in Section~\ref{sec:update:basic} that only updates the data structures and removes irrelevant shortcuts, and an advanced version in Section~\ref{sec:update:advanced} that searches for missing shortcuts.
For simplicity, we assume that each update is processed individually.
In practice, this may not be viable because updates arrive too frequently.
Thus, updates that are received while an update phase is running are buffered.
Once it has finished, the buffered updates are combined into a single input for the next phase.

\subsection{Basic}
\label{sec:update:basic}
For a delay update~$\delayUpdate$, the basic update phase first computes the new delay scenario~$\delayScenario$ by incorporating~$\delayUpdate$ into the previous scenario.
It then recalculates the set of routes such that no trips of the same route overtake each other in~$\delayScenario$ and rebuilds the TB query data structures accordingly.
Finally, two types of irrelevant shortcuts are removed from the set~$\shortcutEdges$ of Delay-ULTRA shortcuts:
The first is the set
\[\infeasibleShortcuts(\delayScenario):=\{\edge=(\originEvent,\destinationEvent)\in\shortcutEdges\mid\departureTime(\delayScenario,\destinationEvent)<\arrivalTime(\delayScenario,\originEvent)+\transferTime(\edge)\}\]
of shortcuts that are infeasible in~$\delayScenario$.
The second type are shortcuts~$\edge=(\originEvent,\destinationEvent)$ for which the arrival delay~$\arrivalDelay(\originEvent)$ is not within the computed origin delay interval~$\originDelayIntervalBound(\edge)$.

\subsection{Advanced}
\label{sec:update:advanced}
The advanced variant searches for shortcuts that are required due to delays above~$\maxDelay$ but are not in~$\shortcutEdges$. 
An exhaustive search for all missing shortcuts would take too long for the update phase, so we propose a heuristic \emph{replacement search} for candidates that were made infeasible by the last delay update.
A candidate with intermediate transfer~$(\originEvent,\destinationEvent)$ becomes infeasible if the arrival delay of~$\originEvent$ increases to the point where~$\destinationEvent$ is no longer reachable.
The update phase therefore collects the set~$\delayedStopEvents$ of stop events whose arrival delay was increased by the current update.
For each stop event in this set, the algorithm collects the affected candidates and searches for replacements that are prefix-optimal in the current delay scenario~$\delayScenario$.

\subparagraph{Replacement Search Routine.}
The core of the replacement search algorithm is the routine~$\FindReplacements$.
It takes as input a set~$\sourceEvents$ of source events, a set~$\targetStops$ of target stops and an upper bound~$\maxArrivalTime(\astop)$ for the arrival time at each target stop~$\astop\in\targetStops$.
For each pair of source event~$\sourceEvent\in\sourceEvents$ and target stop~$\targetVertex\in\targetStops$, it searches for a Pareto-optimal~$\astop(\sourceEvent)$-$\targetVertex$-journey with exactly~2 trips that departs no earlier than~$\departureTime(\sourceEvent)$ and arrives no later than~$\maxArrivalTime(\targetVertex)$.
If one exists, a shortcut representing its intermediate transfer is generated.

The routine uses one-to-all MR searches restricted to two rounds as its basic building blocks.
First, a backward MR search establishes a latest departure time~$\backwardDepartureTime(\aVertex,n)$ at each vertex~$\aVertex$ for each round~$n$.
This is the latest departure time at~$\aVertex$ such that at least one target stop~$\targetVertex\in\targetStops$ can be reached from~$\aVertex$ with at most~$i$ trips no later than~$\maxArrivalTime(\targetVertex)$.
For each target stop~$\targetVertex\in\targetStops$, $\backwardDepartureTime(\targetVertex,0)$ is initialized with~$\maxArrivalTime(\targetVertex)$.
The backward MR search is then run with the following pruning rule:
Let~$\minTime$ be the earliest departure time in~$\delayScenario$ of any source event in~$\sourceEvents$.
Whenever the search improves~$\backwardDepartureTime(\aVertex,n)$ for a vertex~$\aVertex$ in round~$n$, it checks whether~$\backwardDepartureTime(\aVertex,n)<\minTime$.
If so, then the search is pruned at~$\aVertex$ since it cannot be reached in time from any source event in~$\sourceEvents$.

Afterwards, for each source event~$\sourceEvent\in\sourceEvents$, a forward MR search from~$\astop(\sourceEvent)$ with departure time~$\departureTime(\sourceEvent,\delayScenario)$ is performed, which computes the earliest arrival time~$\forwardArrivalTime(\aVertex,n)$ at each vertex~$\aVertex$ for each round~$n$.
Whenever~$\forwardArrivalTime(\aVertex,n)$ is improved for a vertex~$\aVertex$ in round~$n$, the algorithm checks whether~$\forwardArrivalTime(\aVertex,n)>\backwardDepartureTime(\aVertex,2-n)$.
If so, the search is pruned at~$\aVertex$ since no target stop~$\targetVertex\in\targetStops$ can be reached from there without exceeding~$\maxArrivalTime(\targetVertex)$.
Finally, all journeys with exactly two trips that have been found at a target stop~$\targetVertex\in\targetStops$ are extracted.
For each intermediate transfer that occurs in one of these journeys, a shortcut is added to the set~$\replacementShortcutEdges$ of replacement shortcuts.

\subparagraph{Calling the Routine.}
\begin{algorithm}[t]
	\caption{Replacement search.}\label{alg:replaceOrigin}
	\Input{Shortcuts~$\shortcutEdges$, delay scenario~$\delayScenario$, infeasible shortcuts~$\infeasibleShortcuts(\delayScenario)$,\\delayed origin event~$\originEvent\in\delayedStopEvents$}
	\Output{Replacement shortcuts~$\replacementShortcutEdges$}
	$\sourceEvents\leftarrow\precedingEvents(\originEvent)$\label{alg:replaceOrigin:sourceEvents}\;
	$\destinationEvents\leftarrow\{\destinationEvent\mid(\originEvent,\destinationEvent)\in\infeasibleShortcuts(\delayScenario)\}$\label{alg:replaceOrigin:destinationEvents}\;
	$\targetStops\leftarrow\{\astop(\targetEvent)\mid\targetEvent\in\succeedingEvents(\destinationEvent),\destinationEvent\in\destinationEvents\}$\label{alg:replaceOrigin:targetStops}\;
	\ForEach{$\targetVertex\in\targetStops$\label{alg:replaceOrigin:startUpperBounds}}{
		$\originArrivalTime(\targetVertex)\leftarrow\arrivalTime(\delayScenario,\originEvent)+\transferTime(\originVertex,\targetVertex)$\;
		$\maxArrivalTime(\targetVertex)\leftarrow\originArrivalTime(\targetVertex)$\label{alg:replaceOrigin:walkToTarget}\;
	}
	\ForEach{$\aTrip[i]\in\destinationEvents$}{
		$\aTrip'\leftarrow\FindEarliestTrip(\aTrip,i,\delayScenario)$\label{alg:replaceOrigin:findReplacementTrip}\;
		\ForEach{$\targetEvent\in\succeedingEvents(\aTrip'[i])$}{
			$\maxArrivalTime(\astop(\targetEvent))\leftarrow\min(\maxArrivalTime(\astop(\targetEvent)),\arrivalTime(\targetEvent))$\label{alg:replaceOrigin:useReplacementTrip}\;
		}
	}
	$\replacementShortcutEdges\leftarrow\FindReplacements(\sourceEvents,\targetStops,\maxArrivalTime(\cdot))$\label{alg:replaceOrigin:runSearch}\;
\end{algorithm}

Algorithm~\ref{alg:replaceOrigin} depicts the replacement search for an origin event~$\originEvent\in\delayedStopEvents$.
The set~$\sourceEvents$ of potential source events consists of all stop events preceding~$\originEvent$ in~$\aTrip(\originEvent)$.
For each shortcut~$(\originEvent,\destinationEvent)$ that is infeasible in~$\delayScenario$, the destination event~$\destinationEvent$ is added to~$\destinationEvents$ and the stops visited by its successor events in~$\aTrip(\destinationEvent)$ are added to the set~$\targetStops$ of target stops.
Lines~\ref{alg:replaceOrigin:startUpperBounds}--\ref{alg:replaceOrigin:useReplacementTrip} compute the upper bounds~$\maxArrivalTime(\cdot)$.
Two alternatives are considered for an infeasible candidate~$\aCandidateJourney=\candidateSequence$ with~$\astop(\targetEvent)=\targetVertex$.
One is the journey~$\langle[\sourceEvent,\originEvent],\targetVertex\rangle$ that transfers directly from~$\originVertex$ to~$\targetVertex$.
Line~\ref{alg:replaceOrigin:walkToTarget} initializes~$\maxArrivalTime(\targetVertex)$ with the arrival time~$\originArrivalTime(\targetVertex):=\arrivalTime(\delayScenario,\originEvent)+\transferTime(\originVertex,\targetVertex)$ of this journey.
The other option is to wait at~$\destinationVertex$ for the next trip with the same stop sequence as~$\aTrip(\destinationEvent)$.
For each destination event~$\aTrip[i]\in\destinationEvents$, line~\ref{alg:replaceOrigin:findReplacementTrip} calls the subroutine~$\FindEarliestTrip(\aTrip,i,\delayScenario)$ to find the earliest trip~$\aTrip'$ with the same stop sequence as~$\aTrip$ such that~$\originArrivalTime(\astop(\aTrip[i]))\leq\departureTime(\delayScenario,\aTrip'[i])$.
Then, for each stop event~$\targetEvent$ succeeding~$\aTrip'[i]$ in~$\aTrip'$, line~\ref{alg:replaceOrigin:useReplacementTrip} incorporates~$\arrivalTime(\targetEvent)$ into~$\maxArrivalTime(\astop(\targetEvent))$.
Finally, the~$\FindReplacements$ routine is invoked in line~\ref{alg:replaceOrigin:runSearch}.

Algorithm~\ref{alg:replaceOrigin} can be sped up further by batching the~$\FindReplacements$ calls.
For all delayed stop events belonging to the same trip, the computed inputs are merged and~$\FindReplacements$ is called just once on the merged inputs.
This reduces the number of MR searches, at the cost of worse upper bounds~$\maxArrivalTime(\cdot)$.
Furthermore, the individual~$\FindReplacements$ calls are independent of each other and can be run in parallel.

\section{Experiments}
\label{sec:experiments}
\begin{table*}[t]
	\center
	\caption{Sizes of the public transit networks and the unrestricted transfer graphs.}
	\label{tbl:networks}
	\begin{tabular*}{\textwidth}{@{\,}l@{\extracolsep{\fill}}r@{\extracolsep{\fill}}r@{\extracolsep{\fill}}r@{\extracolsep{\fill}}r@{\extracolsep{\fill}}r@{\extracolsep{\fill}}r@{\,}}
		\toprule
		Network & Stops & Routes & Trips & Events & Graph vertices & Graph edges\\
		\midrule
		London & 19\,682 & 1\,955 & 114\,508 & 4\,508\,644 & 181\,642 & 575\,364 \\
		Germany & 243\,167 & 230\,255 & 2\,381\,394 & 48\,380\,936 & 6\,870\,496 & 21\,367\,044 \\
		Stuttgart & 13\,584 & 12\,351 & 91\,304 & 1\,561\,972 & 1\,166\,604 & 3\,682\,232 \\
		Switzerland & 25\,125 & 13\,786 & 350\,006 & 4\,686\,865 & 603\,691 & 1\,853\,260 \\
		\bottomrule
	\end{tabular*}
\end{table*}
All algorithms were implemented in C\raisebox{0.15ex}{\small++}17 and compiled with GCC version 11.3.0 and optimization flag \mbox{-O3}.
Shortcut computations were run on a machine with two 64-core AMD Epyc Rome 7742 CPUs clocked at~2.25\,GHz, with a turbo frequency of~3.4\,GHz,~1024\,GiB of~\mbox{DDR4-3200}~RAM, and~256\,MiB of L3 cache, using all 128 cores.
All other experiments were conducted on a single core of a machine with two 8-core Intel Xeon Skylake SP Gold 6144 CPUs clocked at~3.5\,GHz, with a turbo frequency of~4.2\,GHz,~192\,GiB of~\mbox{DDR4-2666}~RAM, and~24.75\,MiB of L3 cache.

We evaluate our algorithms on the networks of Switzerland, Greater London, the greater region of Stuttgart, and Germany, all of which were previously used to evaluate ULTRA~\cite{Bau23}.
See Table~\ref{tbl:networks} for an overview.
The public transit networks for London and Switzerland were obtained from Transport for London\footnote{\url{https://data.london.gov.uk}} and a publicly available GTFS feed\footnote{\url{https://gtfs.geops.ch/}}, respectively; the other two are based on proprietary data.
All networks cover the timetables of two successive business days except for London, which covers a single business day.
Except for London, all networks include long-distance trains, but the vast majority of trips represent local transport (e.g., buses, trams, subways).
Transfer graphs were extracted from OpenStreetMap\footnote{\url{https://download.geofabrik.de/}} following the methodology in~\cite{Bau23}.
We used walking as the transfer mode, assuming a constant speed of 4.5\,km/h on each edge.

To generate delay scenarios for our experimental evaluation, we use a synthetic delay model.
In Section~\ref{sec:delay-impact}, we describe the model in detail, discuss how we generate test queries and analyze the impact of delays on the four benchmark networks.
Afterwards, we evaluate the performance of Delay-ULTRA, the update phases and the query algorithms in Section~\ref{sec:exp:alg}.

\subsection{Delay Model}
\label{sec:delay-impact}
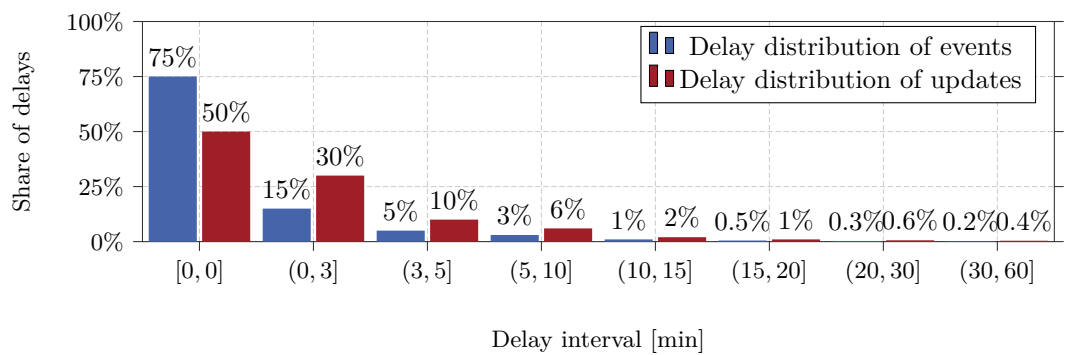
\begin{figure*}[t]
	\newcommand{\plotHeight}{4.5cm}
\newcommand{\plotWidth}{0.97\textwidth}

\begin{tikzpicture}
\arrayrulecolor{legendColor}
\pgfplotsset{
    grid style={KITblack20,line width = 0.2pt,dash pattern = on 2pt off 1pt}
}

\begin{axis}[
   nodes near coords,
   ybar,
   height=\plotHeight,
   width=\plotWidth,
   xmin=-0.5,
   xmax=7.5,
   ymin=0,
   ymax=100,
   xtick={0, 1, 2, 3, 4, 5, 6, 7},
   ytick={0, 25, 50, 75, 100},
   xtick pos=left,
   xticklabels={{$[0,0]$},{$(0,3]$},{$(3,5]$},{$(5,10]$},{$(10,15]$},{$(15,20]$},{$(20,30]$},{$(30,60]$}},
   yticklabels={0\%,25\%,50\%,75\%,100\%},
   xlabel={Delay interval [min]},
   xlabel style={yshift=-10pt},
   ylabel={Share of delays},
   ylabel style={yshift=5pt},
   xtick align=outside,
   ytick align=outside,
   grid=both,
   axis line style={legendColor},
   at={(0,0)},
   point meta=explicit symbolic,
]
\addplot[draw=none,fill=KITblue,bar width=18pt] coordinates {
    (0,75) [75\%]
    (1,15) [15\%]
    (2,5) [5\%]
    (3,3) [3\%]
    (4,1) [1\%]
    (5,0.5) [0.5\%]
    (6,0.3) [0.3\%]
    (7,0.2) [0.2\%]
};
\addplot[draw=none,fill=KITred,bar width=18pt] coordinates {
    (0,50) [50\%]
    (1,30) [30\%]
    (2,10) [10\%]
    (3,6) [6\%]
    (4,2) [2\%]
    (5,1) [1\%]
    (6,0.6) [0.6\%]
    (7,0.4) [0.4\%]
};

\legend{Delay distribution of events,Delay distribution of updates}
\end{axis}
\arrayrulecolor{black}
\end{tikzpicture}%
	\caption{%
        Delay distribution in our synthetically generated delay scenarios.
		There are eight delay intervals, ranging from 0 to 60 minutes.
		The percentages of stop events whose delay falls within each interval is shown in blue, and the percentage of delay updates is shown in red.
		Within an interval, delays are distributed uniformly.
	}%
	\label{fig:delayDistribution}%
\end{figure*}
Since we do not have access to proper real-world delay data, we generate synthetic delay scenarios using the model by Bast et al.~\cite{Bas13b}:
We create one delay update per trip~$\aTrip$, with the index~$i$ of the first affected event chosen uniformly at random and the delay~$\delay$ chosen at random according to the probability distribution outlined below.
The arrival and departure delays of each stop event~$\aTrip[j]$ with~$j\geq{}i$ are set to~$\delay$; updates with~$\delay=0$ are discarded.
For the reveal time, we choose the arrival time~$\arrivalTime(\aTrip[i])$ of the first affected event.
The distribution for the expected delay of a stop event is shown in blue in Figure~\ref{fig:delayDistribution}; it is roughly based on aggregated real-world punctuality data for London\footnote{\url{https://www.raildeliverygroup.com/punctuality.html}. Accessed May 6th, 2022.}.
Note that since our model divides each trip into an undelayed first portion and a second portion with delay~$\delay$, half of all stop events remain unaffected by delays in expectation.
To obtain the desired expected delay distribution, we reduce the probability of delay~0 to~50\% and double all other probabilities; the resulting probability distribution is shown in red in Figure~\ref{fig:delayDistribution}.
To allow for a reasonably fast simulation of the update phase, we limit our delay scenarios to the interval between~12 and~1\,PM of the first covered day.
Updates with a reveal time after~1\,PM are discarded.
For updates with a reveal time before~12\,PM, we omit all affected stop events before~12\,PM from the update.\looseness=-1

Note that our synthetic model is heavily simplified, for example by assuming that there is at most one cause of delay along each trip and that this delay is never caught up.
However, most of our simplifications make the synthetic scenarios more challenging than real ones.
Since delays are chosen uniformly at random within each interval, instead of preferring lower delays, the average delay is overestimated.
Delay updates are revealed at the latest possible moment, which puts more pressure on the update phase to process them in time.
Finally, our model does not reproduce knock-on effects caused by trips waiting for each other or a delayed train blocking a platform.
In practice, these effects make it more likely that precomputed transfers stay feasible in the presence of delays.

\subparagraph*{Impact of Delays.}
\begin{table*}[t]
	\center
	\caption{Impact of delays on queries in the synthetic delay scenarios. For each network, random queries were generated until 1\,000 affected ones were found. A query is affected if MR without delay information returns at least one suboptimal or infeasible journey. Also reported are the number of returned journeys that are suboptimal and infeasible in the synthetic delay scenarios, respectively. For each infeasible journey, the highest slack among its infeasible transfers was measured. The reported slack is the median value across all infeasible journeys.}%
	\label{tbl:testQueries}%
	\begin{tabular*}{\textwidth}{@{\,}l@{\extracolsep{\fill}}r@{\extracolsep{\fill}}r@{\extracolsep{\fill}}r@{\extracolsep{\fill}}r@{\extracolsep{\fill}}r@{\extracolsep{\fill}}r@{\,}}
		\toprule
		\multirow{2}{*}{Network} & \multicolumn{2}{c}{Queries} & \multicolumn{3}{c}{Journeys} & \multirow{2}{*}{Slack [s]}\\
		\cmidrule{2-3} \cmidrule{4-6}
		& Total & Affected & Total & Affected & Infeasible & \\
		\midrule
		London & 5\,572 & 1\,000 (17.94\%) & 25\,532  & 1\,585 (6.20\%) & 956 (3.74\%) & 54\\
		Germany & 10\,182 & 1\,000 (\phantom{1}9.82\%) & 54\,852 & 1\,699 (3.09\%) & 783 (1.42\%) & 68\\
		Stuttgart & 102\,362 & 1\,000 (\phantom{1}0.97\%) & 424\,480 &1\,528 (0.35\%) & 724 (0.17\%) & 136\\
		Switzerland & 30\,375 & 1\,000 (\phantom{1}3.29\%) & 143\,570 & 1\,601 (1.11\%) & 894 (0.62\%) & 120\\
		\bottomrule
	\end{tabular*}
\end{table*}
To investigate how many queries are affected by delays in each networks, we generate queries until~1\,000 affected ones are found.
We choose source and target vertices uniformly at random and departure times uniformly at random between~12 and~1\,PM.
For the query execution times, we use the departure times.
This maximizes the number of delays that are already known when a query is executed, making it more likely to be affected by delays are therefore more challenging to answer correctly.
In reality, the execution time often lies well before the departure time.
For each query, we perform one MR search in an undelayed scenario~$\delayScenario_\text{base}$ and one in the scenario~$\delayScenario_\text{del}$ that incorporates all delay updates that are known at execution time.
This yields Pareto sets~$\journeys_\text{base}$ and~$\journeys_\text{del}$, respectively.
If any journey~$\aJourney\in\journeys_\text{base}$ is infeasible in~$\delayScenario_\text{del}$ or has a later arrival time than a journey in~$\journeys_\text{del}$ with at most~$|\aJourney|$ trips, the query counts as affected.
The results are shown in Table~\ref{tbl:testQueries}.
We observe that among the affected queries, the share of affected and infeasible journeys is similar across all networks: around a third are affected and about half of those are infeasible.
However, the share of affected queries varies heavily between networks.
In particular, it is much higher for London and Germany than for Switzerland and Stuttgart.

To explain these differences, we measure the \emph{average headway} of a route, i.e., the mean difference in departure time between two consecutive trips of the route.
Figure~\ref{fig:headways} shows the distribution of the average headways for each network.
We notice a stark difference between London, which is a large metropolitan area, and the other networks, which also contain many rural regions.
For London, the distribution centers around headways between 10 and 20 minutes, whereas the other networks mostly contain trips with headways above 30 minutes.

\begin{figure*}
	\newcommand{\plotHeight}{5cm}
\newcommand{\plotWidth}{\textwidth}

\begin{tikzpicture}
\arrayrulecolor{legendColor}
\pgfplotsset{
    grid style={KITblack20,line width = 0.2pt,dash pattern = on 2pt off 1pt}
}

\begin{axis}[
   ybar,
   height=\plotHeight,
   width=\plotWidth,
   xmin=-0.5,
   xmax=7.5,
   ymin=0,
   ymax=50,
   xtick={0, 1, 2, 3, 4, 5, 6, 7},
   ytick={0, 10, 20, 30, 40, 50},
   xtick pos=left,
   xticklabels={{$[0,5]$},{$(5,10]$},{$(10,20]$},{$(20,30]$},{$(30,60]$},{$(60,120]$},{$(120,240]$},{$(240,\infty)$}},
   yticklabels={0\%,10\%,20\%,30\%,40\%,50\%},
   xlabel={Average headway interval [min]},
   xlabel style={yshift=-10pt},
   ylabel={Percentage of routes},
   ylabel style={yshift=5pt},
   xtick align=outside,
   ytick align=outside,
   grid=both,
   axis line style={legendColor},
   at={(0,0)},
]

\addplot[draw=none,fill=KITred,bar width=5pt] table[x=Headway,y=London,col sep=tab]{fig/headways.dat};\label{legend:london}
\addplot[draw=none,fill=KITorange,bar width=5pt] table[x=Headway,y=Germany,col sep=tab]{fig/headways.dat};\label{legend:germany}
\addplot[draw=none,fill=KITgreen,bar width=5pt] table[x=Headway,y=Switzerland,col sep=tab]{fig/headways.dat};\label{legend:switzerland}
\addplot[draw=none,fill=KITblue,bar width=5pt] table[x=Headway,y=Stuttgart,col sep=tab]{fig/headways.dat};\label{legend:stuttgart}

\legend{London,Germany,Switzerland,Stuttgart}
\end{axis}
\arrayrulecolor{black}
\end{tikzpicture}%
	\caption{%
        Headway distribution per network. Only trips that depart during the first service day are counted. Routes with fewer than two trips are ignored.
    }%
	\label{fig:headways}%
\end{figure*}
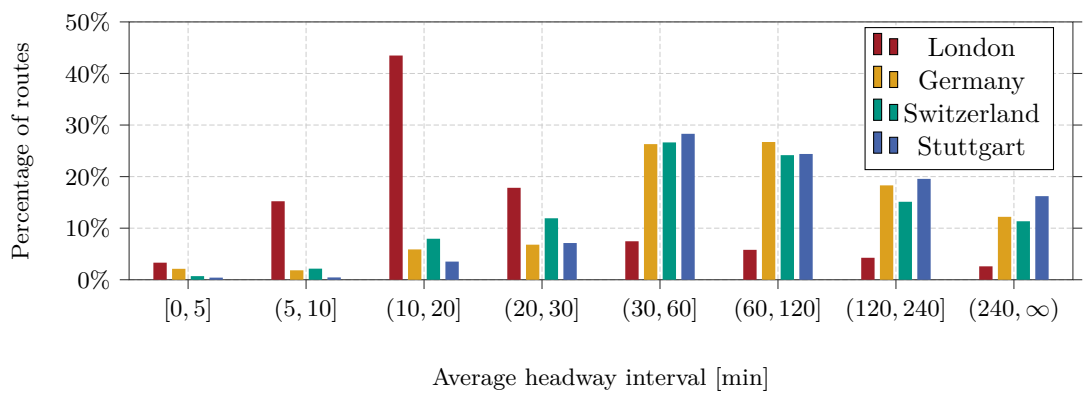

If headways are low, intermediate transfers have low slacks and are therefore more likely to become infeasible in the presence of delays.
To confirm this, we measure the highest slack among the infeasible transfers of each infeasible journey.
As shown in Table~\ref{tbl:testQueries}, the median slack across all infeasible journeys is particularly low for London and Germany.
While the headway distribution of Germany is closer to that of Switzerland and Stuttgart, it has a comparatively high share of routes with an average headway below 5 minutes.
This indicates its greater structural diversity:
While its overall structure is similar to that of Switzerland and Stuttgart, it contains some metropolitan areas that are closer in structure and size to London.
Accordingly, journeys that span these areas tend to use routes with lower headways and are therefore more likely to be affected by delays.

\subsection{Algorithm Performance}
\label{sec:exp:alg}
We now evaluate the performance of Delay-ULTRA.
First we describe our experimental setup.
For Switzerland, we then analyze how the delay limit affects the algorithm's performance.
For the other networks, we focus on a smaller selection of delay limits but analyze other aspects in more detail.

\subparagraph*{Experimental Setup.}
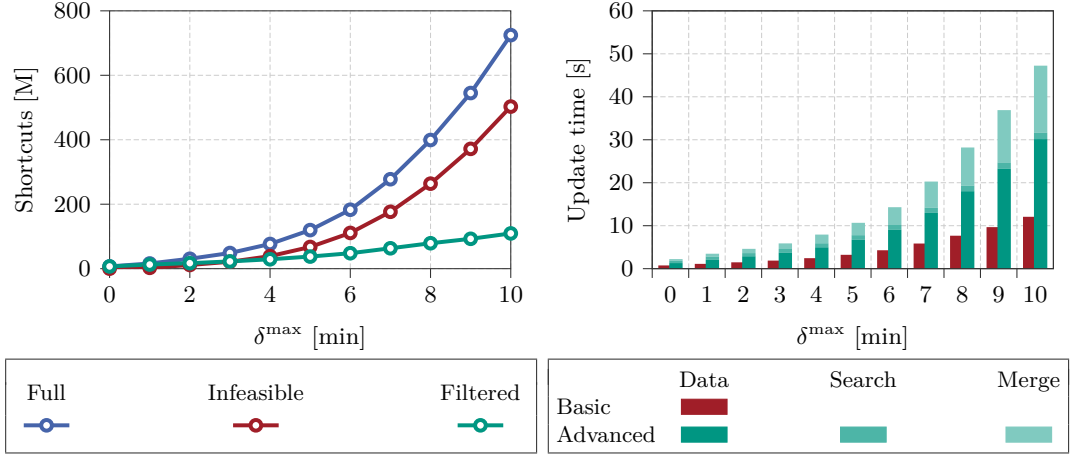
\begin{figure*}[t]
	\newcommand{\plotHeight}{5cm}
\newcommand{\plotWidth}{0.49\textwidth}

\begin{tikzpicture}
	\arrayrulecolor{legendColor}
	\pgfplotsset{
		grid style={KITblack20,line width = 0.2pt,dash pattern = on 2pt off 1pt}
	}
	
	\begin{axis}[
		name=plot1,
		height=\plotHeight,
		width=\plotWidth,
		xmin=0,
		xmax=10,
		ymin=0,
		ymax=800,
		xlabel={$\maxDelay$ [min]},
		xlabel style = {yshift=-2pt},
		ylabel={Shortcuts [M]},
		ylabel style = {yshift=-1pt},
		ytick={0, 200, 400, 600, 800},
		xtick pos=left,
		ytick pos=left,
		ytick align=outside,
		xtick align=outside,
		grid=both,
		axis line style={legendColor},
		at={(0,0)}
		]
		
		\addplot[color=KITblue,line width=1.5pt,mark=*,mark options={fill=white}] table[x=Delay,y=FullShortcuts,col sep=tab]{fig/preprocessing.dat};\label{legend:full}
		\addplot[color=KITred,line width=1.5pt,mark=*,mark options={fill=white}] table[x=Delay,y=InfeasibleShortcuts,col sep=tab]{fig/preprocessing.dat};\label{legend:infeasible}
		\addplot[color=KITgreen,line width=1.5pt,mark=*,mark options={fill=white}] table[x=Delay,y=UndelayedShortcuts,col sep=tab]{fig/preprocessing.dat};\label{legend:filtered}
	\end{axis}

	\node[inner sep=0pt,outer sep=0pt] (legend) at (plot1.outer south west) {};
	
	\node[inner sep=0pt,outer sep=0pt,anchor=north west] at (legend) {
		\footnotesize
		\begin{tabular*}{0.5\textwidth}{|@{~~}c@{\extracolsep{\fill}}c@{\extracolsep{\fill}}c@{~~}|}
			\hline
			&                                    &                                     \\[-1pt]
			Full & Infeasible & Filtered \\[3pt]
			\legend{\ref{legend:full}}    & \legend{\ref{legend:infeasible}}   & \legend{\ref{legend:filtered}} \\[4.5pt]
			\hline
		\end{tabular*}
	};
	
	\begin{axis}[
		height=\plotHeight,
		width=\plotWidth,
		xmin=0,
		xmax=11,
		ymin=0,
		ymax=60,
		xtick={0, 1, 2, 3, 4, 5, 6, 7, 8, 9, 10},
		ytick style={draw=none},
		ytick={10, 20, 30, 40, 50,60},
		xtick pos=left,
		xticklabels={,,},
		yticklabels={,,},
		xtick align=outside,
		grid=both,
		axis line style={legendColor},
		at={(0.51\linewidth,0)}]
		
		\addlegendimage{color=KITred,fill=KITred,area legend}\label{legend:basic}
		\addlegendimage{color=KITgreen,fill=KITgreen,area legend}\label{legend:advanced:data}
		\addlegendimage{color=KITgreen70,fill=KITgreen70,area legend}\label{legend:advanced:search}
		\addlegendimage{color=KITgreen50,fill=KITgreen50,area legend}\label{legend:advanced:merge}
	\end{axis}
	
	\begin{axis}[
		hide axis,
		ybar stacked,
		height=\plotHeight,
		width=\plotWidth,
		xmin=-0.5,
		xmax=10.5,
		ymin=0,
		ymax=60,
		axis line style={legendColor},
		at={(0.51\linewidth,0)}
		]
		
		\addplot[draw=none,fill=KITred,bar width=5pt] table[x expr=\coordindex-0.15,y=BasicUpdateTime,col sep=tab]{fig/update.dat};
	\end{axis}
	
	\begin{axis}[
		hide axis,
		ybar stacked,
		height=\plotHeight,
		width=\plotWidth,
		xmin=-0.5,
		xmax=10.5,
		ymin=0,
		ymax=60,
		axis line style={legendColor},
		at={(0.51\linewidth,0)}
		]
		
		\addplot[draw=none,fill=KITgreen,bar width=5pt] table[x expr=\coordindex+0.15,y=AdvancedDataUpdateTime,col sep=tab]{fig/update.dat};
		\addplot[draw=none,fill=KITgreen70,bar width=5pt] table[x expr=\coordindex+0.15,y=AdvancedSearchTime,col sep=tab]{fig/update.dat};
		\addplot[draw=none,fill=KITgreen50,bar width=5pt] table[x expr=\coordindex+0.15,y=AdvancedMergeTime,col sep=tab]{fig/update.dat};
	\end{axis}
	
	\begin{axis}[
		name=plot2,
		ybar stacked,
		height=\plotHeight,
		width=\plotWidth,
		xmin=-0.5,
		xmax=10.5,
		ymin=0,
		ymax=60,
		xlabel={$\maxDelay$ [min]},
		ylabel={Update time [s]},
		xlabel style = {yshift=-2pt},
		ytick={0, 10, 20, 30, 40, 50, 60},
		xtick={0, 1, 2, 3, 4, 5, 6, 7, 8, 9, 10},
		xtick pos=left,
		xtick style={draw=none},
		ytick pos=left,
		ytick align=outside,
		xtick align=outside,
		grid=none,
		at={(0.51\linewidth,0)}]
	\end{axis}
	
	\node[inner sep=0pt,outer sep=0pt,anchor=north east,xshift=\textwidth] at (legend) {
		\footnotesize
		\begin{tabular*}{0.49\textwidth}{|@{~}l@{~~~}c@{\extracolsep{\fill}}c@{\extracolsep{\fill}}c@{~}|}
			\hline
			&                                &                                   &                                \\[-6pt]
			&             Data              &             Search              &             Merge              \\[1pt]
			Basic   & \legend{\ref{legend:basic}} & & \\[1pt]
			Advanced & \legend{\ref{legend:advanced:data}} & \legend{\ref{legend:advanced:search}} & \legend{\ref{legend:advanced:merge}} \\[1pt]
			\hline
		\end{tabular*}
	};
	
	\arrayrulecolor{black}
\end{tikzpicture}%
	\caption{%
		Impact of the delay limit~$\maxDelay$ on the number of precomputed shortcuts (left) and the running times of the update phase (right), measured on the Switzerland network. We report the full set of shortcuts and, for an undelayed scenario, the number of infeasible shortcuts and those remaining after filtering. For the advanced update phase, the running time is split into updating the data structures, performing the replacement search and merging the replacement shortcuts.
	}%
	\label{fig:preprocessing}%
\end{figure*}
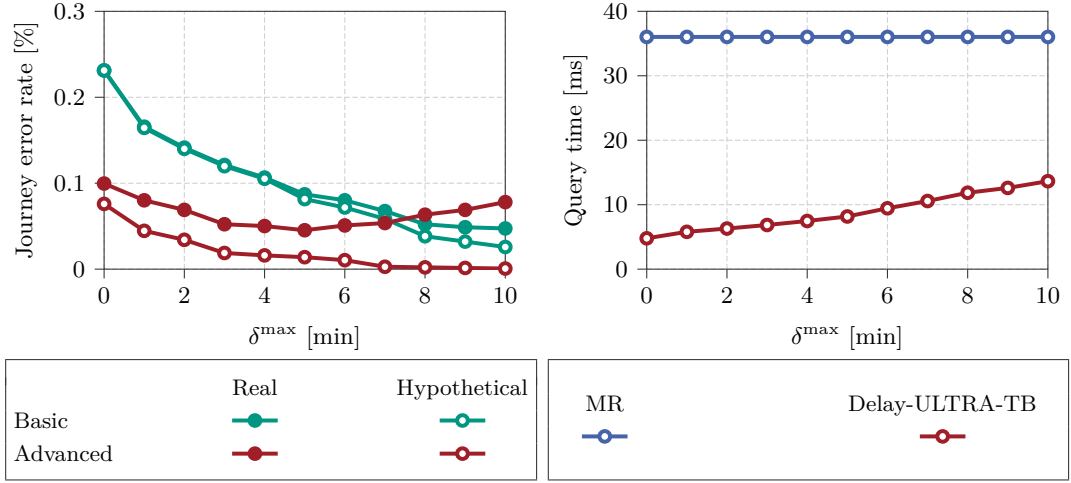
\begin{figure*}[th!]
	\newcommand{\plotHeight}{5cm}
\newcommand{\plotWidth}{0.49\textwidth}

\begin{tikzpicture}
\arrayrulecolor{legendColor}
\pgfplotsset{
   grid style={KITblack20,line width = 0.2pt,dash pattern = on 2pt off 1pt}
}

\begin{axis}[
   name=plot1,
   height=\plotHeight,
   width=\plotWidth,
   xmin=0,
   xmax=10,
   ymin=0,
   ymax=0.3,
   xlabel={$\maxDelay$ [min]},
   xlabel style = {yshift=-2pt},
   ylabel={Journey error rate [\%]},
   ylabel style = {yshift=-1pt},
   xtick pos=left,
   ytick pos=left,
   ytick align=outside,
   xtick align=outside,
   grid=both,
   axis line style={legendColor},
   at={(0,0)}
]

\addplot[color=KITgreen,line width=1.5pt,mark=*] table[x=Delay,y=NoReplacementCoveredJourneys,col sep=tab]{fig/coverage.dat};\label{legend:noreplacement:real}
\addplot[color=KITgreen,line width=1.5pt,mark=*,mark options={fill=white}] table[x=Delay,y=NoReplacementHypotheticalCoveredJourneys,col sep=tab]{fig/coverage.dat};\label{legend:noreplacement:hypothetical}
\addplot[color=KITred,line width=1.5pt,mark=*] table[x=Delay,y=ReplacementCoveredJourneys,col sep=tab]{fig/coverage.dat};\label{legend:replacement:real}
\addplot[color=KITred,line width=1.5pt,mark=*,mark options={fill=white}] table[x=Delay,y=ReplacementHypotheticalCoveredJourneys,col sep=tab]{fig/coverage.dat};\label{legend:replacement:hypothetical}
\end{axis}

\node[inner sep=0pt,outer sep=0pt] (legend) at (plot1.outer south west) {};

\node[inner sep=0pt,outer sep=0pt,anchor=north west] at (legend) {
\footnotesize
\begin{tabular*}{0.5\textwidth}{|@{~}l@{\extracolsep{\fill}}c@{\extracolsep{\fill}}c@{~}|}
    \hline
                    &                                    &                                     \\[-3pt]
    & Real & Hypothetical \\[3pt]
    Basic    & \legend{\ref{legend:noreplacement:real}}   & \legend{\ref{legend:noreplacement:hypothetical}} \\[3pt]
    Advanced    & \legend{\ref{legend:replacement:real}}   & \legend{\ref{legend:replacement:hypothetical}} \\[4pt]
    \hline
\end{tabular*}
};

\begin{axis}[
    name=plot2,
    height=\plotHeight,
    width=\plotWidth,
    xmin=0,
    xmax=10,
    ymin=0,
    ymax=40,
    xlabel={$\maxDelay$ [min]},
    xlabel style = {yshift=-2pt},
    ylabel={Query time [ms]},
    ylabel style = {yshift=-1pt},
    xtick pos=left,
    ytick pos=left,
    ytick align=outside,
    xtick align=outside,
    grid=both,
    axis line style={legendColor},
    at={(0.51\linewidth,0)}
]

\addplot[color=KITblue,line width=1.5pt,mark=*,mark options={fill=white}] table[x=Delay,y=BaselineQueryTime,col sep=tab]{fig/querytime.dat};\label{legend:mr}
\addplot[color=KITred,line width=1.5pt,mark=*,mark options={fill=white}] table[x=Delay,y=QueryTime,col sep=tab]{fig/querytime.dat};\label{legend:tb}
\end{axis}

\node[inner sep=0pt,outer sep=0pt,anchor=north east,xshift=\textwidth] at (legend) {
    \footnotesize
    \begin{tabular*}{0.49\textwidth}{|@{~~~~}c@{\extracolsep{\fill}}c@{~~~~}|}
    \hline
    &                                     \\[3pt]
    MR & Delay-ULTRA-TB \\[3pt]
    \legend{\ref{legend:mr}}    & \legend{\ref{legend:tb}} \\[10.5pt]
    \hline
    \end{tabular*}
};

\arrayrulecolor{black}
\end{tikzpicture}%
	\caption{%
		Impact of the delay limit~$\maxDelay$ on the query algorithm, measured on Switzerland. \emph{Left:} Journey error rate for various configurations of the update phase, averaged over the test queries from Section~\ref{sec:delay-impact}. \emph{Right:} Mean running times of Delay-ULTRA-TB and MR for 10\,000 random queries.
	}%
	\label{fig:query}%
\end{figure*}
To evaluate the update phase and the result quality of our Delay-ULTRA-TB query algorithm, we simulate the update algorithm within the interval from 12\,PM to 1\,PM and then run the test queries from Section~\ref{sec:delay-impact}.
For a query with execution time~$\queryExecutionTime$, we run MR on the delay scenario that is current at~$\queryExecutionTime$ and Delay-ULTRA-TB on the data produced by the last update phase that was fully finished before~$\queryExecutionTime$.
We then compute two error rates:
The \emph{query error rate} is the share of queries for which Delay-ULTRA-TB fails to find at least one Pareto-optimal journey.
The \emph{journey error rate} is the share of Pareto-optimal journeys (aggregated across all queries) that are not found.
Since the update phase takes non-negligible time, delays do not become known to the query algorithm instantaneously.
To quantify the effect on the result quality, we also compute \emph{hypothetical error rates}: we repeat the same experiment, but this time run a query with execution time~$\queryExecutionTime$ on the data produced by the last update phase started (rather than finished) by~$\queryExecutionTime$, resulting in \emph{hypothetical error rates}.

While this setup is useful for evaluating the result quality, it is not suitable for measuring query times because it switches between answering queries and processing delay updates with a high frequency.
Each update entails rebuilding the TB data structures, which effectively flushes the machine's cache and thereby distorting the query time measurements.
We therefore use a different setup for measuring query times:
We generate 10\,000 random queries with the departure time chosen uniformly at random between 1 and 2\,PM.
Then we simulate the update phases within the interval between 12 and 1\,PM and run the queries on the output of the last update.
Note that this results in a set of queries whose execution time and departure time does not match: the execution time is always 1\,PM while the departure time is up to one hour later.
To quantify the impact that this difference has on the result quality, we also measure the error rates for this set of queries.

\subparagraph*{Impact of Delay Limit.}
Figure~\ref{fig:preprocessing} (left) shows the impact of the delay limit on the number of shortcuts computed by Delay-ULTRA.
For a scenario without any delays, we also report how many shortcuts are infeasible and how many remain after removing both infeasible and irrelevant shortcuts.
Subsequent experiments (see Table~\ref{tbl:update-full}) show that these numbers are almost identical in scenarios with delays.
The growth in the number of total and infeasible shortcuts is roughly quadratic.
The number of filtered shortcuts, which influences the speed of the query algorithm, grows more slowly.
However, the size of the unfiltered shortcut set still affects the running time of the update phase, which also grows quadratically as shown in Figure~\ref{fig:preprocessing} (right).
The advanced update phase takes 3--4 times longer than the basic one.
This is not primarily due to the replacement search itself, but due to building the additional data structures required by it and merging the replacement shortcuts.

\begin{table*}[t]
	\center
	\caption{Results for the Delay-ULTRA shortcut computation. Reported are the full set of computed shortcuts and the number of infeasible and filtered shortcuts in an undelayed scenario.}%
	\label{tbl:preprocessing-full}%
	\begin{tabular*}{\textwidth}{@{\,}l@{\extracolsep{\fill}}r@{\extracolsep{\fill}}r@{\extracolsep{\fill}}r@{\extracolsep{\fill}}r@{\extracolsep{\fill}}r@{\,}}
		\toprule
		\multirow{2}{*}{Network} & $\maxDelay$ & Time & \multicolumn{3}{c}{\#\,Shortcuts}\\
		\cmidrule{4-6}
		& [s] & [hh:mm:ss] & Full & Infeasible & Filtered\\
		\midrule
		\multirow{4}{*}{London}
		&   0 & 00:03:44 &      8\,576\,120 &             0 (\phantom{0}0.00\%) &   8\,576\,120 (100.00\%)\\
		& 120 & 00:10:26 &     86\,573\,982 &  36\,435\,150 (42.09\%) &  39\,665\,271  (\phantom{1}45.82\%)\\
		& 180 & 00:10:41 &    176\,630\,627 &  93\,727\,958 (53.06\%) &  57\,941\,841  (\phantom{1}32.80\%)\\
		& 300 & 00:14:37 &    609\,565\,189 & 398\,706\,627 (65.41\%) & 114\,153\,168  (\phantom{1}18.73\%)\\[3pt]
		\multirow{3}{*}{Germany}
		&   0 & 02:53:57 &     77\,515\,291 &             0 (\phantom{0}0.00\%) &  77\,515\,291 (100.00\%)\\
		& 180 & 11:24:36 &    588\,027\,879 & 209\,113\,002 (35.56\%) & 327\,040\,211  (\phantom{1}55.62\%)\\
		& 300 & 11:51:22 & 1\,401\,861\,101 & 717\,161\,810 (51.16\%) & 533\,061\,808  (\phantom{1}38.03\%)\\[3pt]
		\multirow{3}{*}{Stuttgart}
		&   0 & 00:00:54 &      1\,973\,321 &            0 (\phantom{0}0.00\%) &  1\,973\,321 (100.00\%)\\
		& 180 & 00:02:35 &     13\,388\,174 &  5\,632\,875 (42.07\%) &  6\,576\,102  (\phantom{1}49.12\%)\\
		& 300 & 00:02:59 &     29\,236\,505 & 15\,106\,602 (51.67\%) & 10\,629\,297  (\phantom{1}36.36\%)\\[3pt]
		\multirow{3}{*}{Switzerland}
		&   0 & 00:01:57 &      6\,938\,012 &            0 (\phantom{0}0.00\%) &  6\,938\,012 (100.00\%)\\
		& 180 & 00:06:46 &     48\,529\,326 & 21\,420\,171 (44.14\%) & 22\,481\,726  (\phantom{1}46.33\%)\\
		& 300 & 00:07:56 &    119\,664\,041 & 67\,335\,101 (56.27\%) & 37\,552\,597  (\phantom{1}31.38\%)\\
		\bottomrule
	\end{tabular*}
\end{table*}
\begin{table*}[t]
	\center
	\caption{Results for the update phases. ``Adv.'' indicates whether the basic ($\circ$) or advanced update phase ($\bullet$) was performed. ``Updates'' is the average number of processed delay updates per phase. For the reported shortcuts, ``Query'' is the number of shortcuts given as input to the query algorithm, averaged across all performed update phases, while ``Added'' is the total number of replacement shortcuts found across all phases. Running times are averaged across the evaluated one-hour interval.\looseness=-1}%
	\label{tbl:update-full}%
	\begin{tabular*}{\textwidth}{@{\,}l@{\extracolsep{\fill}}r@{\extracolsep{\fill}}c@{\extracolsep{\fill}}r@{\extracolsep{\fill}}r@{\extracolsep{\fill}}r@{\extracolsep{\fill}}r@{\extracolsep{\fill}}r@{\extracolsep{\fill}}r@{\extracolsep{\fill}}r@{\,}}
		\toprule
		\multirow{2}{*}{Network} & $\maxDelay$ & \multirow{2}{*}{Adv.} & \multirow{2}{*}{Updates} & \multicolumn{2}{c}{Shortcuts} & \multicolumn{4}{c}{Time [s]}\\
		\cmidrule{5-6} \cmidrule{7-10}
		& [s] & & & Query & Added & Update & Search & Merge & Total \\
		\midrule
		\multirow{6}{*}{London} & 0 & $\circ$ & 6.0 & 8\,508\,656 & 0 & -- & -- & -- & 1.0 \\
		& 0 & $\bullet$ & 6.0 & 8\,586\,763 & 98\,149 & 2.0 & 0.5 & 0.6 & 3.1 \\
		& 120 & $\circ$ & 6.0 & 39\,625\,853 & 0 & -- & -- & -- & 3.6 \\
		& 120 & $\bullet$ & 9.8 & 39\,668\,596 & 54\,547 & 7.6 & 0.8 & 2.9 & 11.2 \\
		& 180 & $\circ$ & 6.0 & 58\,051\,426 & 0 & -- & -- & -- & 5.6 \\
		& 180 & $\bullet$ & 15.7 & 58\,082\,966 & 39\,714 & 12.1 & 0.9 & 5.1 & 18.1 \\[3pt]
		\multirow{4}{*}{Germany} & 0 & $\circ$ & 344.2 & 77\,286\,195 & 0 & -- & -- & -- & 10.8 \\
		& 0 & $\bullet$ & 469.3 & 77\,536\,055 & 314\,714 & 17.7 & 17.8 & 5.8 & 41.2 \\
		& 180 & $\circ$ & 344.2 & 326\,591\,619 & 0 & -- & -- & -- & 26.0 \\
		& 180 & $\bullet$ & 1\,290.6 & 326\,691\,285 & 131\,255 & 47.1 & 51.2 & 18.6 & 116.9 \\[3pt]
		\multirow{4}{*}{Stuttgart} & 0 & $\circ$ & 17.9 & 1\,968\,084 & 0 & -- & -- & -- & 0.2 \\
		& 0 & $\bullet$ & 17.9 & 1\,971\,959 & 6\,096 & 0.4 & 0.3 & 0.1 & 0.9 \\
		& 300 & $\circ$ & 17.9 & 10\,620\,259 & 0 & -- & -- & -- & 0.9 \\
		& 300 & $\bullet$ & 18.8 & 10\,621\,317 & 1\,945 & 2.1 & 0.5 & 0.7 & 3.3 \\[3pt]
		\multirow{4}{*}{Switzerland} & 0 & $\circ$ & 13.6 & 6\,920\,689 & 0 & -- & -- & -- & 0.8 \\
		& 0 & $\bullet$ & 22.1 & 6\,934\,598 & 19\,407 & 1.3 & 0.6 & 0.4 & 2.2 \\
		& 300 & $\circ$ & 25.7 & 37\,551\,717 & 0 & -- & -- & -- & 3.2 \\
		& 300 & $\bullet$ & 42.1 & 37\,555\,783 & 5\,856 & 6.7 & 1.0 & 2.9 & 10.7 \\
		\bottomrule
	\end{tabular*}
\end{table*}
Note that our implementation of TB is geared towards query performance rather than quick updates.
Hence, we believe that the running time of the update phase could be improved significantly.
To improve cache locality during a query, we assign consecutive IDs to all trips and stop events belonging to the same route.
The set~$\shortcutEdges$ of shortcut edges is stored as a graph in adjacency array representation, where the outgoing edges~$(\aVertex,\bVertex)$ of a stop event~$\aVertex$ are sorted according to the ID of~$\bVertex$.
After regrouping the set of routes in order to avoid overtaking, the update phase must reorder the trips and stop events to ensure that their IDs are still consecutive.
Afterwards, $\shortcutEdges$ must be reordered as well.
Our implementations of these reordering steps are not fully optimized.
Furthermore, they could be omitted entirely by using an implementation of TB that does not rely on consecutive IDs.

Figure~\ref{fig:query} (left) reports the journey error rate.
Original ULTRA-TB with basic updates is already fairly delay-robust: slightly above~0.2\% of optimal journeys are missed, compared to~1.1\% without updates.
The replacement search reduces this to 0.1\%.
Delay-ULTRA shortcuts offer drastic improvements for low delay limits, but these are eventually offset by the longer update time.
With advanced updates, the real and hypothetical error rate diverge significantly: the latter reaches near-zero while the former stagnates and then increases again.
This shows that a faster implementation of the update phase would significantly improve the result quality.
The smallest errors are achieved with a limit of~10\,min for basic updates and~5\,min for advanced updates.
As shown in Figure~\ref{fig:query} (right), the latter configuration is preferable overall since it produces fewer shortcuts and therefore yields lower query times.
Overall, query times exhibit moderate growth with increasing delay limit: compared to the original ULTRA-TB, they roughly double for a delay limit of~$5$\,min and triple for~$10$\,min.
Compared to MR, the speedup is reduced from 7.5 to 4.4 and 2.6, respectively.

\subparagraph*{Shortcut Computation.}
Results for the shortcut precomputation on all networks are reported in Table~\ref{tbl:preprocessing-full}.
Compared to the original ULTRA, the preprocessing time is~3--4 times longer.
For Stuttgart and Germany, the growth in the number of shortcuts depending on the delay limit is similar to what we observed for Switzerland.
By contrast, London exhibits a much faster growth but also a significantly higher share of infeasible shortcuts.
Again, this is explained by the fact that transfers in the London network have much lower slacks on average and therefore become infeasible much faster.

\subparagraph*{Update Phase.}
\begin{table*}[t]
	\center
	\caption{Query result quality, averaged over the test queries from Section~\ref{sec:delay-impact}. ``Adv.'' indicates whether the basic ($\circ$) or advanced update phase ($\bullet$) was performed. We report the real and hypothetical error rates for queries and journeys, as well as the median detours of suboptimal journeys compared to their optimal counterparts with the same number of trips. For the real update phase, we also report the number of returned journeys that are infeasible and the number of queries with at least one infeasible journey.}%
	\label{tbl:coverage-full}%
	\begin{tabular*}{\textwidth}{@{\,}l@{\extracolsep{\fill}}r@{\extracolsep{\fill}}c@{\extracolsep{\fill}}r@{\extracolsep{\fill}}r@{\extracolsep{\fill}}r@{\extracolsep{\fill}}r@{\extracolsep{\fill}}r@{\extracolsep{\fill}}r@{\extracolsep{\fill}}r@{\extracolsep{\fill}}r@{\,}}
		\toprule
		\multirow{2}{*}{Network} & $\maxDelay$ & \multirow{2}{*}{Adv.} & \multicolumn{3}{c}{Error rate (real)} & \multicolumn{3}{c}{Error rate (hypo.)} & \multicolumn{2}{c}{Infeasible (real)} \\
		\cmidrule{4-6} \cmidrule{7-9} \cmidrule{10-11}
		& [s] & & Queries & Journeys & Detour & Queries & Journeys & Detour & Queries & Journeys \\
		\midrule
		\multirow{6}{*}{London} & 0 & $\circ$ & 7.54\% & 2.30\% & 6.21\% & 7.54\% & 2.30\% & 6.21\% & 0.00\% & 0.00\% \\
		& 0 & $\bullet$ & 0.88\% & 0.26\% & 6.53\% & 0.32\% & 0.09\% & 6.80\% & 0.00\% & 0.00\% \\
		& 120 & $\circ$ & 4.00\% & 1.18\% & 5.81\% & 3.98\% & 1.18\% & 5.85\% & 0.02\% & 0.00\% \\
		& 120 & $\bullet$ & 0.50\% & 0.16\% & 6.45\% & 0.14\% & 0.05\% & 5.94\% & 0.14\% & 0.04\% \\
		& 180 & $\circ$ & 3.03\% & 0.85\% & 5.90\% & 3.01\% & 0.85\% & 5.93\% & 0.02\% & 0.00\% \\
		& 180 & $\bullet$ & 0.65\% & 0.20\% & 7.43\% & 0.07\% & 0.02\% & 5.95\% & 0.22\% & 0.06\% \\[3pt]
		\multirow{4}{*}{Germany} & 0 & $\circ$ & 5.10\% & 1.51\% & 8.13\% & 5.10\% & 1.51\% & 8.13\% & 0.00\% & 0.00\% \\
		& 0 & $\bullet$ & 3.11\% & 0.90\% & 7.56\% & 2.28\% & 0.64\% & 7.27\% & 0.09\% & 0.03\% \\
		& 180 & $\circ$ & 2.31\% & 0.67\% & 8.90\% & 2.25\% & 0.65\% & 8.87\% & 0.04\% & 0.01\% \\
		& 180 & $\bullet$ & 1.20\% & 0.37\% & 11.31\% & 0.43\% & 0.13\% & 8.50\% & 0.25\% & 0.08\% \\[3pt]
		\multirow{4}{*}{Stuttgart} & 0 & $\circ$ & 0.27\% & 0.09\% & 10.40\% & 0.27\% & 0.09\% & 10.40\% & 0.00\% & 0.00\% \\
		& 0 & $\bullet$ & 0.15\% & 0.05\% & 8.33\% & 0.13\% & 0.04\% & 7.62\% & 0.00\% & 0.00\% \\
		& 300 & $\circ$ & 0.09\% & 0.03\% & 13.85\% & 0.09\% & 0.03\% & 13.85\% & 0.00\% & 0.00\% \\
		& 300 & $\bullet$ & 0.04\% & 0.01\% & 14.66\% & 0.01\% & 0.00\% & 10.64\% & 0.00\% & 0.00\% \\[3pt]
		\multirow{4}{*}{Switzerland} & 0 & $\circ$ & 0.85\% & 0.22\% & 10.31\% & 0.85\% & 0.22\% & 10.31\% & 0.00\% & 0.00\% \\
		& 0 & $\bullet$ & 0.35\% & 0.10\% & 10.81\% & 0.27\% & 0.08\% & 10.45\% & 0.00\% & 0.00\% \\
		& 300 & $\circ$ & 0.34\% & 0.09\% & 10.70\% & 0.33\% & 0.08\% & 10.65\% & 0.00\% & 0.00\% \\
		& 300 & $\bullet$ & 0.15\% & 0.05\% & 11.24\% & 0.05\% & 0.01\% & 8.66\% & 0.02\% & 0.01\% \\
		\bottomrule
	\end{tabular*}
\end{table*}
Table~\ref{tbl:update-full} shows statistics for the update phase.
We observe that the number of replacement shortcuts decreases as the delay limit increases, even though there are more infeasible shortcuts that need to be replaced.
This indicates that most replacements are already included in the precomputed shortcut set.
Furthermore, the number of replacement shortcuts is negligible compared to the number of filtered Delay-ULTRA shortcuts.
Hence, even when running the update phase for an entire day, the replacement shortcuts will not significantly impact the query time.
Even in the most expensive configuration, the update phase takes at most a few seconds on the three smaller networks.
This is different for Germany due to the sheer size of the network.
Here, the time required for the replacement search starts to become a significant factor.

\subparagraph*{Error Rate.}
Table~\ref{tbl:coverage-full} reports error rates for all networks.
ULTRA-TB with basic updates already answers~50--75\% of delay-affected queries correctly.
Delay-ULTRA with advanced updates reduces the error rate by a factor of 4--15, while a hypothetical instant update phase would reduce it by more than a factor of~10 on all networks.
The median detours of suboptimal journeys compared to their optimal counterparts are not significantly influenced by the delay limit or the type of update phase; they are between~5 and~15\% depending on the network.
The replacement search is particularly effective for London, where it resolves almost~90\% of incorrectly answered queries on its own.
This effect narrows the gap in the error rate compared to the other networks, which have much fewer delay-affected queries to begin with.
Still, London and Germany exhibit higher error rates than the other networks.
We observe a tradeoff between the overall error rate and the number of infeasible journeys:
A higher delay limit reduces the error rate overall, but updates take longer to incorporate, which increases the risk of returning journeys that are already known to be infeasible at execution time.\looseness=-1
	
\subparagraph*{Query Performance.}
\begin{table*}[t]
    \center
    \caption{Running times of MR and (Delay-)ULTRA-TB for 10\,000 random queries. Also reported are the mean error rates of Delay-ULTRA-TB for these queries, with basic and advanced updates.}%
    \label{tbl:performance-full}%
    \begin{tabular*}{\textwidth}{@{\,}l@{\extracolsep{\fill}}r@{\extracolsep{\fill}}r@{\extracolsep{\fill}}r@{\extracolsep{\fill}}r@{\extracolsep{\fill}}r@{\extracolsep{\fill}}r@{\extracolsep{\fill}}r@{\,}}
        \toprule
        \multirow{2}{*}{Network} & $\maxDelay$ & \multicolumn{2}{c}{Query times [ms]} & \multicolumn{2}{c}{Error rate (basic)} & \multicolumn{2}{c}{Error rate (advanced)} \\
        \cmidrule{3-4} \cmidrule{5-6} \cmidrule{7-8}
        & [s] & MR & TB & Queries & Journeys & Queries & Journeys \\
        \midrule
        \multirow{3}{*}{London} & 0 & 18.4 & 4.0 & 1.94\% & 0.56\% & 0.28\% & 0.06\% \\
        & 120 & 18.4 & 8.0 & 1.21\% & 0.35\% & 0.11\% & 0.02\%\\
        & 180 & 18.4 & 9.8 & 0.71\% & 0.21\% & 0.05\% & 0.01\% \\[3pt]
        \multirow{2}{*}{Germany} & 0 & 920.4 & 78.8 & 3.57\% & 1.02\% & 2.14\% & 0.60\% \\
        & 180 & 920.4 & 113.2 & 1.34\% & 0.39\% & 0.57\% & 0.16\% \\[3pt]
        \multirow{2}{*}{Stuttgart} & 0 & 36.8 & 3.8 & 0.08\% & 0.02\% & 0.02\% & 0.00\% \\
        & 300 & 36.8 & 6.1 & 0.02\% & 0.00\% & 0.00\% & 0.00\% \\[3pt]
        \multirow{2}{*}{Switzerland} & 0 & 36.0 & 4.8 & 0.37\% & 0.11\% & 0.14\% & 0.03\% \\
        & 300 & 36.0 & 8.2 & 0.08\% & 0.02\% & 0.01\% & 0.00\% \\
        \bottomrule
    \end{tabular*}
\end{table*}
Finally, Table~\ref{tbl:performance-full} reports query times.
The speedup of Delay-ULTRA-TB over MR ranges from~2.3 for London to~8.0 for Germany.
Compared to ULTRA-TB without delay information, this corresponds to a slowdown of~1.4--2.0.
For this different set of queries, whose execution time is up to an hour before the departure time, the error rates are significantly lower.
The effect is particularly drastic for London, where trips and journeys are shorter on average and the impact of an individual delay therefore dissipates more quickly.
As a result, the error rate for~$\maxDelay=3$\,min is lower here than for~$\maxDelay=2$\,min, even though the opposite was the case in Table~\ref{tbl:coverage-full}.
This shows that incorporating a delay update too late mostly affects queries that depart right away, whereas the replacement shortcuts provide a more long-term benefit.
Still, $\maxDelay=2$\,min is the preferable configuration overall since the query time is lower.
With a moderate delay limit and advanced updates, the journey error rate is minuscule on all networks except Germany, where it is still very low at 0.16\%.
Overall, this shows that Delay-ULTRA offers a significant speedup over MR at near-perfect solution quality.

\section{Conclusion}
\label{sec:conclusion}
We adapted ULTRA to retain its performance benefits in the presence of delays.
Small delays are handled in the preprocessing phase and larger ones heuristically in the update phase.
Together, both steps reduce the error rate for random queries by up to a factor of~30 compared to an algorithm without any delay information.
Hence, Delay-ULTRA offers near-perfect solution quality while achieving a speedup of~2--8 over the fastest delay-robust competitor.
In the future, we would like to apply the replacement search to trip cancellations as well.
Further insights could also be gained from evaluating our algorithms on real or more sophisticated synthetic delay scenarios.

\bibliography{references}
\end{document}